\documentclass[11pt,a4paper]{article}
\usepackage{a4wide}
\usepackage{authblk}
\usepackage[utf8]{inputenc}
\usepackage{amsfonts,amsmath,url,xspace}
\usepackage{amsthm}
\usepackage{enumerate}
\usepackage[scr=euler]{mathalfa}
\usepackage{hyperref}
\usepackage[bibliography=common]{apxproof}
\usepackage{tikz-cd}
\usepackage{amssymb}
\usepackage{mathtools}

\newtheoremrep{theorem}{Theorem}[section]
\newtheoremrep{proposition}[theorem]{Proposition}
\newtheoremrep{lemma}[theorem]{Lemma}
\newtheorem{corollary}[theorem]{Corollary}
\theoremstyle{definition}
\newtheorem{example}[theorem]{Example}
\newtheorem{definition}[theorem]{Definition}
\newtheorem{remark}[theorem]{Remark}

\usetikzlibrary{calc}
\usetikzlibrary{decorations.pathmorphing}
\tikzset{curve/.style={settings={#1},to path={(\tikztostart)
    .. controls ($(\tikztostart)!\pv{pos}!(\tikztotarget)!\pv{height}!270:(\tikztotarget)$)
    and ($(\tikztostart)!1-\pv{pos}!(\tikztotarget)!\pv{height}!270:(\tikztotarget)$)
    .. (\tikztotarget)\tikztonodes}},
    settings/.code={\tikzset{quiver/.cd,#1}
        \def\pv##1{\pgfkeysvalueof{/tikz/quiver/##1}}},
    quiver/.cd,pos/.initial=0.35,height/.initial=0}
\tikzset{tail reversed/.code={\pgfsetarrowsstart{tikzcd to}}}
\tikzset{2tail/.code={\pgfsetarrowsstart{Implies[reversed]}}}
\tikzset{2tail reversed/.code={\pgfsetarrowsstart{Implies}}}
\tikzset{no body/.style={/tikz/dash pattern=on 0 off 1mm}}

\makeatletter
\newcommand{\subalign}[1]{%
  \vcenter{%
    \Let@ \restore@math@cr \default@tag
    \baselineskip\fontdimen10 \scriptfont\tw@
    \advance\baselineskip\fontdimen12 \scriptfont\tw@
    \lineskip\thr@@\fontdimen8 \scriptfont\thr@@
    \lineskiplimit\lineskip
    \ialign{\hfil$\m@th\scriptstyle##$&$\m@th\scriptstyle{}##$\hfil\crcr
      #1\crcr
    }%
  }%
}
\makeatother

\newcommand{\emphat}{\emph}
\newcommand{\define}{\textbf}

\newcommand{\cat}{\mathscr}
\newcommand{\ncat}{\mathsf}
\newcommand{\nat}{\Rightarrow}
\newcommand{\Hom}{\textup{Hom}}
\newcommand{\ob}{\textup{ob}}
\newcommand{\mor}{\textup{mor}}
\newcommand{\id}{\textup{id}_}
\newcommand{\colim}{\textup{colim}}
\newcommand{\emb}{\rightarrowtail}

\newcommand{\into}{\hookrightarrow}
\newcommand{\epi}{\twoheadrightarrow}
\newcommand{\Q}{\mathscr{Q}}
\newcommand{\M}{\mathscr{M}}
\renewcommand{\epsilon}{\varepsilon}

\renewcommand{\phi}{\varphi}
\newcommand{\+}[1]{\mathcal{#1}}
\newcommand{\bb}[1]{\mathbb{#1}}
\renewcommand{\tilde}{\widetilde}
\newcommand{\tup}[1]{\langle #1 \rangle}
\newcommand{\dow}{{\downarrow}}
\newcommand{\acc}{E}
\newcommand{\subs}[1]{[R_=/#1]}

\newcommand{\debt}{{\rm debt}}
\newcommand{\arity}{{\rm arity}}
\newcommand{\tr}{{\rm tr}}
\newcommand{\midd}{\mathrel{\,\mid\,}}
\newcommand{\sigmadgl}{\sigma_{\rm DGL}}
\newcommand{\sigmacdxp}{\sigma_{\rm CDXP}}
\newcommand{\bsigma}{\overline{\sigma}}

\newcommand{\st}{{\rm ST}} 
\newcommand{\dom}{{\rm dom}}

\newcommand{\union}{\cup}

\newcommand{\last}{{\rm last}}
\newcommand{\logicfullname}{Path Predicate Modal Logic\xspace}
\newcommand{\logicname}{PPML}
\newcommand{\logic}{\ensuremath{{\rm \logicname}}\xspace}
\newcommand{\logick}{\ensuremath{{\rm \logicname}_k}\xspace}
\newcommand{\logicpos}{\ensuremath{{\rm \logicname}^{+}\xspace}}
\newcommand{\logicposk}{\ensuremath{{\rm \logicname}_k^{+}\xspace}}
\newcommand{\datagl}{\ensuremath{{\rm DataGL}}\xspace}
\newcommand{\cdxp}{\ensuremath{{\rm CoreDataXPath}}\xspace}
\newcommand{\bml}{\ensuremath{{\rm BML}}\xspace}

\newcommand{\bisim}{\mathrel{\underline{~~}\hspace{-0.85em}{\leftrightarrow}}}
\newcommand{\simu}{\mathrel{\underline{~~}\hspace{-0.75em}{\rightarrow}}}

\newcommand{\Ck}{\mathbb{C}_k}
\newcommand{\Mk}{\mathbb{M}_k}
\newcommand{\Ek}{\mathbb{E}_k}
\newcommand{\Pn}{\mathbb{P}_n}
\newcommand{\Pkn}{\bb P_{k, n}}
\newcommand{\Pkns}{\bb P_{k+1, N}^*}
\newcommand{\A}{\+ A}
\newcommand{\bpa}{a}
\newcommand{\Ap}{{(\A, \bpa)}}
\newcommand{\B}{\+ B}
\newcommand{\bpb}{b}
\newcommand{\Bp}{{(\B, \bpb)}}
\newcommand{\C}{\+ C}
\newcommand{\bpc}{c}
\newcommand{\Cp}{{(\C, \bpc)}}
\newcommand{\T}{\+ T}
\newcommand{\bpt}{u}
\newcommand{\Tp}{{(\T, \bpt)}}

\newcommand{\G}{\+ G}

\newcommand{\bpp}{p_0}
\newcommand{\Pp}{{(\+ P, \bpp)}}

\newcommand{\N}[0]{\mathbb{N}}

\newcommand{\MM}{\mathfrak{M}}
\newcommand{\tMM}{{t \, \MM}}
\newcommand{\MMp}{{(\mathfrak{M}, w)}}
\newcommand{\tMMp}{{t \, \MMp}}
\newcommand{\MMpp}{{(\mathfrak{M}', w')}}
\newcommand{\tMMpp}{{t \, \MMpp}}
\newcommand{\Struct}[1][\sigma]{\ensuremath{{\sf Struct}(#1)}\xspace}

\newcommand{\Structpointed}[1][\sigma]{\ensuremath{{\sf Struct}_*(#1)}\xspace}
\newcommand{\Structpointedfin}[1][\sigma]{\ensuremath{{\sf Struct}_*^{\text{f}}(#1)}\xspace}
\newcommand{\St}{\Struct}
\newcommand{\Stp}{\Structpointed}

\newcommand{\Stpgl}{\Structpointed[\sigmadgl]}

\newcommand{\Stpf}{\Structpointedfin}
\newcommand{\EMf}{\EM^\text{f}}
\newcommand{\Spf}{{\sf S}_*^\text{f}}
\newcommand{\ModDGL}{\ensuremath{{\sf ModDGL}}\xspace}
\newcommand{\EM}{\text{\sf{EM}}}
\newcommand{\Kl}{\text{\sf{Kl}}}

\newcommand{\kbisimproblemsig}[1][\sigma]{\ensuremath{k\textup{\textsc{-Bisim}}(#1)}\xspace}
\newcommand{\modelcheckproblemsig}[1][\sigma]{\ensuremath{\textup{\textsc{ModelCheck}}(#1)}\xspace}
\newcommand{\satproblem}{\ensuremath{\textup{\textsc{PPML-Sat}}}\xspace}

\newcommand{\satproblembml}{\ensuremath{\textup{\textsc{BML-Sat}}}\xspace}
\newcommand{\pspace}{\textup{\textsc{PSpace}}\xspace}
\newcommand{\ptime}{{\textsc{PTime}}\xspace}

\begin{document}

\title{Modal Logic with Relations over Paths: a Theoretical Development through Comonadic Semantics}

\author[1, 3]{Santiago Figueira}
\author[2, 3]{Gabriel Goren-Roig}
\affil[1]{Universidad de Buenos Aires. Facultad de Ciencias Exactas y Naturales. Departamento de Computación. Buenos Aires, Argentina.}
\affil[2]
{Universidad de Buenos Aires. Facultad de Ciencias Exactas y Naturales. Departamento de Matemática. Buenos Aires, Argentina.}
\affil[3]{CONICET - Universidad de Buenos Aires. Instituto de Ciencias de la Computación (ICC). Buenos Aires, Argentina.}
\date{}
\setcounter{Maxaffil}{0}
\renewcommand\Affilfont{\itshape\small}
\maketitle

\begin{abstract}
Game comonads provide categorical semantics for comparison games in Finite Model Theory, thus providing an abstract characterisation of logical equivalence for a wide range of logics, each one captured through a specific choice of comonad.
Motivated by the goal of applying comonadic tools to the study of data-aware logics such as \cdxp, in this work we introduce
a generalisation of Modal Logic that allows relation symbols of arbitrary arity as atoms of the syntax, which we call \logicfullname or \logic.
We motivate this logic as arising from a shift in perspective on a previously studied fragment of \cdxp, called \datagl, and prove that \logic recovers \datagl for a specific choice of signature. We argue that this shift in perspective allows the capturing and designing of new data-aware logics. On the other hand, \logic enjoys an intrinsic motivation in that it extends Modal Logic to predicate over more general models.
Having introduced resource-bounded simulation and bisimulation games for \logic together with a proof of the Hennessy-Milner property relating bisimilarity and logical equivalence, we define the \logic comonad, which essentially amounts to an unravelling construction on models of \logic, and prove that it captures these games, following analogous results in the literature. However, we depart from the literature in our proof strategy, since we draw upon the axiomatic framework of arboreal categories, giving intuition for the axioms involved and supplying detailed verifications. Subsequently, we develop the model-theoretical understanding of \logic by making systematic use of the comonadic framework. This includes results such as a tree-model property and an alternative proof of the one-way Hennessy-Milner property using a correspondence between positive \logic formulas and canonical models. We also use the comonadic perspective to establish connections with other logics, such as bounded quantifier rank and bounded variable number fragments of First Order Logic on one side and Basic Modal Logic on the other, and show how the \logic comonad induces a syntax-free characterisation of logical equivalence for \datagl, our original motivation. With respect to Basic Modal Logic, a functorial assignment from \logic unravellings into Kripke trees enables us to obtain polynomial-time reductions from \logic problems to their Basic Modal Logic counterparts.
\end{abstract}

\section{Introduction}

\noindent One of the main problems of database theory is finding appropriate balances between complexity and expressivity of query languages. In this context,
data-aware logics are languages that reason on data graphs, i.e.\ finite graphs whose nodes are decorated with a label from  a finite alphabet and a data value from an infinite domain.
Formulas in data-aware logics express queries based not only on the graph topology and node labels, as modal logics do, but also with reference to data values. However, instead of directly accessing data values as constants, data-aware logics only allow comparison of data values in a controlled way through specific syntactic constructions~\cite{CoreDataXPath,benedikt2009xpath,abriola2018bisimulations}. One important comparison operation consists in checking for equality of data values, which is sufficient to express the data join, arguably the most important construct of a query language.

On the other hand, comonadic semantics~\cite{abramsky2017pebbling, abramsky2021relating,abramsky2021comonadic,abramsky2021arboreal,abramsky2022hybrid,dawar2021lovasz,jakl2023categorical,abramsky2022preservation} is a novel framework for Finite Model Theory in which a categorical language and methodology is adopted. The cornerstone of this theory consists in the observation that model comparison games for different logics can be expressed through corresponding \emphat{game comonads}. These comonads are indexed by a resource parameter that controls some notion of complexity of the formulas in the associated language. Given that monads and comonads feature prominently in formal semantics and functional programming,%
\footnote{More generally, monads and comonads are core concepts of Category Theory, arising from any adjunction between categories.}
this is particularly interesting from the perspective of unifying the two main strands of Theoretical Computer Science, which have been called `structure' (semantics and compositionality) and `power' (expressivity and complexity) in~\cite{abramsky2021relating} and subsequent works. Tending bridges between these two communities and their methods will hopefully provide new insights into the discipline of Theoretical Computer Science as a whole.

The following work constitutes a first application of comonadic semantics to the study of data-aware logics.
We define a family of logics which extends the syntax of Basic Modal Logic to relational signatures with symbols of arbitrary arity, and we study its model theory through the tools of comonadic semantics.
We call this family of logics \logicfullname or \logic for short. \logic seems well suited to express data-aware logics; in particular, it provides a framework for the language \datagl, studied in~\cite{dataGL} from a proof-theoretical point of view. It is also of independent interest as a modal logic which reasons over a more general class of models. In this sense, it gives a formal answer to the question of what it means to reason modally about arbitrary relational structures—at least once one has adopted a binary relation symbol as an accessibility relation. The corresponding \logic comonad occupies a middle ground between the Modal comonad on one side and the Ehrenfeucht-Fraïssé and Pebbling comonads on the other ~\cite{abramsky2021relating}, and shares a fundamental technical property with the former, namely idempotence, which helps us establish tight connections between \logic and Basic Modal Logic.

\paragraph*{Outline.}
The paper is structured as follows: after having fixed terminology and notation in Section~\ref{sec:prelim}, in Section~\ref{sec:ppml} we introduce the main objects of study—\logic on one side, and the \logic comonad on the other—and establish their fundamental interrelationship. Then in Sections~\ref{sec:model_th_of_ppml} and~\ref{sec:ppml_and_other_logics} we take advantage of the comonadic formalism in order to establish model-theoretic results about \logic and correspondences with other languages, namely First Order Logic, \datagl and Basic Modal Logic. We close with a discussion of conclusions and future lines of work in Section~\ref{sec:conclu}. For brevity, some of the proofs have been deferred to the Appendix.

\paragraph*{Contributions.}
After introducing the syntax and semantics of \logic (Defs.~\ref{def:syntax_ppml} and~\ref{def:semantics_ppml}), in Section~\ref{sec:ppml} we begin by defining appropriate notions of resource-bounded simulation and bisimulation (Def.~\ref{def:bisim}), together with their formulation as Spoiler-Duplicator games (Def.~\ref{def:ppml_games}), and proving a Hennessy-Milner property linking $k$-bisimilarity with logical indistinguishability by formulas of \logick, the fragment of \logic with modal depth bounded by $k$ (Thm.~\ref{thm:hennessy-milner}). On the other hand, we define a $k$-indexed family of comonads on the category of pointed relational structures, which we denote by $\{\Ck\}_{k\in\N}$, and identify the result of applying the comonad $\Ck$ to a given structure as constructing a $k$-step unravelling of the structure, turning it into a tree-shaped structure which we call path-predicate tree or pp-tree. Bisimulation serves as the point of contact between the logic and the comonad, since, as we show in the remaining of Section~\ref{sec:ppml}, the fundamental categorical constructions accompanying the \logic comonad (the Kleisli and Eilenberg-Moore categories associated with $\Ck$) capture $k$-similarity (Prop.~\ref{prop:kleisli_morphisms}, which holds almost by definition of $\Ck$) and $k$-bisimilarity (Thm.~\ref{thm:full_logical_equivalence}, which requires additional conceptual scaffolding) in terms of the existence of certain kinds of homomorphisms involving the unravelling construction.
The results of Section~\ref{sec:ppml} are analogous to previous results for Basic Modal Logic, both from the side of logic~\cite{blackburn2001modal} and from the side of comonads~\cite{abramsky2021relating}. However, we depart from the existing literature in our exposition leading to the proof of Theorem~\ref{thm:full_logical_equivalence} since we emphasise the axiomatic framework of arboreal categories~\cite{abramsky2021arboreal}, giving intuition for the axioms involved and supplying detailed verifications for their validity in our particular case. In this way, our proof of Theorem~\ref{thm:full_logical_equivalence} draws on abstract results from~\cite{abramsky2021arboreal} proven for arbitrary arboreal covers. The strategy is as follows: we first prove that the family $\{\Ck\}_k$ induces a resource-indexed arboreal cover of the category of pointed relational structures (Thm.~\ref{thm:arborealcover}). Thus, for each $k$ there exists an abstract Spoiler-Duplicator game played between objects of the Eilenberg-Moore category of $\Ck$, and winning strategies for Duplicator in this game correspond to spans of \emph{open pathwise embeddings}, an abstract notion that generalises functional bisimulation. We prove that the abstract game is equivalent to our definition of the \logic $k$-bisimulation game (Prop.~\ref{prop:games_coincide}) and that open pathwise embeddings in our case coincide with an appropriate definition of bounded morphism for \logic (Prop.~\ref{prop:open_PE}). Then by \cite[Prop. 46]{abramsky2021arboreal} we conclude that spans of bounded morphisms characterise $k$-bisimilarity for \logic.
We also note that previous, specific instances of game comonads have targeted well-known logics, while, on the contrary, here we undertake the initial characterisation of a new logic using the comonadic framework.

In Section~\ref{sec:model_th_of_ppml} we explore three additional topics in the model theory of \logic, which showcase the utility of the comonadic framework. First we show that isomorphism of resource-bounded unravellings coincides with resource-bounded bisimilarity for an extension of \logic with graded modalities which we denote by $\logic^\#$ (Thm.~\ref{thm:graded_bisimilarity_by_iso_classes}), which immediately implies a homomorphism-counting property using results from~\cite{dawar2021lovasz} (Thm.~\ref{thm:lovasz_for_ppml}). Then we conclude from Thm.~\ref{thm:graded_bisimilarity_by_iso_classes} a pp-tree-model property for $\logic^\#$, in particular for \logic. This is immediate from the fact that $\Ck$ is an idempotent comonad. Thirdly, we prove a Chandra-Merlin-like correspondence~\cite{chandra1977optimal} between formulas in the negation-free fragment of \logic and finite pp-trees (Corollaries~\ref{coro:chandra_merlin_for_ppml_1} and~\ref{coro:chandra_merlin_for_ppml_2}) which enables an alternative, comonadic proof of the Hennessy-Milner property for $k$-simulations (Thm.~\ref{thm:hm_with_chandra-merlin}). This suggests that Chandra-Merlin-like correspondences may serve as an independent point of contact between logic and comonads. Except for this last observation, the results in this section follow the footsteps of previously established results on game comonads; however, the induction arguments in the proofs of Thm.~\ref{thm:graded_bisimilarity_by_iso_classes} and of the Chandra-Merlin-like correspondence require non-trivial adaptations from Basic Modal Logic and involve novel constructions.

In Section~\ref{sec:ppml_and_other_logics} we establish relationships between \logic and First Order Logic, \datagl and Basic Modal Logic. Beginning with First Order Logic, we show that \logic translates into First Order Logic with bounded quantifier rank and, whenever the arity of atomic symbols in \logic is bounded, with bounded variable number (Prop.~\ref{prop:PPML_FOL_tranlation}). We also show how this is mirrored by the relationship between the corresponding comonads: $\Ck$ turns out to be a subcomonad of the Ehrenfeucht-Fraïssé and Pebbling comonads~\cite{abramsky2021relating} in a suitable sense (Prop.~\ref{prop:subcomonads}). We then return to our motivation of developing the theory of data-aware logics by considering \datagl as our starting point. We show that models of \datagl can be embedded as a subclass of models of \logic with a particular choice of relation symbols by encapsulating the actual data values into their corresponding `equal-data' relation. In this way, both logics are equi-expressive over this class (Thm.~\ref{thm:datagl-equiexpressive}). In this sense \logic contains \datagl and this allows us to capture logical indistinguishability by \datagl formulas in terms of morphisms involving \logic unravellings (Thm.~\ref{thm:full_logical_equivalence_datagl}). We also show how seeing \datagl as contained in \logic lets us define other data-aware logics by small modifications. Finally, we study the close relationship between \logic and Basic Modal Logic. Here lie the algorithmic contributions of this paper. We define a fully-faithful functor $K$ from pp-trees to Kripke trees which preserves and reflects $k$-bisimilarity for all $k$ (Thm.~\ref{thm:K_preserves_bisimilarity}) and we use it to establish polynomial-time computational reductions from the problems of checking $k$-bisimilarity, model checking and satisfiability for \logic to their Basic Modal Logic counterparts. In doing so, we prove the finite-model property for \logic as a corollary, and we observe that the functor $K$ also establishes an injective function from \logic logical types into modal logical types.

\paragraph*{Note on the category-theoretical background.}
In the spirit of bridging the gap between structure and power, we have strived to give an informative exposition of the necessary categorical concepts. This means that for most of this paper we only assume some familiarity with categories, functors and natural transformations. We hope that this will be helpful to readers interested in learning how to manipulate these concepts. Some additional categorical concepts are used in Section~\ref{sec:arboreal}, namely limits, colimits and adjunctions. Adjunctions also make an appearance in the proof of Theorem~\ref{thm:lovasz_for_ppml}.

\subsection{Preliminaries}\label{sec:prelim}

\paragraph*{Sequences.} For a set $\Sigma$, let $\Sigma^i$ be the set of all finite sequences of length $i$ over $\Sigma$, let $\Sigma^* \coloneqq \bigcup_{0 \leq i}\Sigma^i$, $\Sigma^+ \coloneqq \bigcup_{1 \leq i}\Sigma^i$, and $\Sigma^{\leq n} \coloneqq \bigcup_{1 \leq i\leq n}\Sigma^i$. For $s\in\Sigma^*$, let $|s|$ be the length of $s$ and let $s(i)$ be the $i$-th element of $s$ from left to right, so that $s=s(1)\dots s(|s|)$. Let $\last_k(s) \coloneqq s(|s|-k)\dots s(|s|)$ if $k\leq|s|$, and $\last_k(s) \coloneqq s$ otherwise. In accordance with the notation to be introduced in Section~\ref{sec:comonad}, we will denote the last element of $s$ by $\epsilon(s)$, if $|s|>0$. The concatenation of an element $a\in \Sigma$ and a sequence $s\in \Sigma^*$ is denoted by $a.s$. Although tuples are represented with parentheses and sequences with square brackets, we do not distinguish between them formally.

\paragraph*{Relational structures.} A \define{relational first-order signature}, or \define{signature} for short, consists of a set $\sigma$, elements of which are called \define{relation symbols}, and a function $\arity: \sigma \to \N_{>0}$ assigning a positive integer to each symbol, which is referred to as its \define{arity}. We refer to a signature $(\sigma, \arity)$ by the symbol $\sigma$. A \define{$\sigma$-structure} $\A$ consists of a set $|\A|$, which we refer to as its \define{universe} or \define{domain}, together with a subset $R^{\A} \subseteq |\A|^{\arity(R)}$ for each $R \in \sigma$, which we refer to as the \define{interpretation of $R$ in $\A$}. Since we use $|{-}|$ for the underlying set of a structure, if $S$ is a set we use the notation $\# S$ to denote its cardinality.

\paragraph*{Morphisms of relational structures.} We denote by \Struct the category whose objects are $\sigma$-structures and whose morphisms are the homomorphisms between them, that is, the interpretation-preserving functions between the underlying domains. \Structpointed denotes the category whose objects are \emphat{pointed} $\sigma$-structures, that is $\sigma$-structures $\A$ equipped with a distinguished element or \define{basepoint} $a \in |\A|$. The morphisms in this case are the \define{pointed} homomorphisms, i.e.\ those that preserve the basepoints. A homomorphism $f: \A \to \B$ is \define{strong} iff for all $R \in \sigma$ and for all $s = (a_1,\dots,a_r) \in |\A|^+$, $s \in R^\A \iff f(s) \in R^\B$, where $f(s) \coloneqq (f(a_1),\dots,f(a_r))$.
That is to say, strong homomorphisms reflect relations as well as preserve them. An \define{embedding of relational structures}, or \define{relational embedding}, is an injective strong homomorphism. When talking about homomorphisms between pointed structures, all homomorphisms will be assumed to be pointed unless stated otherwise. We say $\Ap$ is an \define{embedded substructure} of $\Bp$ iff $|\A| \subseteq |\B|$, $a=b$ and the inclusion $|\A| \hookrightarrow |\B|$ is a relational embedding.

\paragraph*{Chains and trees.} Let $V$ be a set and $R \subseteq V \times V$ a binary relation, to be treated as an accessibility relation. If $(v,v')\in R$, we say that $v'$ is a \define{successor} of $v$ and that $v$ is a \define{predecessor} of $v'$, and write $v \prec v'$. We denote by $R^+$ the transitive closure of $R$, and if $(v, v') \in R^+$ we say that $v'$ is \define{accessible from $v$}. We write $R(v)$ for the set of successors of $v$. We say that a sequence $s = [v_1,\dots,v_\ell] \in V^+$ is a \define{chain} if $v_i \prec v_{i+1} \forall i \in \{0,\dots,\ell-1\}$. The \define{length} of a chain is its length as a sequence.

We say that $(V, R, v)$, where $v \in V$, is a \define{rooted tree}, or \define{tree} for short, if $v$ (called the \define{root} of the tree) has no predecessors, and moreover all non-root points are accessible from $v$ and have a unique predecessor. The \define{height} of a point $v' \in V$ is defined as the unique $n \in \N$ such that $(v, v') \in R^n$, where $R^n$ is the $n$-fold composition of $R$ with itself (with the convention that $R^0$ is the identity relation). We say $(V, R, v)$ is of \define{finite height} if there exists a maximum height over all its points, which we refer to as the height of the tree. A point in a tree without successors is called a \define{leaf}. A maximal chain in a tree is called a \define{branch}. A tree is \define{finitely branching} if each point has a finite number of successors.

\paragraph*{Chains and trees in structures.} Throughout this paper we will assume that, unless stated otherwise, relational signatures contain a distinguished binary symbol $\acc$, which plays the role of an accessibility relation.
Given a relational structure $\A$, we use the notation and vocabulary of the preceding paragraph with $V = |\A|, R = \acc^\A$. To be more explicit, we may refer to a chain in $(|\A|, \acc^\A)$ as an \define{$\acc$-chain}. Moreover, we say that a pointed structure $\Ap$ is an \define{$\acc$-tree} if $(|\A|, \acc^\A, a)$ is a tree. We also say that $\Ap$ is an $\acc$-chain if it is an $\acc$-tree with a single branch, and that $\A$ is finitely branching if it is finitely branching as an $\acc$-tree.

\paragraph*{Functors and categories.} Given a category $\cat{A}$, we write $\ob(\cat{A})$ for its collection of objects and $\mor(\cat{A})$ for its collection of morphisms. Given $A, A' \in \cat{A}$, we write $A \cong A'$ iff $A$ and $A'$ are isomorphic (there exists an invertible morphism between them). We write $A \in \cat{A}$ to mean $A \in \ob(\cat{A})$ and given $f: A \to A'$ in $\cat{A}$, we denote by $\dom(f)$ its domain $A$. Let $F: \cat{A} \to \cat{B}$ be a functor between categories $\cat{A}$ and $\cat{B}$. We say that $F$ is \define{full} if the functions that define its action on morphisms are all surjective, i.e.\ if for every $A, A' \in \cat{A}$ and for every morphism $g: FA \to FA'$, there exists a morphism $f: A \to A'$ such that $g = Ff$. We say that $F$ is \define{faithful} if such actions are injective, i.e.\ if whenever $Ff = Ff'$ for a pair of morphisms $f, f': A \to A'$, it must be the case that $f = f'$. We say $F$ is \define{fully faithful} if it is full and faithful, in which case it defines bijections between the homsets (sets of morphisms) $\Hom_\cat{A}(A, A')$ and $\Hom_\cat{B}(FA, FA')$ for all $A, A' \in \cat{A}$. We use the notations $\Hom_\cat{A}(A,A')$, $\Hom(A, A')$ and $\cat{A}(A, A')$ interchangeably. The image of a fully faithful functor is a subcategory of its codomain, and in particular is a \define{full subcategory}, which means that it contains all morphisms between the objects it contains. We say $F$ is \define{essentially surjective on objects} if for all $B \in \cat{B}$ there exists some $A \in \A$ such that $FA \cong B$. Given a category $\cat{A}$, we denote the identity functor on $\cat{A}$ by $1_\A$. Throughout this paper, all categories may be safely assumed to be \define{locally small} and \define{well powered}, which means that the collections of morphisms between any two objects are sets, and that the collection of subobjects of any given object is also a set. More generally, we say that a collection is \define{small} if it constitutes a set. A \define{small category} is a category $\cat{A}$ such that $\mor(\cat{A})$ is small (and hence also $\ob(\cat{A})$ is).

\section{\logicfullname}\label{sec:ppml}

\begin{definition}\label{def:syntax_ppml}
    Let $\sigma$ be a first-order relational signature including a binary relation symbol $\acc$. The syntax of \define{\logicfullname (\logic) over $\sigma$} (or $\sigma$-\logic) is defined by the grammar
    \begin{align*}
    \varphi & \ \Coloneqq \top              \ \mid \ 
                    R                 \ \mid \ 
                    \lnot\varphi        \ \mid \ 
                    \varphi\land\varphi \ \mid \ 
                    \Diamond\varphi                 \tag{$R \in \bsigma$}
    \end{align*}
    where $\bsigma \coloneqq \sigma \setminus \{\acc\}$ (notice that $\acc$ is not an atom of the language). The \define{modal depth} of a formula $\varphi$ is defined as the maximum number of nested $\Diamond$ symbols in $\varphi$.
\end{definition}

\begin{remark}
    Throughout this paper we will assume that all relational signatures include a designated binary relation symbol $\acc$ unless stated otherwise.
\end{remark}

Analogously to Basic Modal Logic (which we will shorten to \bml), the truth value of a $\sigma$-\logic formula is defined relative to a $\sigma$-structure $\A$ and a specific point $a \in |\A|$. However, the evaluation of a \logic formula involves the construction of a path on the structure, to be interpreted as a history which must be remembered in order to continue the evaluation at any given point.\footnote{This interpretation is connected to memory logics~\cite{AFFM08b}, although this particular remembrance device is comparatively simple.} This lets us think of the language as manipulating paths at a propositional level, hence the name \logicfullname.

\begin{definition}\label{def:semantics_ppml}
    Given a signature $\sigma$, we define the semantics of \logic over a $\sigma$-structure $\A$ and a sequence or \define{valuation} $s \in |\A|^+$ as follows:
    \begin{align*}
    \A,s &\models \top            & &\text{always}\\
    \A,s &\models R              &\text{iff \quad} &\text{$\arity(R) \leq |s|$ and $\last_{\arity(R)}(s) \in R^{\A}$}\\
    \A,s &\models \lnot\varphi     &\text{iff \quad} &\A,s\not\models \varphi\\
    \A,s &\models \varphi\land\psi &\text{iff \quad} &\A,s\models \varphi \text{ and }\A,s\models \psi\\
    \A,s &\models \Diamond\varphi  &\text{iff \quad} &\exists a\in |\A|.(\epsilon(s),a)\in \acc^{\A} \text{ and } \A, s.a \models \varphi.
    \end{align*}
    We write $\A, a \models \phi$ for $\A, [a] \models \phi$, and say that a pointed structure $(\A, a) \in \Structpointed$ satisfies $\phi$ iff $\A, a \models \phi$.

    The \define{positive fragment} of \logic, \logicpos, consists of the subset of negation-free formulas. We denote by \logick and \logicposk~the fragments of \logic and \logicpos consisting of formulas of modal depth at most $k$. We write
    \begin{align*}
    \A,a \Rrightarrow^+_k \B,b \quad \text{ if }& \quad \A,a\models\varphi \implies \B,b\models\varphi \text{ for all $\varphi$ in \logicposk,} \\
    \A,a\equiv^+_k\B,b \quad \text{ if }& \quad \A,a\models\varphi \iff \B,b\models\varphi \text{ for all $\varphi$ in \logicposk,} \\ 
    \A,a\equiv_k\B,b \quad \text{ if }& \quad \A,a\models\varphi \iff \B,b\models\varphi \text{ for all $\varphi$ in \logick, and} \\
    \A,a\equiv\B,b \quad \text{ if }& \quad \A,a\models\varphi \iff \B,b\models\varphi \text{ for all $\varphi$ in \logic.}
    \end{align*}
\end{definition}

By default, we consider the semantics of \logic to be single-pointed, treating the more general valuation semantics mostly as a means to define the former.

\begin{example}\label{ex:ppml_first_examples}
Let $\sigma$ be a relational signature with $E \in \sigma$.
\begin{enumerate}
    \item If $\arity(R) = 1$ for all $R \in \bsigma$, we say that $\sigma$ is a \define{unimodal signature}, and we might suggestively write $\bsigma = \{p, q, r, \dots\}$. Then a $\sigma$-structure is a Kripke structure with propositional variables $\bsigma$ and accessibility relation $\acc$. Moreover the syntax of \logic coincides with the syntax of \bml, and the semantics of a formula seen as a \logic-formula and as a \bml-formula coincide. In this case.
    This example shows that \logic is an extension of Basic Modal Logic.
    
    \item Now let $\sigma \coloneqq \{E, p, q, r, \dots, S\}$ where $S$ is binary. An example $\sigma$-formula is $\phi = \Diamond (S \land \Diamond S)$. It is the case that a pointed structure $\Ap$ satisfies $\phi$ if and only if there exists an $E$-chain $[a, a', a'']$ starting at $a$ such that $\{(a, a'), (a', a'')\} \subseteq S^\A$.
    
    \item We now add a ternary relation $T$ to $\sigma$. Consider the formula $\psi_1 \coloneqq \Diamond (\lnot S \land \Diamond(\lnot S \land \lnot T))$. The extension of $\psi_1$ in a structure $\A$ consists of all points $a$ such that there exists an $E$-chain $[a, a', a'']$ starting at $a$ with $(a, a'), (a', a'') \not\in S^\A$ and $(a, a', a'') \not\in T^\A$.

    \item If $R$ is a relation symbol of arity greater than $1$, then $\phi = R$ is unsatisfiable. More generally, any instance of a relation symbol $R$ appearing in a formula $\phi$ not nested in at least $\arity(R) - 1$ diamond symbols can be rewritten to the falsum constant $\bot \coloneqq \lnot \top$. In this way, one can always rewrite a \logic formula to a formula in which all relation symbols are appropriately nested.
\end{enumerate}
\end{example}

The observation in the last example is formalised by the following definition.

\begin{definition}\label{def:modal_debt}
    Given $\phi \in \logic$, its \define{modal debt}, denoted by $\debt(\phi)$, is defined inductively as follows:
    \begin{align*}
        \debt(\top) &\coloneqq 0 \\
        \debt(R) &\coloneqq \arity(R) - 1\\
        \debt(\phi_1 \land \phi_2) &\coloneqq \max\{\debt(\phi_1), \debt(\phi_2)\}\\
        \debt(\lnot \psi) &\coloneqq \debt(\psi)\\
        \debt(\Diamond \psi) &\coloneqq \min\{0, \debt(\psi) - 1\}.
    \end{align*}
    We say that $\phi$ is \define{well nested} if $\debt(\phi) = 0$, and \define{badly nested} if otherwise.
\end{definition}
Intuitively, modal depth and modal debt play complementary roles: while the modal depth of a \logic formula $\phi$ quantifies how much one must explore the $\sigma$-structure from a given starting point in order to evaluate $\phi$, its modal debt quantifies how much one must have \emph{already} explored the structure before arriving at the current position, in order to evaluate $\phi$ in a sensible way.

\begin{remark}\label{rem:hints_at_relationship_with_modal}
    The fact that when $\sigma$ is unimodal the syntax of \logic becomes \emphat{exactly} that of \bml hints at the possibility of relating \logic and \bml by redeclaring all relation symbols in $\bsigma$ as unary. This is the approach we follow in Section~\ref{sec:relationship_modal} to obtain computational reductions from \logic problems to their \bml analogues. However, care must be taken when dealing with badly-nested formulas.
\end{remark}

\begin{remark}\label{rem:bounded_arity_signatures}
    Suppose that the arity of all symbols in $\sigma$ is bounded. This happens in important cases such when $\sigma$ is finite, when $\sigma$ is unimodal or when $\sigma$ is obtained from a unimodal signature by adding a finite collection of additional symbols.
    In such a situation, let $W$ be the maximum arity of relations in $\bsigma$, which we think of as a memory size. We can now replace the semantics given above in terms of valuations with equivalent semantics based on bounded-length valuations $s\in |\A|^{\leq W}$. The last clause in Definition~\ref{def:semantics_ppml} becomes
    $$A,s \models \Diamond\varphi  \text{\quad iff \quad} \exists a\in |\A|.(\epsilon(s),a)\in \acc^{\A} \text{ and } \A, (\last_{W-1}(s).a) \models \varphi.$$
    The resulting logic is equivalent, yet this makes explicit the fact that \logic effectively requires only a bounded amount of memory. We explore this in more detail in Section~\ref{sec:fragments_and_subcomonads}.
\end{remark}

\subsection{\boldmath{\logic} Bisimulation} \label{sec:bisim}

We now present natural notions of resource-bounded bisimulation and one-way simulation for \logic. These differ from their \bml counterparts in that checking the \logic analogue of atomic harmony requires remembering more than just the last visited node. We follow the stratified version of bisimulations, well studied in the literature (see e.g.\ \cite[Definition 2.30]{blackburn2001modal}).
\begin{definition}\label{def:bisim}
    Given two $\sigma$-structures $\A$ and $\B$, consider a chain of non-empty binary relations $\emptyset \neq Z_k \subseteq \dots \subseteq Z_0$ between sequences in $|\A|$ and sequences in $|\B|$ of the same length, such that for all $0\leq j\leq k$, the sequences related by $Z_j$ have length at most $k-j+1$.
    That is to say, $Z_j\subseteq \bigcup_{1\leq i\leq k-j+1}|\A|^i\times|\B|^i$ for each $j$.
    We say that these relations constitute a \define{$k$-bisimulation} between $\A$ and $\B$ if the following conditions hold:
    \begin{enumerate}
    \item\label{sim:harmony} if $s Z_j t$ for some $j$ (or, equivalently, if $s Z_0 t$), then
    $\A, s \models R \iff \B, t \models R$ for all $R \in \bsigma$;

    \item\label{sim:forth} whenever $s Z_j t$ for some $j\neq0$, for each $a\in|\A|$ such that $\epsilon(s) \prec a$ there exists some $b\in|\B|$ such that $\epsilon(t) \prec b$ and $(s.a) \ Z_{j-1} \ (t.b)$; and

    \item\label{sim:back} whenever $s Z_j t$ for some $j\neq0$, for each $b\in|\B|$ such that $\epsilon(t) \prec b$ there exists some $a\in|\A|$ such that $\epsilon(s) \prec a$ and $(s.a) \ Z_{j-1} \ (t.b)$.
    \end{enumerate}
    We say that $\Ap$ and $\Bp$ are \define{$k$-bisimilar}, denoted by $\Ap \bisim_k \Bp$, if there exists a $k$-bisimulation $(Z_j)_{0 \leq j\leq k}$ between $\A$ and $\B$ such that $\bpa Z_k \bpb$.

    A \define{$k$-simulation} from $\A$ to $\B$ is a family of non-empty relations $(Z_j)_{0 \leq j\leq k}$ defined analogously except that instead of satisfying conditions (1), (2) and (3), it satisfies condition (2) together with

    \begin{enumerate}
    \setcounter{enumi}{3}
    \item\label{sim:harmonypos}
    if $s Z_j t$ then
    $\A, s \models R \implies \B, t \models R$ for all $R \in \bsigma$.
    \end{enumerate}
\end{definition}

\begin{example}\label{ex:bisimilar_structures}
Let $\sigma = \{\acc, S\}$ where $S$ is binary. Figure~\ref{fig:example_bisimilar_structures} shows two $\sigma$-structures $\A$ and $\B$ together with a $2$-bisimulation $(Z_i)_{i\leq2}$ between them. The nested boxes represent the nested relations of the bisimulation, while the labelled edges represent choices of $a\in|\A|$ and $b\in|\B|$ in rules~\ref{sim:forth} and~\ref{sim:back} of Definition~\ref{def:bisim} respectively.

Since $(a,b)\in Z_2$ we conclude $(\A,a)\bisim_2(\B,b)$, and this trivially implies $(\A,a)\bisim_1(\B,b)$ and $(\A,a)\bisim_0(\B,b)$. Furthermore, $(\A,a)\bisim_k(\B,b)$ for any $k\geq 3$ via the $k$-bisimulation $(Z_i)_{i\leq k}$ given by $Z_i=Z_2$ for $i=3,\dots,k$. Hence $(\A,a)\bisim_k(\B,b)$ for all $k$.
\end{example}

\begin{figure}[h]
\centering
\includegraphics[scale=0.275]{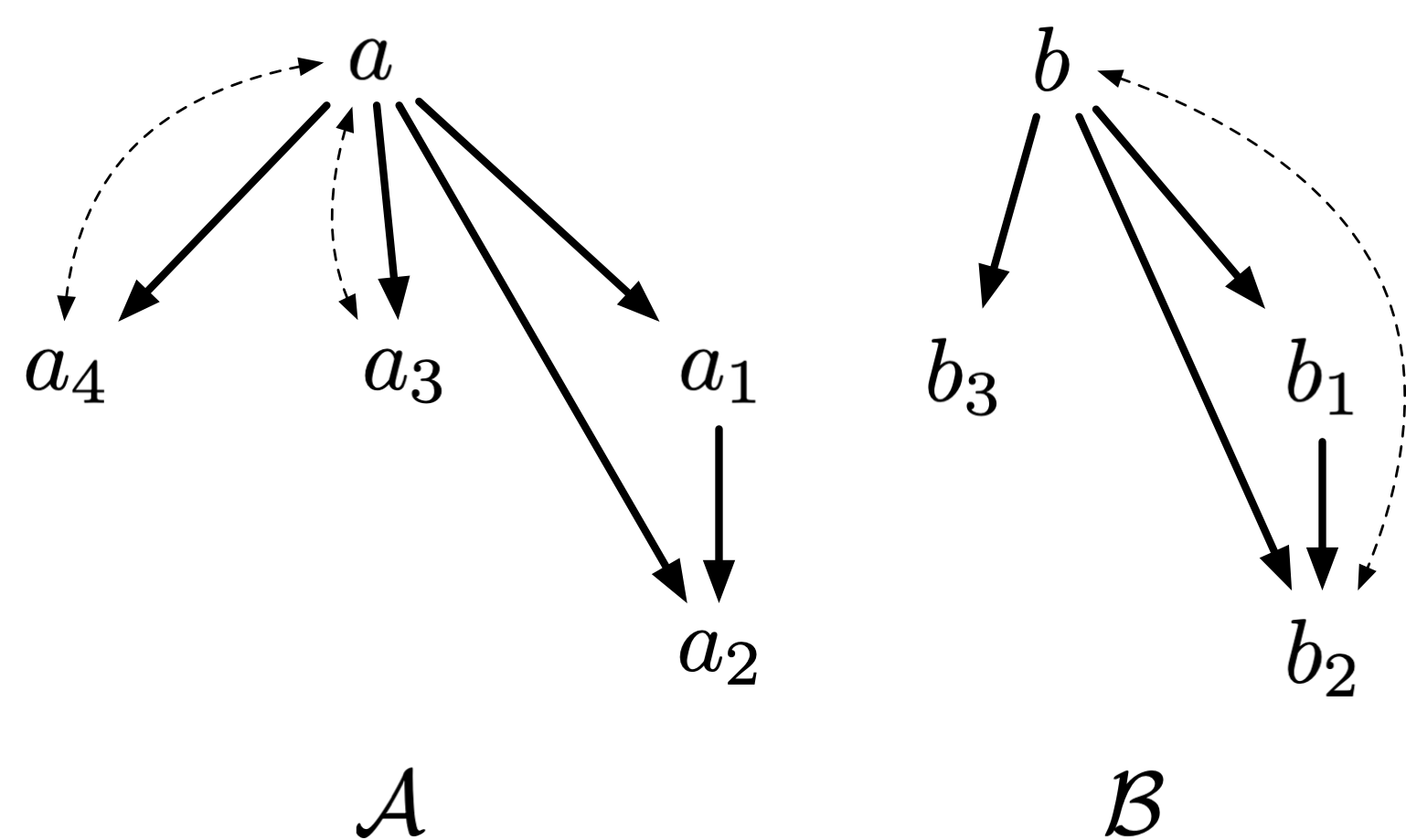}
\vspace{0.5cm}

\includegraphics[scale=0.275]{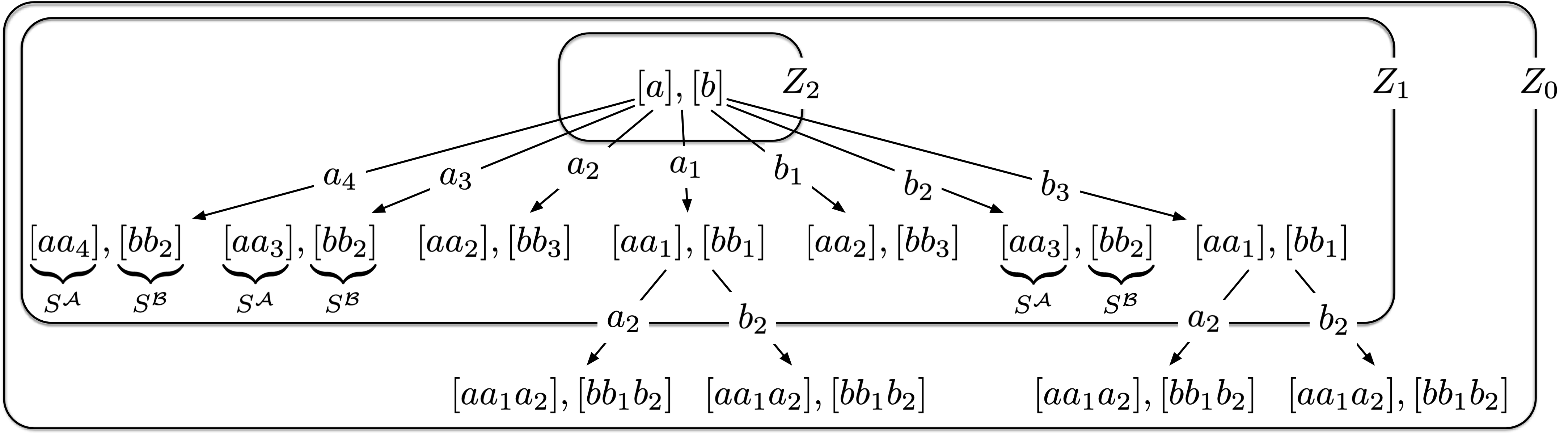}
\caption{
Two pointed $\sigma$-structures for 
$\sigma=\{\acc, S\}$, where $S$ is binary, given by $E^\A = \{a\}\times\{a_1,a_2,a_3,a_4\}$, $S^\A = \{(a,a_3),(a_3,a),(a,a_4),(a_4,a)\}$, $E^\B = \{b\}\times\{b_1,b_2,b_3\}$ and $S^\B = \{(b,b_2),(b_2,b)\}$. $\acc$ is represented by bold arrows, while $S$, being symmetric in these example structures, is represented by two-headed dotted arrows.
Below, we depict a $2$-bisimulation $(Z_i)_{i\leq 2}$ between witnessing that $\Ap \bisim_2 \Bp$. Each node of the tree contains a pair of sequences and the nested boxes represent the nested relations of the bisimulation.
}
\label{fig:example_bisimilar_structures}
\end{figure}

Simulations and bisimulations constitute a fundamental tool for the study of expressivity of modal languages thanks to the presence of so-called Hennessy-Milner properties. In our case, this is established by the following theorem.

\begin{theoremrep}\label{thm:hennessy-milner}
Let $\sigma$ be a relational signature with $E \in \sigma$ and let $(\A,a), (\B,b) \in \Structpointed$. Assume that $\sigma$ is finite or $\A$ and $\B$ are finitely branching.
Then
\begin{enumerate}
    \item $(\A,a) \simu_k (\B,b)$ if and only if $(\A,a) \Rrightarrow^+_k (\B,b)$, and
    \item $(\A,a) \bisim_k (\B,b)$ if and only if $(\A,a) \equiv_k (\B,b)$.
\end{enumerate}
\end{theoremrep}

\begin{proof}
We show $\A,a\simu_k\B,b$ iff $\A,a\Rrightarrow^+_k\B,b$. The proof for $\A,a\bisim_k\B,b$ iff $\A,a\equiv_k\B,b$ is analogous.

For the left-to-right implication, we assume $\A,a\simu_k\B,b$ via $(Z_j)_{j\leq  k}$ and we show that if $s Z_j t$, $\varphi\in \logicpos_j$ and $\A,s\models\varphi$, then $\B,t \models\varphi$ by structural induction in $\varphi$. If $\varphi=R$ of arity $r$ then $r \leq |s| = |t|$ and $\last_r(s) \in R^{\A}$. Since $Z_j\subseteq Z_0$, we have $s Z_0 t$ and by item~\ref{sim:harmonypos} of Definition~\ref{def:bisim} we conclude $\last_r(t) \in R^{\B}$. The case of $\varphi=\psi_1\land\psi_2$ is straightforward. Suppose $\varphi=\Diamond\psi$ and $j>0$. Then there is $a\in A$ such that $(\epsilon(s),a)\in \acc^{\A}$ and $\A, (s.a) \models \psi$. By item~\ref{sim:forth} of Definition~\ref{def:bisim}, there is $b\in|\B|$ such that $(\epsilon(t),b)\in \acc^\B$ and $(s.a) Z_{j-1} (t.b)$. 
By the inductive hypothesis, $\B, (t.b) \models \psi$ and then $\B, t \models \Diamond\psi$.

For the right-to-left implication, assume $\A,a\Rrightarrow^+_k\B,b$. We will show that a certain family of relations $(Z_j)_{j\leq k}$ is a $k$-simulation from $\A$ to $\B$ such that $a Z_k b$. This family is defined in the usual way as follows: 
for $s\in|\A|^{\leq k-j+1}$ and $t\in|\B|^{\leq k-j+1}$ we define
$s Z_j t$ iff for all $\varphi\in\logicpos_j$, if $\A,s \models\varphi$ then $\B,t \models\varphi$.

First, observe that $a Z_k b$ follows by the definition of $Z_k$ and the hypothesis that $\A,a\Rrightarrow^+_k\B,b$. We will now verify that $(Z_j)_{j\leq k}$ satisfies items~\ref{sim:harmonypos} and~\ref{sim:forth} of Definition~\ref{def:bisim}. For item~\ref{sim:harmonypos}, suppose $s Z_0 t$ and $\last_r(s) \in R^{\A}$ for some $R$ of arity $r \leq |s| = |t|$. Then, since relation symbols are formulas in $\logicpos_0$, we have $\last_r(t) \in R_i^{\B}$.

For item~\ref{sim:forth},
assume $s Z_j t$ for some $j>0$. By contradiction, suppose there is $a\in|\A|$ such that $\epsilon(s) \prec a$ and for all $b\in|\B|$ with $\epsilon(t) \prec b$ we have that $(s.a) Z_{j-1} (t.b)$ does not hold. Using the definition of $Z_{j-1}$, this means that for each successor $b$ of $\epsilon(t)$ we can choose a formula $\psi_b \in \logicpos_{j-1}$ such that $\A,s.a \models \psi_b$ and $\B,t.b \not\models \psi_b$. Let $\Psi$ be the set of all of these formulas, one for each $b$.

If $\B$ is finitely branching, then $\Psi$ is finite. Thus we can form $\varphi=\Diamond \left(\bigwedge \Psi \right) \in \logicpos_j$, which holds for $\A, a$ but not for $\B, b$, contradicting the definition of $Z_j$. If instead $\Psi$ is not finite but $\sigma$ is finite, we replace $\Psi$ for a subset $\Psi' \subseteq \Psi$ consisting of one formula of $\Psi$ representing each equivalence class of formulas under logical equivalence. Since $\sigma$ is finite, there are only finitely many equivalence classes in $\logicpos_{j-1}$, and thus $\Psi'$ will be finite. Thus $\varphi=\Diamond \left(\bigwedge \Psi' \right) \in \logicpos_j$ witnesses that $s$ and $t$ cannot be related by $Z_j$, giving us the desired contradiction.
\end{proof}

As in Basic Modal Logic, \logic $k$-bisimulations and $k$-simulations can also be presented in terms of games, which constitute variations of the \bml bisimulation game~\cite{blackburn2006handbook}.

\begin{definition}\label{def:ppml_games}
Given $(\A, a_0), (\B, b_0) \in \Stp$ and $k\geq 0$, the \define{$k$-round \logic bisimulation game}, denoted by $\G_k((\A, a_0), (\B, b_0))$, is played between two players, called Spoiler and Duplicator. There are $k$ rounds of the game, and the state of the game at round $\ell\leq k$ is given by a pair of sequences $(s, t) \in |\A|^{\ell + 1} \times |\B|^{\ell + 1}$.
We say that $(s, t)$ satisfies the \define{winning condition} for Duplicator iff $\A,s\models R \iff \B,t\models R$ for all $R \in \bsigma$.

The initial position (round $\ell = 0$) is $([a_0], [b_0])$.
 If $a_0$ and $b_0$ do not satisfy exactly the same unary relations, Duplicator loses the game. Otherwise, assuming the position $([a_0, \dots, a_{\ell}], [b_0, \dots, b_{\ell}])$ is reached after $\ell$ rounds with $0 \leq \ell < k$, position $\ell + 1$ is determined as follows: either Spoiler chooses $a_{\ell + 1}$ such that $a_\ell \prec a_{\ell+1}$ and Duplicator responds with $b_{\ell + 1}$ such that $b_\ell \prec b_{\ell+1}$, or Spoiler chooses $b_{\ell + 1}$ such that $b_\ell \prec b_{\ell+1}$ and Duplicator responds with $a_{\ell + 1}$ such that $a_\ell \prec a_{\ell+1}$. If Spoiler cannot make such a choice, then Duplicator wins the game immediately.
The resulting position for round $\ell + 1$ is $([a_0, \dots, a_{\ell + 1}], [b_0, \dots, b_{\ell + 1}])$. We say that Duplicator wins the round $\ell + 1$ if Duplicator is able to respond with a move which is valid according to the preceding description and which moreover makes the resulting state satisfy the winning condition. Otherwise, the game ends and Duplicator loses immediately.

A \define{winning strategy} for Duplicator consists in a choice of response that makes Duplicator win the round $\ell + 1$ for every move that Spoiler may make after any number $\ell < k$ of rounds and for any possible game state $([a_1, \dots, a_\ell], [b_0, \dots, b_\ell])$ reachable from the initial state by the progression of the game.

We also define the \define{$k$-round simulation game} $\G_k^{\rightarrow}((\A, a_0), (\B, b_0))$ as a variation of the bisimulation game in which Spoiler can only play on $\A$, Duplicator can only play on $\B$, and in which the winning condition for Duplicator is modified by replacing $\iff$ with $\implies$.
\end{definition}

\begin{example}
Duplicator has a winning strategy in the game $\G_2((\A, a), (\B, b))$, where $\A$ and $\B$ are the ones of Example~\ref{ex:bisimilar_structures}. If we ignore the boxes in Figure~\ref{fig:example_bisimilar_structures} labelled by the relations $Z_i$, the figure
shows a winning strategy for Duplicator in the form of a tree: Spoiler moves are represented as labeled arrows and positions of the game are represented as labeled nodes.
This also shows that Duplicator has a winning strategy in the game $\G_k((\A, a), (\B, b))$ for any $k\geq2$
since the game cannot be continued from the leaves of the tree.
\end{example}

\bigskip

The following result characterises $k$-(bi)simulation in terms of winning strategies in the corresponding games. We omit the proof since it follows standard ideas from comparison games.

\begin{theorem}\label{thm:bisimilarity_game}
    Given $(\A, a), (\B, b) \in \Structpointed$,
    \begin{itemize}
        \item $(\A, a) \bisim_k (\B, b)$ if and only if there exists a winning strategy for Duplicator in the game $\G_k((\A, a), (\B, b))$, and
        \item $(\A, a) \simu_k (\B, b)$ if and only if there exists a winning strategy for Duplicator in the game $\G_k^{\rightarrow}((\A, a), (\B, b))$.
    \end{itemize}
\end{theorem}

\begin{corollary}
    Given $(\A, a), (\B, b) \in \Structpointed$, if $\sigma$ is finite or $\A$ and $\B$ are finitely branching, then
    \begin{itemize}
        \item $(\A, a) \equiv_k (\B, b)$ if and only if there exists a winning strategy for Duplicator in the game $\G_k((\A, a), (\B, b))$, and
        \item $(\A, a) \equiv^+_k (\B, b)$ iff there exist winning strategies for Duplicator both in the game $\G_k^{\rightarrow}((\A, a), (\B, b))$ and in $\G_k^{\rightarrow}((\B, b), (\A, a))$.
    \end{itemize}
\end{corollary}

\subsection{The \boldmath{\logic} Comonad}
\label{sec:comonad}
We now introduce now a $k$-indexed family of comonads which corresponds to the comparison games described above. This allows us to understand multiple aspects of \logic through naturally arising constructions associated to any comonad. Just as bisimulations and games offer complementary perspectives on bisimilarity, this approach will lead us to a third characterisation of (bi)similarity in terms of the existence of certain morphisms.

In order to give a self-contained account of the comonadic characterisation of \logic, we recall the definitions of comonad and related notions as they become necessary.

\begin{definition}\label{def:comonad}
Given a category $\mathscr{E}$, a \define{comonad}%
\footnote{We give the definition of a comonad in its \emphat{comonoidal} form, i.e.\ as a comonoid object in a monoidal category of endofunctors. There is an equivalent definition, sometimes called the Manes-style or Kleisli definition, which is often useful. We focus on the comonoidal definition to emphasise the copy-and-discard informational intuition.}
on $\mathscr{E}$ is a functor $G: \mathscr{E} \to \mathscr{E}$ equipped with natural transformations $\epsilon: G \nat 1_{\mathscr{E}}$ and $\delta: G \nat GG$, called the \define{counit} and \define{comultiplication} of $G$,
such that the following diagrams commute for all objects $X \in \mathscr{E}$:
\begin{center}
\begin{tikzcd}
G X \ar[r, "\delta_X"]  \ar[d, "\delta_X"']
& G G X \ar[d, "G \delta_X"] \\
G G X \ar[r, "\delta_{G X}"]  
& G G G X
\end{tikzcd}
$\qquad \qquad$
\begin{tikzcd}
& G X \ar[ld, "\delta_X"] \ar[d,equal] \ar[rd, "\delta_X"'] & \\
G G X \ar[r, "\epsilon_{G X}", swap] & G X & G G X \ar[l, "G \epsilon_X"].
\end{tikzcd}
\end{center}
\end{definition}

Following a standard abuse of notation, we often refer to a comonad $(G, \epsilon,\delta)$ by its underlying functor $G$.

An intuition that may be useful is that applying $G$ to an object $X$ amounts to exposing information contained in $X$ and assembling it into a new object $G X$ of the same kind. Readers acquainted with the notion of the universal cover of a graph or the tree unravelling of a Kripke structure (see for instance \cite[Def. 21]{blackburn2006handbook}) might keep such construction in mind: unfolding a directed graph into a tree is a procedure which exposes the information about the paths on a graph and organises it into a new directed graph. Indeed, the comonad that we will introduce shortly constitutes a straightforward generalisation of the \bml unravelling, and other game comonads can be seen as further, more distant variations on the same core idea.

From this point of view, the component $\delta_X$ of $\delta$ duplicates the extra information in $G X$ about $X$, while $\epsilon_X$ discards it.
The diagram on the left expresses the property that for every $n$, there is a unique way of iterating this duplication of information $n$ times (a property called co-associativity) while the diagram on the right expresses the fact that
duplicating information and then discarding one of the two copies is the same as doing nothing. 

Notice that the unravelling construction satisfies a property stronger than co-associativity, namely that the unravelling of a graph, being a tree already, is isomorphic to its own unravelling. Therefore, it is not really possible to duplicate the information by applying the construction twice. This is captured by the fact that $\delta_X: GX \to GGX$ is an isomorphism for all $X$ (i.e.\ $\delta$ is a natural isomorphism), in which case we say that the comonad $G$ is \define{idempotent}.

We are now ready to define the \logic comonads $\{\Ck\}_{k\geq 0}$.

\begin{definition}\label{def:ppml_comonad}
    Let $\sigma$ be a relational signature with $E \in \sigma$ and let $k \geq 0$. Given a pointed $\sigma$-structure $\Ap$ we define $\Ck \Ap$ to be the pointed $\sigma$-structure $(\Ck \Ap, [\bpa])$ with universe
    $$|\Ck \Ap| \coloneqq \{[a_0, \dots, a_\ell] \in |\A|^{\leq k+1} \mid a_0 = \bpa \textup{ and } a_j \prec a_{j+1} \forall j \in \{0, \dots, \ell-1\}\}$$
    and basepoint $[\bpa]$.

    Relations are interpreted as follows. Let $\epsilon_\Ap: |\Ck \Ap| \to |\A|$ be the function that sends a sequence to its last element. Then for each $R \in \sigma$ of arity $r$, $(s_1, \dots, s_r) \in R^{\Ck \A}$ iff $(\epsilon_\Ap(s_1), \dots, \epsilon_\Ap(s_{r})) \in R^\A$ and moreover $s_{j+1}$ is an immediate successor of $s_j$ in the prefix order for all $j \in \{1, \dots, r -1\}$.

    Given a morphism $f: \Ap \to \Bp$ in $\Stp$, we define the homomorphism $\Ck f: \Ck \Ap \to \Ck \Bp$ by $\Ck f ([a_0,\dots,a_\ell]) \coloneqq [f(a_0),\dots,f(a_\ell)]$.\footnote{Notice that, if we write $s = [a_0,\dots,a_\ell]$, we may also write $\Ck f (s) = f(s)$ when this does not lead to confusion.} We also define, for each $\Ap \in \Stp$, a homomorphism $\delta_\Ap: \Ck \Ap \to \Ck \Ck \Ap$ by $\delta_\Ap([a_0, \dots, a_\ell]) \coloneqq [[a_0], [a_0, a_1], \dots, [a_0, \dots, a_\ell]]$.
\end{definition}

\begin{propositionrep}
    For each $k \geq 0$, $(\Ck, \epsilon, \delta)$ is a comonad on the category $\Stp$.
\end{propositionrep}
\begin{proof}
    This is a routine verification, which we spell out in detail for the interested reader.

    To see that the data for $\Ck$ given above define a functor, we must check that given $f: \Ap \to \Bp$ and $g: \Bp \to \Cp$, $\Ck(g \circ f) = \Ck g \ \circ \ \Ck f$. This is immediate from the definition. Secondly, we must check that $\epsilon$ and $\delta$ are well defined. Notice that $\epsilon_\Ap$ and $\delta_{\Ap}$ are well defined homomorphisms for all $\Ap$. Now we must check that the collections of morphisms $\{\epsilon_{\Ap} : \Ap \in \Stp\}$ and $\{\delta_{\Ap} : \Ap \in \Stp\}$ assemble into natural transformations, which is to say that the following two squares commute
\[\begin{tikzcd}[ampersand replacement=\&]
	\Ck\Ap \& \Ap \&\& \Ck\Ap \& \Ck\Ck\Ap \\
	\Ck\Bp \& \Bp \&\& \Ck\Bp \& \Ck\Ck\Bp
	\arrow["{\delta_{\Ap}}", from=1-4, to=1-5]
	\arrow["{\Ck f}"', from=1-4, to=2-4]
	\arrow["{\Ck \Ck f}", from=1-5, to=2-5]
	\arrow["{\delta_{\Bp}}"', from=2-4, to=2-5]
	\arrow["{\epsilon_{\Ap}}", from=1-1, to=1-2]
	\arrow["f", from=1-2, to=2-2]
	\arrow["{\Ck f}"', from=1-1, to=2-1]
	\arrow["{\epsilon_{\Bp}}"', from=2-1, to=2-2]
\end{tikzcd}\]
    for all $f: \Ap \to \Bp$. Indeed, take an arbitrary $[a_0,\dots,a_\ell] \in \Ck \Ap$. If we trace this element through the two squares
\[\begin{tikzcd}[ampersand replacement=\&]
	{[a_0,\dots,a_\ell]} \& {a_\ell} \& {[a_0,\dots,a_\ell]} \& {[[a_0],\dots,[a_0,\dots,a_\ell]]} \\
	{[f(a_0),\dots,f(a_\ell)]} \& {f(a_\ell)} \& {[f(a_0),\dots,f(a_\ell)]} \& {[[f(a_0)],\dots,[f(a_0),\dots,f(a_\ell)]]}
	\arrow["{\delta_{\Ap}}", maps to, from=1-3, to=1-4]
	\arrow["{\Ck f}"', maps to, from=1-3, to=2-3]
	\arrow["{\Ck \Ck f}", maps to, from=1-4, to=2-4]
	\arrow["{\delta_{\Bp}}"', maps to, from=2-3, to=2-4]
	\arrow["{\epsilon_{\Ap}}", maps to, from=1-1, to=1-2]
	\arrow["f", maps to, from=1-2, to=2-2]
	\arrow["{\Ck f}"', maps to, from=1-1, to=2-1]
	\arrow["{\epsilon_{\Bp}}"', maps to, from=2-1, to=2-2]
\end{tikzcd}\]
    we verify that both paths lead to the same result. Finally, we must check the commutativity of the diagrams in Definition~\ref{def:comonad} with $G = \Ck$, which is again accomplished by tracing a generic element through them:
\[\begin{tikzcd}[ampersand replacement=\&]
	{[a_0,\dots,a_\ell]} \& {[[a_0],\dots,[a_0,\dots,a_\ell]]} \\
	{[[a_0],\dots,[a_0,\dots,a_\ell]]} \& {\subalign{&[ \, [[a_0]], \\ &[[a_0],[a_0,a_1]], \\ &\dots, \\ &[[a_0],\dots,[a_0,\dots,a_\ell]] \, ]}} \\
	\& {[a_0,\dots,a_\ell]} \\
	{[[a_0],\dots,[a_0,\dots,a_\ell]]} \& {[a_0,\dots,a_\ell]} \& {[[a_0],\dots,[a_0,\dots,a_\ell]]}
	\arrow["{\delta_{\Ap}}"', maps to, from=3-2, to=4-1]
	\arrow["{\delta_{\Ap}}", maps to, from=3-2, to=4-3]
	\arrow["{\epsilon_{\Ck\Ap}}"', maps to, from=4-1, to=4-2]
	\arrow["{G \epsilon_{\Ap}}", maps to, from=4-3, to=4-2]
	\arrow["{\delta_{\Ap}}"', maps to, from=1-1, to=1-2]
	\arrow["{\Ck\delta_{\Ap}}"', maps to, from=1-2, to=2-2]
	\arrow["{\delta_{\Ap}}", maps to, from=1-1, to=2-1]
	\arrow["{\delta_{\Ck\Ap}}", maps to, from=2-1, to=2-2]
	\arrow[Rightarrow, no head, from=3-2, to=4-2].
\end{tikzcd}\]
\end{proof}

We will often refer to $\{\Ck\}_{k\in\N}$ as \emphat{the} \logic comonad $\Ck$, even though strictly speaking it is an indexed family of comonads. We will also often omit the basepoint of pointed structures inside subscripts, writing e.g.\ $\epsilon_\A$ instead of $\epsilon_\Ap$.

The fact that \logic reduces to \bml for unimodal signatures (see Example~\ref{ex:ppml_first_examples}) is reflected by their corresponding comonads.

\begin{definition}\label{def:comonad_morphism}
    Given two comonads $(F, \epsilon^F, \delta^F)$ and $(G, \epsilon^G, \delta^G)$ over a common category, a \define{comonad morphism} $(F, \epsilon^F, \delta^F) \Rightarrow (G, \epsilon^G, \delta^G)$ is a natural transformation $\alpha : F \nat G$ between the underlying functors such that $\epsilon^F = \epsilon^G \circ \alpha$ and $\delta^G \circ \alpha = \alpha \alpha \circ \delta^F$.\footnote{$\alpha \alpha$ is the horizontal composition of $\alpha$ with itself, which can be computed as $G \alpha \circ \alpha F$.} 
\end{definition}

\begin{proposition}
    If $\sigma$ is a unimodal signature, then $\Ck$ is isomorphic to the Modal Comonad $\Mk$ on $\sigma$~\cite{abramsky2021relating}.
\end{proposition}
\begin{proof}
    This is immediate since the isomorphism between structures $\Ck\Ap$ and $\Mk\Ap$ (natural in $\Ap$) is evident and the counit and comultiplication for both comonads are given by the same formula.
\end{proof}
For our purposes we will identify $\Ck$ and $\Mk$ whenever $\sigma$ is unimodal, in which case we refer to both as the \define{Basic Modal Comonad}.
Notice moreover that for any pointed Kripke structure $\Ap \in \Stp$, $\Mk\Ap$ is precisely the unravelling of $\A$ up to $k$ steps starting from $\bpa$, which is a Kripke tree of height at most $k$.
If now we allow $\sigma$ to be an arbitrary signature with $E \in \sigma$, we can still think of $\Ck\Ap$ as an unravelling of $\Ap$ transforming it into a particular kind of tree.

\begin{definition}\label{def:pptree}
    A \define{path-predicate tree} or \define{pp-tree} is a pointed $\sigma$-structure $(\+{T}, u)$ such that (1) $(|\+{T}|, \acc^{\+T}, u)$ is a rooted tree and (2) for each $R \in \bsigma$ of arity $r$, if $(u_1, \dots, u_r) \in R^{\+T}$, then $u_1 \prec \dots \prec u_r$, i.e.\ $[u_1,\dots,u_r]$ is the unique chain of length $r$ ending in $u_r$. The \define{height} of a pp-tree is its height as a rooted tree.
\end{definition}

It is immediate that given any $\Ap \in \Structpointed$, $\Ck \Ap$ is a pp-tree of height at most $k$ (see Figure~\ref{fig:pptrees}).

\begin{remark}\label{rem:notation_for_pptrees}
    We may denote a pp-tree $\Tp$ by $\T$, since the root is determined by the requirement that the $\sigma$-structure $\T$ is a pp-tree. Moreover, given any $v \in |\T|$, we write $\T_v$ both for the unique $E$-chain in $\T$ from the root to $v$, and for the embedded substructure of $\T$ determined by this chain. Which usage is meant will be clear from context. We will also treat $\T_v$ as a tuple or valuation for the semantics of \logic. In particular, notice that condition (2) in Definition~\ref{def:pptree} can be restated as follows: if $\T, s \models R$, then $s$ is a suffix of $\T_{\epsilon(s)}$; in particular, $\T, \T_{\epsilon(s)} \models R$.
\end{remark}

\begin{figure}[h]
\centering
\includegraphics[scale=0.3]{example_bisim.png}\qquad
\includegraphics[scale=0.3]{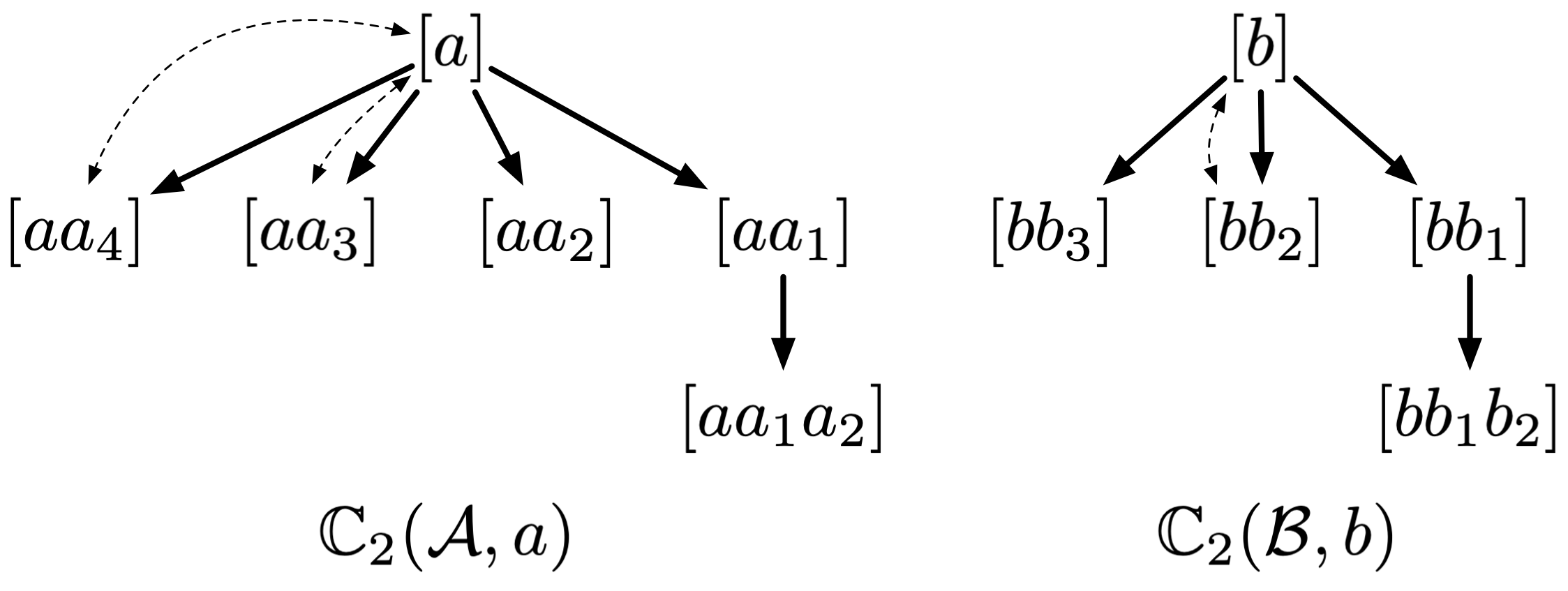}
\caption{The structures $\A$ and $\B$ from Figure~\ref{fig:example_bisimilar_structures} and their \logic-unravellings $\bb C_2 \Ap$ and $\bb C_2 \Bp$.
} 
\label{fig:pptrees}
\end{figure}

\begin{remark}\label{rem:Ck_vs_Ek_vs_Hybrid}
    The interpretation of relations in $\Ck\Ap$ is similar to that of $\Ek\Ap$ where $\Ek$ is the Ehrenfeucht-Fraïssé (EF) comonad (see Definition~\ref{def:EF_and_pebbling_comonads})
    but with an additional locality constraint: tuples of sequences related by some $R \in \sigma$ must be \emphat{immediate} extensions of each other. Thus $\Ck\Ap$ is \emphat{not} an embedded substructure of $\Ek \Ap$, contrary to the case of the Hybrid and Bounded comonads~\cite{abramsky2022hybrid}. In this sense, the \logic comonad occupies another, distinct middle ground between the Modal and EF comonads.
\end{remark}

A comonad $G$ that arises from the study of a certain comparison game receives the name of a game comonad. Although this is not a formal definition, all game comonads defined to date share multiple properties, among which we take the following to be fundamental:
\begin{enumerate}[(I)]
    \item morphisms of type $G X \to Y$
    correspond to winning strategies for Duplicator in some existential, one-way model comparison game played from $X$ to $Y$, and
    \item the category of coalgebras of $G$ is arboreal, and therefore pairs of objects in that category are equipped with an intrinsic notion of back-and-forth comparison game between them.
\end{enumerate}
We now prove that $\Ck$ satisfies these two properties, and that the corresponding games coincide with the \logic simulation and bisimulation games, thus making the name `\logic comonad' appropriate.
The first property is quite straightforward; indeed, the definition of $\Ck$ is reverse-engineered from the desideratum that it holds.

\begin{proposition}\label{prop:kleisli_morphisms}
    Given $(\A, a), (\B, b) \in \Structpointed$, there is a bijective correspondence between homomorphisms $\Ck (\A, a) \to (\B, b)$
    and the set of winning strategies for Duplicator in the $k$-round simulation game $\G_k^{\rightarrow}((\A, a), (\B, b))$.
\end{proposition}

\begin{proof}
    By definition, the elements of $|\Ck \Ap|$ are exactly the valid sequences of moves for Spoiler. The definition of the interpretations of relations in $\Ck \Ap$ is exactly such that for any function $f: |\Ck \Ap| \to |\B|$, $f$ is a pointed homomorphism if and only if, for all $s \in |\Ck\Ap|$, $f(s)$ is a valid and winning answer of Duplicator to the state of the game up to that point.
\end{proof}

\begin{corollary}\label{coro:logical_equivalence_kleisli}
    Let $(\A, a), (\B, b) \in \Structpointed$ and suppose that either $\sigma$ is finite or $\A$ and $\B$ are finitely branching structures. Then $(\A, a) \equiv_k^+ (\B, b)$ if and only if there exist homomorphisms $\Ck(\A, a) \to (\B, b)$ and $\Ck(\B, b) \to (\A, a)$.
\end{corollary}

From a categorical perspective, morphisms $\Ck\Ap \to \Bp$ can be understood as functions $\Ap \to \Bp$ that depend on `extra input' (cf. side effects, which are `extra output' of computations, captured by \emphat{monads}). This may serve as motivation for the following definition.

\begin{definition}
    Given a category $\cat{E}$ and a comonad $(G, \epsilon, \delta)$ on $\cat{E}$, the \define{Kleisli category} of $G$, denoted by $\Kl(G)$, is the category whose objects are the objects of $\cat{E}$ and whose morphisms $f: X \xrightarrow{\cdot} Y$, which we annotate with a dot to distinguish them from $\cat{E}$-morphisms, are given by morphisms $f: G X \to Y$ in $\cat{E}$. The identity morphism on an object $X$ is given by $\epsilon_X$, and given morphisms $f: X \xrightarrow{\cdot} Y$ and $g: Y \xrightarrow{\cdot} X$, their composite $g\circ f$ is given by the $\cat{E}$-morphism $G X \xrightarrow{\delta_X} G G X \xrightarrow{G f} GY \xrightarrow{g} Z$.
\end{definition}

Thus, the content of Proposition~\ref{prop:kleisli_morphisms} can be restated by saying that $\Kl(\Ck)$ is the category of $\sigma$-structures where a morphism $\Ap \to \Bp$ is precisely a winning strategy for Duplicator in the game $\G_k^{\rightarrow}(\Ap, \Bp)$. This point of view emphasises that strategies can be composed with each other by the composition law of $\Kl(\Ck)$. Moving from $\Stp$ to $\Kl(\Ck)$ can also be thought of as replacing homomorphisms by a weaker notion of morphism, in the sense that the existence of a homomorphism $\Ap \to \Bp$ is a strictly stronger condition than the existence of a homomorphism $\Ck\Ap \to \Bp$.\footnote{To obtain the latter from the former, simply precompose with $\epsilon_\A$.} In this way morphisms in $\Kl(\Ck)$ approximate homomorphisms (see e.g. \cite[Section 7]{abramsky2021relating}).

The Kleisli category of a comonad is one of the two fundamental categorical constructions that can be produced from it. The second one is the category of coalgebras, or Eilenberg-Moore category, to which we now turn.\footnote{Any adjunction between categories presents a comonad (and a monad). In the other direction, starting from a given comonad, these two constructions are the two universal solutions to finding an adjunction that presents it. Monads and comonads enjoy an incredibly rich theory which revolves around these two constructions and the associated adjunctions.}

\begin{definition}
    Let $(G, \epsilon, \delta)$ be a comonad on a category $\mathscr{E}$. A \define{coalgebra} for $G$ or \define{$G$-coalgebra} consists of an object $X \in \mathscr{E}$ together with a morphism $\gamma: X \to G X$ in $\mathscr{E}$, called its \define{structure map}, such that $\epsilon_X \circ \gamma = \id{X}$ and $\delta_X \circ \gamma = G(\gamma) \circ \gamma$. A morphism between coalgebras $f: (X, \gamma) \to (Y, \eta)$ is a morphism $f: X \to Y$ in $\mathscr{E}$ that commutes with the structure maps, i.e.\ $\eta \circ f = \gamma \circ G(f)$. This defines the category of coalgebras or \define{Eilenberg-Moore category} of $G$, denoted by $\EM(G)$. 
\end{definition}

The reader need not keep in mind the above definition for too long, since the following fact allows us to simplify the discussion of $\Ck$-coalgebras enormously.

\begin{propositionrep}\label{prop:Ck_is_idempotent}
    $\Ck$ is an idempotent comonad, i.e. $\delta$ is a natural isomorphism.
\end{propositionrep}
\begin{proof}
    We show that $\Ck\epsilon: \Ck\Ck \nat \Ck$ is the inverse of $\delta$.

    For every $\Ap$, we must show that $\Ck \epsilon_\A \circ \delta_\A = \id{\Ck \A}$ and $\delta_\A \circ \Ck \epsilon_\A = \id{\Ck \Ck \A}$. $\Ck \epsilon_\A$ acts by sending a sequence of sequences $[s_0,\dots,s_\ell]$ to the sequence $[\epsilon_\A(s_0),\dots,\epsilon_\A(s_\ell)]$, from which the first of the two equalities is immediate. The second amounts to proving that, for any $[s_0,\dots,s_\ell] \in \Ck \Ck \Ap$,
    $$[s_0,\dots,s_\ell] = [[\epsilon(s_0)], [\epsilon(s_0), \epsilon(s_1)], \dots, [\epsilon(s_0), \dots, \epsilon(s_\ell)]]. $$
    We do this by induction on $\ell$. For $\ell=0$ it holds since if $[s_0] \in \Ck \Ck \Ap$ then $s_0$ must be the basepoint of $\Ck \Ap$, that is to say $s_0 = [\bpa]$, thus $[s_0] = [[a]]$. On the other hand if the equality holds for sequences of length $\ell + 1$, then for any sequence $[s_0,\dots,s_{\ell+1}] \in \Ck \Ck \Ap$ of length $\ell + 2$, we have 
    $$ [s_0,\dots,s_{\ell+1}] = [[\epsilon(s_0)], [\epsilon(s_0), \epsilon(s_1)], \dots, [\epsilon(s_0), \dots, \epsilon(s_\ell)], s_{\ell+1}]. $$
    Finally, $s_{\ell+1} = [\epsilon(s_0), \dots, \epsilon(s_{\ell+1})]$ since $s_{\ell+1}$ must be an immediate successor of $s_\ell$ in the prefix ordering.
\end{proof}

Two well-known consequences of a comonad $G$ on $\cat{E}$ being idempotent are the following (see \cite[Prop 4.2.3]{borceux1994handbook} for the statements in dual form).
\begin{enumerate}
    \item $\EM(G)$ is a coreflective subcategory of $\cat{E}$. This means in particular that if an object $X \in \mathscr{E}$ admits a coalgebra structure, it is unique. Thus, we may talk about $X$ being a $G$-coalgebra as a \emphat{property} of the object rather than structure on it.
    \item If $X$ is a $G$-coalgebra, its structure map $X \to G X$ is an isomorphism. Therefore, all $G$-coalgebras are isomorphic to an object of the form $G X$ for some $X \in \mathscr{E}$.
\end{enumerate}
Thanks to (1) we need only identify which pointed $\sigma$-structures are $\Ck$-coalgebras to identify $\EM(\Ck)$ as the full subcategory of $\Structpointed$ spanned by those objects. From (2), we obtain immediately the desired characterisation.

\begin{corollary}\label{coro:caract_coalgebras}
    A pointed structure $(\A, a) \in \Structpointed$ is a $\Ck$-coalgebra if and only if it is a pp-tree of height at most $k$. Therefore, $\EM(\Ck)$ is the full subcategory of $\Stp$ spanned by pp-trees of height at most $k$.
\end{corollary}
\begin{proof}
    If $\Ap$ is a $\Ck$-coalgebra, then by (2) above $\Ap$ is isomorphic to $\Ck\Ap$, which is a pp-tree of height at most $k$, and the property of being a pp-tree is invariant under isomorphism, as is the height of the pp-tree. Conversely, given a pp-tree $\T$ of height at most $k$, the counit $\epsilon_\T: \Ck\T \to \T$ is injective, surjective and strong, hence an isomorphism.
\end{proof}

Following the discussion of previous game comonads, we may introduce the coalgebra number $\kappa\Ap$ defined as the smallest $k$ such that $\Ap$ is a $\Ck$-coalgebra, if it exists.
This parameter generalises the corresponding coalgebra number for $\Mk$ in the obvious way: it is defined only for pp-trees, and it coincides with the height of the pp-tree. Coalgebras of idempotent comonads are not equipped with extra structure with respect to $\sigma$-structures, which explains why they do not give rise to rich combinatorial parameters. In contrast, tree-depth and tree-width arise in this way from the non-idempotent game comonads $\Ek$ and $\bb P_k$, respectively.

\subsection{Arboreal Categories and Bisimilarity through Bounded Morphisms}\label{sec:arboreal}

\begin{toappendix}
\subsection{Proofs for Section 2.3}
It is well known that the categories $\Struct, \Structpointed$ are complete and cocomplete.\footnote{A category is said to be (co)complete if it has all (co)limits of small diagrams, i.e. diagrams indexed by small categories.} For reference, we state here explicit descriptions of (co)products and (co)equalisers in $\Structpointed$, in terms of which all other (co)limits can be expressed (see e.g.\ \cite[Thm. 3.4.12]{riehl2017category} or any other standard reference on basic Category Theory).

\begin{proposition}\label{prop:caract_limits_in_structpointed}
    For any small family of pointed relational structures $(\A_\alpha, \bpa_\alpha)_{\alpha \in \Lambda}$, the product $\prod_\alpha (\A_\alpha, \bpa_\alpha)$ can be characterised as the following structure. Its universe is the Cartesian product $\prod_\alpha |\A_\alpha|$ and its basepoint is $(a_\alpha)_{\alpha\in\Lambda}$. The product projections are given by the projection functions $\pi_\alpha$ from the Cartesian product to each of its set-theoretical factors, and the interpretation of $R \in \sigma$ is
    $R^{\prod_\alpha \A_\alpha} \coloneqq \cap_\alpha \pi_\alpha^{-1}(R^{\A_\alpha})$.
    The empty product ($\Lambda =\varnothing$), i.e.\ the terminal object, is given by a singleton universe $\{*\}$ with interpretations $R^{\{*\}} \coloneqq \{(*,\dots,*)\}$ for each $R \in \sigma$ and basepoint $*$.

    Moreover, given a parallel pair of morphisms $f,g: \Ap \rightrightarrows \Bp$ in $\Structpointed$, their equaliser can be characterised as the inclusion morphism $h: (\A', a) \into \Ap$ where $\A'$ is the embedded substructure of $\A$ on the set of all $a'$ such that $f(a') = g(a')$.
\end{proposition}

\begin{proposition}\label{prop:caract_colims_in_structpointed}
    For any small family of pointed relational structures $(\A_\alpha, \bpa_\alpha)_{\alpha \in \Lambda}$, the coproduct $\coprod_\alpha (\A_\alpha, \bpa_\alpha)$ can be characterised as the following structure. Its universe is given by $\coprod_\alpha |\A_\alpha|/{\sim}$, i.e.\ the disjoint union of the universes modulo the equivalence relation $\sim$ generated by identifying all basepoints. Concretely, $(\alpha, a') \sim (\beta, a'')$ iff $a' = \bpa_\alpha$ and $a'' = \bpa_\beta$, or $(\alpha, a') = (\beta, a'')$. Its basepoint is the equivalence class of all basepoints. The coproduct inclusions $\iota_\alpha$ are given by sending $a \mapsto [(\alpha,a)]$, and the interpretation of $R \in \sigma$ is
    $R^{\coprod_\alpha \A_\alpha} \coloneqq \union_\alpha \iota_\alpha(R^{\A_\alpha})$.
    The empty coproduct ($\Lambda = \varnothing$), i.e.\ the initial object, is given by a singleton universe with empty interpretations.

    Moreover, given a parallel pair of morphisms $f,g: \Ap \rightrightarrows \Bp$ in $\Structpointed$, their coequaliser can be characterised as the morphism $h: \Bp \to (\B/{\simeq}, [\bpb])$ given as follows. $\B/{\simeq}$ is the quotient of $\B$ by the equivalence relation ${\simeq}$ generated by the pairs of the form $(f(a),g(a))$ for all $a \in |\A|$, that is to say $|\B/{\simeq}| \coloneqq |\B|/{\simeq}$ and $R^{\B/{\simeq}} \coloneqq \{([b_1],\dots,[b_r]) \mid (b_1,\dots,b_r) \in R^\B\}$ for each $R \in \sigma$; and $h$ is the quotient map $h(b) \coloneqq [b]$.
\end{proposition}
\end{toappendix}

In this section we explain how the \logic comonad encodes bisimulation games, and hence how we can characterise bisimilarity in categorical terms. To this end, we will use the axiomatic approach of arboreal categories. This lets us access general results in a clean fashion; concretely, we will make use of Propositions 42 and 46 in~\cite{abramsky2021arboreal}.
Arboreal categories and covers constitute a categorical axiomatisation of the situation arising from a game comonad. In~\cite{abramsky2021arboreal}, the authors draw an analogy with computability and computational complexity: in the same way that assigning a program (an intensional description) to a computable function (an extensional object) allows us to assign some complexity measure to the function itself, objects in an arboreal category may serve as intensional descriptions of objects in some `extensional' category that we wish to study—in our case, the category $\Stp$. These intensional descriptions are axiomatically defined to be tree-shaped since they represent processes unfolding in space and time, hence the name `arboreal category'.

We proceed as follows.

\begin{itemize}
    \item We define arboreal categories and arboreal covers, and give some intuition for them.
    \item We prove property (II) above, namely that $\EM(\Ck)$ is an arboreal category for each $k \geq 0$. In fact we prove a slightly stronger statement: the $k$-indexed family of \logic comonads constitutes a resource-indexed arboreal cover of $\Stp$.
    \item We show that the abstract bisimulation game associated with each arboreal category $\EM(\Ck)$ coincides with the \logic $k$-bisimulation game.
    \item We prove that the general notion of open pathwise embedding in an arboreal category reduces in our case to a natural notion of bounded morphism.
    \item Putting all of this together, we end the section with Theorem~\ref{thm:full_logical_equivalence}, which characterises \logic $k$-bisimilarity using bounded morphisms.
\end{itemize}

\begin{definition}
    Let $\ncat{ppTree}$ be the full subcategory of $\Structpointed$ spanned by all (not necessarily finite) pp-trees.
\end{definition}
Notice that $\ncat{ppTree}$ contains all the categories $\EM(\Ck)$ as full subcategories. Our first step towards proving (II) is to verify that $\ncat{ppTree}$ is arboreal, hence we will interleave intermediate definitions with their verifications in the case of $\ncat{ppTree}$.

We have already assumed knowledge of the basic categorical concepts of categories, functors, and natural transformations. For the remainder of this section we will also assume familiarity with limits and colimits (including products, coproducts and pullbacks) as well as adjunctions.
We review the definition of pullbacks in order to fix some terminology. Given a diagram $X \xrightarrow{f} Z \xleftarrow{g} Y$ in $\cat{C}$, its pullback, if it exists, is an object $X \times_Z Y$ in $\cat{C}$ together with morphisms $\overline{f}$ and $\overline{g}$ making the following square commute
\[\begin{tikzcd}
	{X\times_ZY} & X \\
	Y & Z
	\arrow["{\overline{g}}", from=1-1, to=1-2]
	\arrow["{\overline{f}}"', from=1-1, to=2-1]
	\arrow["\lrcorner"{anchor=center, pos=0.125}, draw=none, from=1-1, to=2-2]
	\arrow["f", from=1-2, to=2-2]
	\arrow["g"', from=2-1, to=2-2]
\end{tikzcd}\]
and such that for any other object $W$ and morphisms $h: W \to X$ and $k: W \to Y$ such that $f\circ h = g\circ k$, there exists a unique morphism $\phi: W \to X\times_Z Y$ such that $\overline{g}\circ\phi = h$ and $\overline{f}\circ\phi = k$. We decorate the above commutative square with the symbol $\lrcorner$ to indicate that it is a pullback square, and we refer to $\overline{f}$ as \define{the pullback of $f$ along $g$}. Notice that, as is always the case for definitions through universal properties, the object $X\times_Z Y$ (and hence the morphisms $\overline{f}$ and $\overline{g}$) are defined only up to isomorphism.

We assume $\cat{C}$ is a locally small and well powered category.

\begin{definition}[Factorisation systems]\label{def:stable_proper_fact_syst}
    Given a category $\cat{C}$ and a pair of arrows $e$ and $m$ in $\cat{C}$, we say that $e$ has the \define{left lifting property} with respect to $m$, or that $m$ has the \define{right lifting property} with respect to $e$, if for every commutative square
\[\begin{tikzcd}
	\bullet & \bullet \\
	\bullet & \bullet
	\arrow["e", from=1-1, to=1-2]
	\arrow[from=1-1, to=2-1]
	\arrow["d"', dashed, from=1-2, to=2-1]
	\arrow[from=1-2, to=2-2]
	\arrow["m"', from=2-1, to=2-2]
\end{tikzcd}\]
    there exists a diagonal filler $d$ (possibly non-unique) such that the two resulting triangles commute.
    A pair of classes of morphisms $(\Q, \M)$ is a \define{weak factorisation system} on $\cat{C}$ iff (1) every morphism in $\cat{C}$ can be factored as $f = m \circ e$ with $e \in \mathscr{Q}$ and $m \in \mathscr{M}$, (2) $\mathscr{Q}$ is precisely class of morphisms having the left lifting property against every morphism in $\mathscr{M}$, and (3) $\mathscr{M}$ is precisely the class of morphisms having the right lifting property against every morphism in $\mathscr{Q}$.
    A factorisation system $(\Q, \M)$ on $\cat{C}$ is \define{proper} if all $\mathscr{Q}$-morphisms are epimorphisms and all $\mathscr{M}$-morphisms are monomorphisms, and it is \define{stable} if for any $e \in \mathscr{Q}$ and $m \in \mathscr{M}$ with common codomain, the pullback of $e$ along $m$ exists and belongs to $\mathscr{Q}$.
    We write $\M$-morphisms using the arrow shape $\emb$ and $\Q$-morphisms using the arrow shape $\epi$.
\end{definition}

We have anticipated that, intuitively, objects in an arboreal category are `tree-shaped'. More precisely, this is enforced by requiring that all objects are path-generated (Def.~\ref{def:path_generated_object}), which means that  they can be obtained by glueing together path-shaped objects along initial segments. However, this in turn requires a well-behaved notion of path-shaped object.

One way to talk about the shape of an object is by considering the shape of its poset of subobjects. Then we could postulate that a path is an object whose poset of subobjects is a finite total order. This can be done in any category, but the standard notion of subobject turns out to be inadequate for this task. For instance, even the subobjects of a $\sigma$-structure with a singleton universe may not be totally ordered, since for each tuple in the interpretation of a relation, one can obtain a proper subobject by subtracting that tuple from the interpretation. Factorisation systems, which are also used in other categorical axiomatic contexts (e.g.\ model categories in homotopy theory~\cite{riehl2008factorization}), give us a way to solve this problem by giving us some flexibility in the definition of a subobject.

\begin{definition}\label{def:subobjects_relative_to_factorisation_system}
    Let $\cat{C}$ be a category equipped with a stable proper factorisation system $(\Q, \M)$ and let $X\in\cat{C}$. An \define{$\mathscr{M}$-subobject} of $X$ is an equivalence class of $\mathscr{M}$-morphisms with codomain $X$ according to the relation $\sim$ given as follows: $m \sim n$ whenever there exists an isomorphism $i$ such that $m = n \circ i$. The set of $\mathscr{M}$-subobjects of $X$ has a natural partial ordering given by $[m] \leq [n]$ iff there exists a morphism $i$ such that $m = n \circ i$. We write $[m] \prec [n]$ if $[n]$ is an immediate successor of $[m]$. As is often done, we abuse notation and refer to a $\M$-subobject by any of its representatives.
\end{definition}

\begin{propositionrep}\label{prop:fact_system}
    In the context of the category $\ncat{ppTree}$, let $\Q$ denote the class of pointed surjective homomorphisms and let $\M$ denote the class of pointed relational embeddings. Then $(\Q, \M)$ is a stable proper factorisation system on $\ncat{ppTree}$.
\end{propositionrep}
\begin{proof}
Notice that epimorphisms in $\ncat{ppTree}$ are precisely the surjective pointed homomorphisms and monomorphisms are precisely the injective pointed homomorphisms; in particular, relational embeddings are monic. Since every $\Q$-morphism is an epimorphism and every $\M$-morphism is a monomorphism, if $(\Q,\M)$ is a weak factorisation system on $\ncat{ppTree}$, then it is proper.

Recall that any homomorphism $f: \Ap \to \Bp$ can be given a canonical image factorisation $f = i \circ f|^{f(\A)}$ where $f|^{f(\A)}: \Ap \to (f(\A), f(\bpa))$ is the co-restriction of $f$ to its image, $f(\A)$ is the structure defined by $|f(\A)| \coloneqq f(|\A|) \subseteq |\B|$ and $R^{f(\A)} \coloneqq R^\B \cap |f(\A)|^r$ for each $R \in \sigma$ of arity $r$, and $i$ is the inclusion $f(\A) \into \B$.
Notice that indeed $f|^{f(\A)}$ is a pointed relational embedding and $i$ is a pointed surjective homomorphism. Moreover this construction, which describes the usual stable proper factorisation system on $\Stp$, restricts properly to $\ncat{ppTree}$, in the sense that if $\Ap$ and $\Bp$ are pp-trees, then $(f(\A),f(\bpa))$ is a pp-tree as well. Indeed, it is straightforward to verify that $f(\bpa)$ has no predecessors in $f(\A)$, $f(a')$ has a unique predecessor in $f(\A)$ for each $a' \neq \bpa$ (given by the image of the unique predecessor of $a'$), every point of $f(\A)$ is obviously accessible from $\bpb$, and relations in $f(\A)$ only hold along paths (since interpretations in $f(\A)$ are subsets of the interpretations in $\B$).

The preceding discussion establishes that any morphism in $\ncat{ppTree}$ admits a $(\Q,\M)$-factorisation. We now show that conditions (2) and (3) of Def.~\ref{def:stable_proper_fact_syst} hold as well. Consider a commutative square in $\ncat{ppTree}$
\[\begin{tikzcd}
	{(\+X, x)} & {(\+Y, y)} \\
	{(\+Z,z)} & {(\+W,w)}
	\arrow["e", two heads, from=1-1, to=1-2]
	\arrow["f"', from=1-1, to=2-1]
	\arrow["d"', dashed, from=1-2, to=2-1]
	\arrow["g", from=1-2, to=2-2]
	\arrow["m"', tail, from=2-1, to=2-2]
\end{tikzcd}\]
where $e \in \Q$ and $m \in \M$. Notice that $f$ is constant along the fibres of $e$, i.e.\ $e(x')=e(x'')$ implies $f(x')=f(x'')$ for all $x', x''\in |\+X|$, because $e(x')=e(x'')$ implies $m\circ f(x') = g\circ e(x')=g\circ e(x'') = m\circ f(x'')$ and $m$ is injective. Hence, the function $f$ lifts to a unique function $d: |\+Y| \to |\+Z|$ such that $d\circ e = f$ as functions.
On the other hand, since $g\circ e = m \circ f = m\circ d\circ e$ and $e$ is surjective, we conclude that $g=m\circ d$ as set-theoretical functions. Let us now see that $d$ is a homomorphism (the fact that it preserves basepoints is obvious from the definition). Given $(y_1,\dots,y_r) \in R^{\+Y}$, we know that $(g(y_1),\dots,g(y_r))\in R^{\+W}$, but $g(y_i) = m\circ d(y_i)$ for each $i \in \{1,\dots,r\}$, hence since $m$ is a relational embedding we conclude that $(d(y_1),\dots,d(y_r))\in R^{\+Z}$.

We have found a diagonal filler for any $m \in \M$ and any $e \in \Q$, which shows that $\Q$ is contained in the class of morphisms having the left lifting property against every morphism in $\M$, and that $\M$ is contained in the class of morphisms having the right lifting property against every morphism in $\Q$. We must now check the reverse inclusions. We do one of them, the other one being dual.
Let $e: (\+X, x) \to (\+Y, y)$ be a morphism in $\ncat{ppTree}$ such that it has the left lifting property against every morphism in $\M$, and take a $(\Q, \M)$-factorisation of $e$, say $e = (\+X,x) \overset{\tilde{e}}{\epi} (\+Z,z) \overset{m}{\emb} (\+Y, y)$. Then there is some diagonal filler $d$ making the square
\[\begin{tikzcd}
	{(\+X,x)} & {(\+Y,y)} \\
	{(\+Z,z)} & {(\+Y,y)}
	\arrow["e", from=1-1, to=1-2]
	\arrow["{\tilde{e}}"', two heads, from=1-1, to=2-1]
	\arrow["d"', dashed, from=1-2, to=2-1]
	\arrow[Rightarrow, no head, from=1-2, to=2-2]
	\arrow["m"', tail, from=2-1, to=2-2]
\end{tikzcd}\]
commute, i.e.\ $m\circ d = 1_{(\+Y,y)}$, i.e.\ $m$ is a split epimorphism. Recall that a split epimorphism that is also a monomorphism is an isomorphism, hence $m$ is an isomorphism. Since $\Q$ is closed under isomorphisms, it follows that $e \in \Q$.

We have shown that $(\Q,\M)$ is a proper factorisation system. As for stability, it is straightforward to verify that the pullback in $\Stp$ of a diagram of pp-trees is also a pp-tree, and since inclusions of full subcategories reflect limits, this defines the pullback in $\ncat{ppTree}$. Hence $\ncat{ppTree}$ has all pullbacks. The fact that the pullback of $e \in Q$ along $m \in \M$ is again in $\Q$ follows from the fact that the forgetful functor $\ncat{ppTree} \to \ncat{Set}$ sending each pp-tree to its underlying universe preserves pullbacks, which means that the universes and underlying set-theoretical functions of the pullback in $\ncat{ppTree}$ are constructed as the pullback of the universes and underlying functions in $\ncat{Set}$. Since the pullback of a surjective function along any function is again surjective, stability follows.
\end{proof}

\begin{definition}\label{def:path}
    Given a category $\cat{C}$ equipped with a stable proper factorisation system $(\Q, \M)$, an object $X \in \cat{C}$ is a \define{path} iff its poset of $\M$-subobjects is a finite chain. We say a morphism $m \in \M$ is a \define{path embedding} if its domain is a path, and we denote by $\bb P X$ the sub-poset of $\M$-subobjects of $X$ which are (represented by) path embeddings.
\end{definition}
Notice that from this categorical perspective paths \emphat{in} an object of $\cat{C}$ are precisely the path embeddings (when considered up to isomorphism, i.e.\ as a particular kind of $\M$-subobject).

\begin{definition}
    A \define{pp-path} is a pointed $\sigma$-structure $\Pp$ such that $(|\+P|, \acc^{\+P})$ is a finite chain with minimal element $\bpp$. In other words, a pp-path is a pp-tree $\+P$ with a single branch.
\end{definition}

\begin{propositionrep}\label{prop:aths_are_pp-paths}
    Paths in $\ncat{ppTree}$ with the factorisation system $(\Q,\M)$ as above are precisely the pp-paths.
\end{propositionrep}
\begin{proof}
    Begin by noticing that given a pointed $\sigma$-structure $\Tp$, for each of its $\M$-subobjects we can take as a representative an embedded substructure of $\Tp$ (which by definition must contain $\bpt$). Moreover, under this choice of representatives, given two subobjects $m : (\+S, s) \emb \Tp$ and $m': (\+S', s')\emb\Tp$, $m \leq m'$ if and only if $|\+S| \subseteq |\+S'|$. Hence we must show that a pp-tree $\T$ is a pp-path if and only if its poset of embedded sub-pp-trees ordered by inclusion is a finite chain.

    Let $\T$ be a pp-path whose underlying $E$-tree is given by $\bpt \prec u_1 \prec \dots \prec u_\ell$ for some $\ell \geq 0$. Then there are exactly $\ell + 1$ embedded sub-pp-trees of $\T$, whose universes are given by the $\ell+1$ non-empty prefixes of the sequence $[\bpt, u_1,\dots,u_\ell]$. In particular, the poset of sub-pp-trees of $\T$ is a finite chain.
    Conversely, let $\T$ be a pp-tree whose embedded sub-pp-trees have universes $\{\bpt\} = |\+S_0| \prec \dots\prec|\+S_\ell| = |\T|$ for some $\ell \geq 0$, where we use the symbol $\prec$ for successors in the inclusion order. Then let us see that for each $0\leq j < \ell$, $|\+S_{j+1}| = |\+S_j| \union \{u_{j+1}\}$ for some $u_{j+1} \in |\T|$. Indeed, suppose that for some $j$ there exist $u', u'' \in |\+S_{j+1}|\setminus|\+S_j|$ with $u' \neq u''$. Since $u'$ is an $E$-successor to some point in $|\+S_j|$, $|\+S_j| \union \{u'\}$ is the universe of yet another embedded sub-pp-tree $\+S'$ such that $|\+S_j| \prec |\+S'| < |\+S_{j+1}|$, a contradiction. We conclude that the underlying $E$-tree of $\T$ is $u\prec u_1\prec\dots\prec u_\ell$, hence it is a pp-path.
\end{proof}

\begin{remark}\label{rem:paths_in_a_pptree}
    Given a pp-tree $\T$ we may refer to a path embedding $m: \+P \to \T$ as a \define{path in $\T$}. If we instead think of path embeddings up to isomorphism, i.e.\ as elements of the poset $\bb P \T$, then by taking an appropriate representative we may identify a path in $\+ T$ with an embedded sub-pp-path of $\T$, i.e.\ an embedded substructure of the form $\T_v$ for some $v \in |\T|$. In other words, paths in a pp-tree $\T$ are exactly the embedded sub-pp-paths, which are all of the form $\T_v$ for some $v$. We may also identify paths in $\T$ with their underlying $E$-chains.
    
    Notice that the assignment $v \mapsto \T_{v}$ induces a bijection between $|\T|$ and $\bb P\T$. If $\T$ is of finite height $\ell$, then these two sets are also in bijection with $\Ck\+T$ for all $k \geq \ell$.
\end{remark}

The following are technical conditions on paths that make them well behaved. Note that the definition of connected object given here, reproduced from~\cite{abramsky2021arboreal}, is not standard.

\begin{definition}\label{def:path_cat_and_connected_object}
    A category $\cat{C}$ equipped with a stable proper factorisation system $(\Q, \M)$ is a \define{path category} if the following conditions hold: (1) $\cat{C}$ has all coproducts of small families of paths, and (2) for any paths $P, Q, R$, if a composite $P \to Q \to R$ is a $\mathscr{Q}$-morphism, then so is $P \to Q$.
    If $\cat{C}$ is a path category, let $\cat{C}_p$ denote the full subcategory of $\cat{C}$ spanned by paths.

    An object $X$ in a path category is \define{connected} if for all small families of \emphat{paths} $(P_i)_{i \in I}$ in $\cat{C}_p$, any morphism $X \to \coprod_{i \in I} P_i$ factors through some coproduct inclusion $P_j \to \coprod_{i \in I} P_i$.
\end{definition}

\begin{proposition}\label{prop:morphisms_of_pptrees_preserve_height}
    Let $f: \T \to \T'$ be a morphism in $\ncat{ppTree}$. Then for all $v \in |\T|$, the height of $v$ is equal to the height of $f(v)$. In particular, morphisms out of a pp-tree are injective.
\end{proposition}
\begin{proof}
    The first claim is proven by an easy induction on the height of points of a pp-tree. The second claim follows since in a pp-path there is at most one point of each height.
\end{proof}

\begin{propositionrep}\label{prop:pptree_is_path_cat}
    $\ncat{ppTree}$ together with the factorisation system $(\Q,\M)$,where $\Q$ is the class of pointed surjective homomorphisms and $\M$ is the class of pointed relational embeddings, is a path category. Moreover, every pp-path is connected.
\end{propositionrep}
\begin{proof}
    By Proposition~\ref{prop:fact_system}, $(\Q, \M)$ as above is a stable proper factorisation system for $\ncat{ppTree}$. To see that $\ncat{ppTree}$ has all coproducts of small families of paths, we show that \emphat{a fortiori} all small coproducts exist in $\ncat{ppTree}$. Indeed, by Prop.~\ref{prop:caract_colims_in_structpointed}, it is clear that if $(\A_\alpha, a_\alpha)_{\alpha\in\Lambda}$ is a collection of structures in $\Stp$ for some set $\Lambda$ which happen to be pp-trees, then $\coprod_\alpha (\A_\alpha, a_\alpha)$ is too: it is an $E$-tree obtained by glueing together all the $E$-trees at their root, and if $([(\alpha',a_1)],\dots,[(\alpha',a_r)]) \in R^{\coprod_\alpha \A_\alpha} \coloneqq \union_\alpha \iota_\alpha (R^{\A_\alpha})$ for some $R \in \bsigma$ and some $\alpha'\in\Lambda$, then it must be the case that $(a_1,\dots,a_r) \in R^{\A_{\alpha'}}$ and hence $([(\alpha',a_1)],\dots,[(\alpha',a_r)])$ is an $E$-chain in $\coprod_\alpha \A_\alpha$. Notice that this also works for $\Lambda = \varnothing$ since the initial object of $\Stp$ is a pp-tree. Finally, because inclusions of categories reflect colimits, this construction must also be the coproduct of $(\A_\alpha, a_\alpha)_{\alpha\in\Lambda}$ in $\ncat{ppTree}$.

    Let us now see that if $\+P, \+Q$ and $\+R$ are pp-paths and $f: \+P \to \+Q$ and $g: \+Q \to \+R$ are morphisms in $\ncat{ppTree}$ such that $g \circ f$ is surjective, then $f$ is surjective. $g \circ f$ being surjective implies that $g$ is surjective as well, hence by Prop.~\ref{prop:morphisms_of_pptrees_preserve_height} it follows that both $g \circ f$ and $g$ are bijective (although not necessarily isomorphisms). Hence $f$ is bijective; in particular it is surjective.
    
    This establishes that $\ncat{ppTree}$ is a path category. That every pp-path is connected is straightforward from Prop.~\ref{prop:caract_colims_in_structpointed} and the proof of Prop.~\ref{prop:pptree_is_path_cat}, since there it is shown that a coproduct of pp-paths $\coprod_{\alpha \in \Lambda} \+P_\alpha$ coincides with the coproduct in $\Stp$, hence it is a tree whose only branching node (point with more than one successor) is the root node. Thus given a pp-path $\+Q$, the image of a morphism $\+Q \to \coprod_{\alpha \in \Lambda} \+P_\alpha$ lies entirely in $\iota_\alpha(\+P_\alpha)$ for some $\alpha$.
\end{proof}
From now on we will refer to $\ncat{ppTree}$ as a path category, leaving the factorisation system implicit.

\begin{definition}
Let $\cat{T}$ denote the category $\ncat{ppTree}$ for the particular choice $\sigma = \{\acc\}$. We refer to $\cat{T}$ as the \define{category of trees}.
\end{definition}
Notice that this notion of tree homomorphism, as a morphism in $\cat{T}$, preserves heights (by Prop.~\ref{prop:morphisms_of_pptrees_preserve_height}), and that contrary to~\cite{abramsky2021arboreal} we do not allow the empty tree.

We quote the following result in order to convey the meaning of the axioms for path categories just given. Although we do not need the empty tree, since all our structures are pointed and hence their associated trees are non-empty, for general path categories we must modify $\cat{T}$ so that it contains the empty tree.

\begin{theorem}[\cite{abramsky2021arboreal}, Theorem 14]\label{thm:functor_into_trees}
Let $\cat{C}$ be a path category. Then the assignment $X \mapsto \bb P X$ induces a functor $\bb P: \cat{C} \to \cat{T}$.
\end{theorem}

In view of this result, we may reformulate the intuition behind arboreal categories by saying that an arboreal category is a path category for which the functor $\bb P$ does not lose any relevant information, i.e.\ $X$ is determined by $\bb P X$. This is captured by the following definition.

\begin{definition}\label{def:path_generated_object}
    Given an object $X$ in a path category $\cat{C}$, consider the diagram consisting of all path embeddings with codomain $X$, together with morphisms between their domains (necessarily path embeddings as well) making the triangles
    \[\begin{tikzcd}
        P && Q \\
        & X
        \arrow[curve={height=6pt}, tail, from=1-1, to=2-2]
        \arrow[curve={height=-6pt}, tail, from=1-3, to=2-2]
        \arrow[tail, from=1-1, to=1-3]
    \end{tikzcd}\]
    commute.
    More precisely, we consider this diagram as a cocone over $X$ in the following way. Let $(\cat{C}/X)_{p}$ be the full subcategory of $\cat{C}/X$
    spanned by the path embeddings%
    \footnote{Given a category $\cat{C}$ and an object $X \in \cat{C}$, an object in the \define{category over $X$}, $\cat{C}/X$, is a pair $(X', f)$ where of $X' \in \cat{C}$ and $f: X' \to X$, while a morphism $(X', f) \to (X'', g)$ is a morphism $h: X' \to X''$ such that $g \circ h = f$.}
    and let $\Pi_X: (\cat{C}/X)_p \to \cat{C}$ be the functor sending each path embedding to its domain. Then the cocone in question is the cocone $\lambda: \Pi_X \Rightarrow X$ whose leg $\lambda_{(P, f)}$ is $f$. We say that $X$ is \define{path-generated} if $\lambda$ is a colimit cocone in $\cat{C}$.
\end{definition}

\begin{definition}\label{def:arboreal_cat}
    An \define{arboreal category} is a path category $\cat{C}$ such that (1) every object of $\cat{C}$ is path-generated, and (2) every path in $\cat{C}$ is connected.
\end{definition}

Definition~\ref{def:path_generated_object} relates to the poset $\bb P X$ for an object $X$ in an arboreal category $\cat{C}$ as follows. Let $D: \cat{J} \to \cat{C}$ be some diagram in $\cat{C}$ and let $\cat{J}'$ be a full subcategory of $\cat{J}$. Suppose that for all $j \in \cat{J}$ there exists some morphism $a: j \to j'$ in $\cat{J}$ with $j' \in \cat{J}'$ and such that $D(a): D(j) \to D(j')$ is an isomorphism. Then the colimit of $D$ coincides with the colimit of its restriction $D|_{\cat{J}'}: \cat{J}' \to \cat{C}$ (assuming they exist). In our case this means that, if we think of $\bb P X$ as a full subcategory of $(\cat{C}/X)_p$ by choosing a representative for each equivalence class in $\bb P X$, then the diagram $\Pi_X$ can be restricted to $\bb P X$ without changing the resulting colimit. We may say succinctly, then, that $X$ is path-generated if and only if it is the colimit of its paths, $X \cong \colim_{P \in \bb P X} P$. It is in this sense that a path-generated object $X$ is determined by $\bb P X$.

\begin{theoremrep}\label{thm:is_arboreal}
    $\ncat{ppTree}$ is an arboreal category.
\end{theoremrep}
\begin{proof}
    By Proposition~\ref{prop:pptree_is_path_cat}, it is enough to show that every object of $\ncat{ppTree}$ is path-generated.
    Given $\T \in \ncat{ppTree}$, consider the diagram $\Pi_{\T}: (\ncat{ppTree}/\T)_p \to \ncat{ppTree}$ of all path embeddings into $\T$ and let $\lambda: \Pi_{\T} \Rightarrow \T$ be the canonical cocone given by $\lambda_{(\+P, f)} \coloneqq f$. In order to prove that this is a colimit cocone in $\ncat{ppTree}$, by taking representatives, we can assume that all paths $\+P$ in the diagram are embedded substructures of $\T$, and thus the legs of this cocone are the corresponding inclusions $i_{\+P}$. Now consider a generic cocone $\eta: \Pi_{\T} \Rightarrow \T'$ in $\ncat{ppTree}$ over the same diagram, this time with nadir $\T'$. We must show that there exists a unique homomorphism $\phi: \T \to \T'$ such that it commutes with the legs of the cocones, as indicated by the following diagram.
    \[\begin{tikzcd}
        \+P && \+Q \\
        & \T \\
        & \T'
        \arrow["{i_{\+P}}", curve={height=6pt}, tail, from=1-1, to=2-2]
        \arrow["{i_Q}"', curve={height=-6pt}, tail, from=1-3, to=2-2]
        \arrow[tail, from=1-1, to=1-3]
        \arrow["{\eta_{\+P}}"', curve={height=12pt}, from=1-1, to=3-2]
        \arrow["{\eta_Q}", curve={height=-12pt}, from=1-3, to=3-2]
        \arrow["{\exists! \phi}"', dashed, from=2-2, to=3-2]
    \end{tikzcd}\]

    Let $v \in |\T|$. Then if $\phi$ exists, by the commutativity condition $\phi \circ i_{\+P} = \eta_{\+P}$ for the case $\+P = \T_{v}$ (see Remark~\ref{rem:paths_in_a_pptree}) it must be the case that $\phi(v) = \eta_{\T_{v}}(v)$. Moreover this choice of $\phi$ makes the corresponding triangles commute \emphat{for all} $\+P$, since
    for any pp-path $\+P$ embedded in $\T$ and containing $v$, it must be the case that $|\T_{v}| \subseteq |\+P|$, thus $\eta_{\+P}$ and $\eta_{\T_{v}}$ agree on $v$ since $\eta$ is a cocone.
    
    Now let us see that the function $\phi$ defines a morphism in $\ncat{ppTree}$. By the argument above it is immediate that $\phi$ preserves the root of $\T$. To see that $\phi$ preserves relations, suppose $s \in R^\T$ for some $R \in \sigma$. Then $s$ is a suffix of $\T_v$ for $v = \epsilon(s)$ (see Remark~\ref{rem:notation_for_pptrees}), hence $\phi(s) = \eta_{\T_v}(s) \in R^{\T'}$ since $\eta_{\T_v}$ is a homomorphism.
\end{proof}

\begin{remark}
    We have seen that path objects in $\ncat{ppTree}$ are precisely the pp-paths. In contrast, the notion is not well-behaved in the category $\Structpointed$ itself. Indeed, any pointed structure with at least three distinct points will fail to be a path since its poset of embedded substructures is not totally ordered. Similarly, no structure with more than one point is a path object in $\Struct$.
\end{remark}

We still have not proved that $\EM(\Ck)$ is arboreal for each $k$, but this follows from the fact that all of these categories sit nicely inside $\ncat{ppTree}$, as we will now show. Recall that given a path category $\cat{C}$ we denote its full subcategory of paths by $\cat{C}_p$.

\begin{definition}\label{def:res_indexed_arboreal_cat_and_cover}
    A collection $\{\cat{C}_k\}_{k \in \N}$ of categories is a \define{resource-indexed arboreal category} if there exists an arboreal category $\cat{C}$ together with a chain of full subcategories of $\cat{C}_p$ $\cat{C}_p^0 \into\dots\into\cat{C}_p^k\into\dots\into\cat{C}_p$
    such that the following hold:
    \begin{enumerate}
        \item each $\cat{C}_p^k$ is closed under $\M$-morphisms, i.e.\ for all $y \in \cat{C}_p^k$ and $\M$-morphisms $x \emb y$, $x \in \cat{C}_p^k$; and
        \item for every $k$,
        $\cat{C}_k$ is the full subcategory of $\cat{C}$ spanned by the \define{$k$-path-generated objects}, that is to say the objects whose cocone of path embeddings with domain in $\cat{C}_p^k$ is a colimit cocone in $\cat{C}$.
    \end{enumerate}

    Moreover, if $\{\cat{C}_k\}_{k\in\N}$ is a resource-indexed arboreal category, a \define{resource-indexed arboreal cover} of a category $\mathscr{E}$ by $\{\cat{C}_k\}_{k\in\N}$ is a $k$-indexed family of adjunctions
    \begin{equation*}
    \begin{tikzcd}
    \cat{C}_k \arrow[r, bend left=25, ""{name=U, below}, "L_k"{above}]
    \arrow[r, leftarrow, bend right=25, ""{name=D}, "R_k"{below}]
    & \mathscr{E}
    \arrow[phantom, "\textnormal{\footnotesize{$\bot$}}", from=U, to=D] 
    \end{tikzcd}
    \end{equation*}
    which are \define{comonadic}, that is to say $\cat{C}_k$ is isomorphic to $\EM(L_k R_k)$ for all $k$.\footnote{Since any adjunction induces a comonad, a $k$-indexed family of adjunctions induces a $k$-indexed family of comonads $\{(G_k, \delta^k, \epsilon^k)\}_{k \in \N}$ where $G_k = L_k R_k$, $\delta^k = L_k \eta_k R_k$ and $\eta^k, \epsilon^k$ are the unit and counit of each adjunction, respectively.}
\end{definition}

\begin{theoremrep}\label{thm:arborealcover}
    $\{\EM(\Ck)\}_{k\in\N}$ is a resource-indexed arboreal category. Therefore, the family of comonadic adjunctions
\[\begin{tikzcd}[ampersand replacement=\&]
	{\EM(\Ck)} \&\& \Stp
	\arrow[""{name=0, anchor=center, inner sep=0}, "{U_k}", shift left=3, from=1-1, to=1-3]
	\arrow[""{name=1, anchor=center, inner sep=0}, "{F_k}", shift left=3, from=1-3, to=1-1]
	\arrow["\dashv"{anchor=center, rotate=-90}, draw=none, from=0, to=1]
\end{tikzcd}\]
    where $U_k$ is the inclusion of categories and $F_k$ is the co-restriction of $\Ck$ to its image, constitutes a resource-indexed arboreal cover of $\Structpointed$.
\end{theoremrep}
\begin{proof}
    By definition, the adjunctions involved are comonadic and the induced comonads on $\Structpointed$ are precisely $\{\Ck\}_k$. Let us then prove that $\{\EM(\Ck)\}_k$ is a resource-indexed arboreal category.

    By, Theorem~\ref{thm:is_arboreal}, $\ncat{ppTree}$ is arboreal, and by Proposition~\ref{prop:aths_are_pp-paths}, $\ncat{ppTree}_p$ is the full subcategory of pp-paths.
    Let $\ncat{ppTree}_p^k$ denote the full subcategory of $\ncat{ppTree}_p$ spanned by pp-paths of length at most $k$, which is clearly closed under embeddings and a full subcategory of $\ncat{ppTree}_p^{k+1}$ for all $k$.
    Let us now see that $\EM(\Ck)$ is precisely the full subcategory of $\ncat{ppTree}$ spanned by $k$-path-generated objects.

    To this end, it is enough to show that in the proof of Theorem~\ref{thm:is_arboreal},
    if the pp-tree $\T \in \ncat{ppTree}$ has height at most $k$, and thus belongs to $\EM(\Ck)$, then all the pp-paths in the diagram have length at most $k$ and hence are actually objects of $\ncat{ppTree}_p^k$. Indeed, by Prop.~\ref{prop:morphisms_of_pptrees_preserve_height}, there are no morphisms from a pp-path of height greater than $k$ into a pp-tree of height at most $k$. Thus, the cocone $\lambda : \Pi_\T \nat \T$
    which was shown above to be a colimit cocone in $\ncat{ppTree}$ coincides with the canonical cocone of path embeddings with codomain $\T$ and domains in $\ncat{ppTree}_p^k$. That is to say, if $\T \in \EM(\Ck)$ then $\T$ is $k$-path-generated.
\end{proof}

\begin{corollary}\label{coro:all_EM_are_arboreal}
    For all $k \geq 0$, $\EM(\Ck)$ is arboreal.
\end{corollary}
\begin{proof}
    Apply \cite[Prop. 42]{abramsky2021arboreal}.
\end{proof}

Recall that our goal is to show how the \logic comonads capture resource-bounded bisimilarity. The link between game comonads and bisimulation games is established through the fact that arboreal categories admit an intrinsic notion of back-and-forth or bisimulation game~\cite{abramsky2021arboreal} whose definition we reproduce in the Appendix for ease of reference (see Definition~\ref{def:abstract_backandforth_game}). 
By Corollary~\ref{coro:all_EM_are_arboreal}, there is a back-and-forth-game associated to each category $\EM(\Ck)$. We now establish their equivalence to the \logic $k$-bisimulation games of Definition~\ref{def:ppml_games}.

\begin{toappendix}
\begin{definition}\label{def:abstract_backandforth_game}
    Let $\cat{C}$ be an arboreal category and let $X,Y \in \cat{C}$. The \define{back-and-forth game} $\mathscr{G}(X,Y)$ between $X$ and $Y$ is defined as follows.
    The game is played between two players, called Spoiler and Duplicator. 
    The state of the game at round $\ell$ is given by a pair of equivalence classes of path embeddings $([m],[n]) \in \mathbb{P}X \times \mathbb{P}Y$.
    We say that $([m],[n])$ satisfies the \define{winning condition} if the domains $\dom(m)$ and $\dom(n)$ are isomorphic (notice that this is well defined on equivalence classes).

    The initial position is $(\bot_X,\bot_Y)$ where $\bot_X\colon P \to X$ and $\bot_Y\colon Q \to Y$ are the roots of $\mathbb{P}{X}$ and $\mathbb{P}{Y}$, respectively (which exist by Theorem~\ref{thm:functor_into_trees}). If $P$ and $Q$ are not isomorphic, then Duplicator loses the game. Otherwise, assuming position $([m],[n])\in\mathbb{P}{X}\times\mathbb{P}{Y}$ is reached after $\ell \geq 0$ rounds, position $\ell + 1$ is determined as follows: either Spoiler chooses some $[m']\succ [m]$ and Duplicator must respond with some $[n']\succ [n]$, or Spoiler chooses some $[n'']\succ [n]$ and Duplicator must respond with $[m'']\succ [m]$.
    If Spoiler cannot make such a choice, then Duplicator wins the game immediately.
    We say that Duplicator wins the round $\ell + 1$ if Duplicator is able to respond with a move which is valid according to the preceding description and which moreover makes the resulting state satisfy the winning condition. Otherwise, the game ends and Duplicator loses immediately. A \define{winning strategy} for Duplicator consists in a choice of a response that makes Duplicator win the round $\ell + 1$ for every move that Spoiler may make after any number of rounds $\ell > 0$.
\end{definition}
\end{toappendix}

\begin{propositionrep}\label{prop:games_coincide}
    Given $\Ap, \Bp \in \Stp$ and $k>0$, the $k$-round bisimulation game for \logic played between $\Ap$ and $\Bp$, $\G_k(\Ap, \Bp)$, is equivalent to the back-and-forth game $\mathscr{G}(\Ck \Ap, \Ck \Bp)$ in the arboreal category $\EM(\Ck)$ played between $\Ck \Ap$ and $\Ck \Bp$. More precisely, there exists a winning strategy for Duplicator in one of the two games if and only if there exists one in the other.\footnote{Although the proof given amounts to a certain equivalence between the games themselves, our precise statement is given in terms of existence of winning strategies since this is all we need and we do not present a formal notion of equivalence between games.}
\end{propositionrep}
\begin{proof}
    A position in the abstract back-and-forth game $\mathscr{G}(\Ck \Ap, \Ck \Bp)$ is a pair $([m],[n]) \in \bb P \Ck \Ap \times \bb P \Ck \Bp$. Following Remark~\ref{rem:paths_in_a_pptree}, this can be seen as a pair $(s, t) \in \Ck \Ap \times \Ck \Bp$, i.e.\ a position in the \logic game $\G_k(\Ap, \Bp)$, by taking $s$ and $t$ to be the maximal points in the images of $[m]$ and $[n]$, respectively. This correspondence maps the starting position $(\bot_{\Ck \Ap}, \bot_{\Ck \Bp})$ of $\mathscr{G}(\Ck \Ap, \Ck \Bp)$ to the starting position $([\bpa], [\bpb])$ of $\G_k(\Ap, \Bp)$. Moreover, the procedures in each game by which a new position is obtained from the previous one coincide under this identification. With respect to the ending of the game, notice that if Duplicator has not lost in $\mathscr{G}(\Ck \Ap, \Ck \Bp)$ after having played the $k$-th round, it wins the game, since Spoiler will not have any valid move to make after playing the $k$-th round (because $\Ck \Ap$ and $\Ck \Bp$ are trees of height at most $k$).
    We conclude that $\mathscr{G}(\Ck \Ap, \Ck \Bp)$ can be regarded as being the same as $\G_k(\Ap, \Bp)$, except for a possibly different winning condition for Duplicator. We focus now on showing that the two winning conditions are indeed equivalent.
    
    The winning condition of $\mathscr{G}(\Ck \Ap, \Ck \Bp)$, namely that $\dom([m]) \cong \dom([n])$, can be restated through the above mentioned correspondence as consisting of positions $(s, t)$ that satisfy the condition
    \begin{equation}
        (\Ck\A)_s \cong (\Ck\B)_t \tag{$*$}
    \end{equation}
    where we have used the notation from Remark~\ref{rem:notation_for_pptrees}.\footnote{Notice that there is at most one possible isomorphism between any pair of pp-paths, in particular between $(\Ck\A)_s$ and $(\Ck\B)_t$.}
    Thus we must show that for any position $(s, t)$ that is reachable during the game,
    (\textbf{$*$}) holds if and only if
    \begin{equation}
        \A,s\models R \iff \B,t\models R \text{ for all $R \in \bsigma$.} \tag{$**$}
    \end{equation}
    We prove this by induction.
    The base case $(s, t) = ([a], [b])$ is immediate. Now assume the result holds for any position $(s, t) = ([a_0,\dots,a_\ell], [b_0,\dots,b_\ell])$ obtained during round $\ell$. If this position does not satisfy any of the two winning conditions (and therefore, by inductive hypothesis, none of them), then Duplicator loses the game and there is no next position to analyse. Otherwise, Duplicator wins that round and the game progresses to round $\ell + 1$. If $\ell + 1 > k$, then there is no next position, either, since Spoiler will not be able to make a move. Hence, suppose that $\ell \leq k$ and both Spoiler and Duplicator have valid moves to make, resulting in position $(s', t') = (s.a_{\ell+1}, t.b_{\ell+1})$. For each $j \in \{0,\dots,\ell+1\}$, let $s_j \coloneqq [a_0,\dots,a_j]$ and $t_j \coloneqq [b_0,\dots,b_j]$.

    First let us see that condition ($*$) implies condition ($**$). For this, the inductive hypothesis is not needed. Indeed, assume that ($*$) holds for $(s', t')$ and let $R \in \bsigma$ have arity $r \leq \ell + 2$ (since otherwise condition ($**$) holds vacuously). Then
    \begin{align*}
        \A,s' \models R &\iff \last_{r}(s') \in R^\A \\
        &\iff [s_{\ell+2-r},\dots,s'] \in R^{\Ck \A} \tag{Def. $R^{\Ck \A}$}\\
            &\iff [t_{\ell+2-r},\dots,t'] \in R^{\Ck \B} \tag{$*$} \\
            &\iff \last_{r}(t') \in R^\B \tag{Def. $R^{\Ck \B}$} \\
            &\iff \B, t' \models R.
    \end{align*}
    
    In the other direction,
    by contrapositive, assume that condition ($*$) does \emphat{not} hold for $(s', t')$. Then there exist $R \in \bsigma$ of arity $r$ and $j_1,\dots,j_{r} \in \{0,\dots,\ell+1\}$ such that $(s_{j_1},\dots,s_{j_{r}}) \in R^{\Ck \A}$ and $(t_{j_1},\dots,t_{j_{r}}) \not\in R^{\Ck \B}$, or viceversa. In particular, $r \leq \ell+2$. Given the definitions of $R^{\Ck \A}$ and $R^{\Ck \B}$, it must be that $j_1,\dots,j_{r}$ are consecutive and non-repeating, and since the previous round had been won by Duplicator, condition ($*$) holds for $(s, t)$, which implies that $[j_1,\dots,j_{r}] = [\ell+2-r, \dots, \ell+1]$. Thus in particular either $\last_{r}(s') \in R^\A$ and $\last_{r}(t') \not\in R^\B$ or viceversa.
\end{proof}

The fact that the \logic $k$-bisimulation game coincides with the abstract back-and-forth game in the arboreal category $\EM(\Ck)$ establishes a connection between \logic and an abstract notion of functional bisimulation internal to any arboreal category, namely open pathwise embeddings~\cite{abramsky2021arboreal}.

\begin{definition}\label{def:open_pe}
    Let $\cat{C}$ be a category equipped with a stable, proper factorisation system $(\Q, \M)$. A morphism $f: X \to Y$ in $\cat{C}$ is said to be a \define{pathwise embedding} if for all path embeddings $e: P \to X$, $f \circ e$ is also a (path) embedding. $f$ is said to be \define{open} iff given any commutative square
\[\begin{tikzcd}[ampersand replacement=\&]
	P \& Q \\
	X \& Y
	\arrow["m", tail, from=1-1, to=1-2]
	\arrow["e"', tail, from=1-1, to=2-1]
	\arrow["d"', dashed, from=1-2, to=2-1]
	\arrow["{e'}", tail, from=1-2, to=2-2]
	\arrow["f"', from=2-1, to=2-2]
\end{tikzcd}\]
    where $P, Q$ are paths and $e,e',m$ are embeddings, there exists a diagonal morphism $d: Q \to X$ making the two triangles commute.
\end{definition}

Intuitively, a morphism $f: \T \to \T'$ in $\ncat{ppTree}$ is a pathwise embedding if it preserves embedded sub-pp-paths. As for the openness condition, the commutativity of the square above encodes the possibility of taking a path of shape $P$ in $X$, pushing it forward along $f$ to a path in $Y$, and then extending it to a longer path of shape $Q$ in $Y$. Then the existence of the diagonal filler $d$ amounts to a lifting of this extension back in $X$. The reader familiar with bounded morphisms in Modal Logic \cite[Def. 2.10]{blackburn2001modal} may be able to recognise them in this definition. Indeed, we now define an appropriate notion of bounded morphism for \logic which generalises that of \bml and show that it coincides with open pathwise embeddings in $\ncat{ppTree}$.

\begin{definition}\label{def:bounded_morphism}
    We say that a morphism $f: \A \to \B$ between non-pointed $\sigma$-structures is \define{bounded} iff the following hold:
    \begin{enumerate}
        \item for all $E$-chains $s$ in $\A$, $\A, s \models R$ if and only if $\B, f(s)\models R$ for all $R \in \bsigma$; and
        \item for all $a \in |\A|$ and $b \in |\B|$, if $f(a) \prec b$ then there exists some $a' \in |\A|$ such that $a \prec a'$ and $f(a') = b$.
    \end{enumerate}
    The same conditions define boundedness for morphisms of pointed structures. We refer to (1) as the \define{harmony condition} and to (2) as the \define{back condition}.

\end{definition}

\begin{proposition}\label{prop:open_PE}
    A morphism $f: \T \to \T'$ in $\ncat{ppTree}$ is an open pathwise embedding if and only if it is bounded.
\end{proposition}
\begin{proof}
    Let $f: \T \to \T'$ be an open pathwise embedding and let $s$ be an $E$-chain in $\T$. Let $v = \epsilon(s)$ and consider the path embedding $e: \T_v \emb \T$. Since $f$ is a pathwise embedding, $f \circ e = f|_{\T_v}$ is an embedding, hence for any $R \in \bsigma$,
    $\T, s \models \T_v, s \models R \iff \T', f(s) \models R$.
    This establishes the harmony condition.
    For the back condition, let $v \in |\T|$ and $v' \in |{\T'}|$ such that $f(v)\prec v'$, and consider the path embeddings $e: \T_v \emb \T$ and $e': \+Q \emb \T'$, where $\+Q \coloneqq f(\T_v) \cup v'$ is considered as an embedded sub-pp-path of $\T'$. By definition, $f$ restricts and corestricts to the function $f|_{\T_v}^{\+Q} = f\circ e|^{\+Q}: \T_v \to {\+Q}$ which is an embedding since $f$ is a pathwise embedding. We thus have a commutative square
\[\begin{tikzcd}
	{\T_v} & {f(\T_v)\cup\{v'\}} \\
	\T & {\T'}
	\arrow["{f\circ e|^{\+Q}}", tail, from=1-1, to=1-2]
	\arrow["e"', tail, from=1-1, to=2-1]
	\arrow["d"', dashed, from=1-2, to=2-1]
	\arrow["{e'}", tail, from=1-2, to=2-2]
	\arrow["f"', from=2-1, to=2-2]
\end{tikzcd}\]
    which, since $f$ is open, induces a diagonal filler $d: f(\T_v) \cup v' \to \T$.
    Let $w \coloneqq d(v')$. The commutativity of the lower triangle means that $f(w) = v'$, while by the commutativity of the upper triangle means that $d(f\circ e|^Q(v)) = v$. By assumption, $f \circ e|^Q(v) \prec v'$, hence since $d$ is a homomorphism, $v \prec w$. Hence $f$ satisfies the back condition.

    In the other direction, assume that $f$ is a bounded morphism. Let $e: \+P \emb \T$ be a path embedding. Without loss of generality we may assume $\+P = \T_v$ for some $v \in |\T|$. By Prop.~\ref{prop:morphisms_of_pptrees_preserve_height} $f \circ e$ is an injective homomorphism; let us see that it is strong.
    Let $s = [u_1,\dots,u_r] \in |\T_v|^+ \subseteq |\T|^+$ be any sequence (not necessarily a chain) and assume that $\T', f(s) \models R$ for some $R \in \bsigma$. Without loss of generality we may assume that $r$ coincides with the arity of $R$. Then, since ${\T'}$ is a pp-tree, this implies that $f(s)$ is an $E$-chain. Using the back condition, since $f(u_1) \prec f(u_2)$, let $u_2' \succ u_1$ be such that $f(u_2') = f(u_2)$. But since $f\circ e$ is injective, we have $u_2'=u_2$. In this way, we conclude by induction on prefixes of $s$ that $s$ is a chain in $\T$. Hence $\T, s\models R$ by the harmony condition. This establishes that $f$ is a pathwise embedding.

    Finally, let us see that $f$ is open. Let $\+P, \+Q$ be paths and let $e, e', m$ be embeddings such that
\[\begin{tikzcd}
	{\+P} & {\+Q} \\
	\T & {\T'}
	\arrow["m", tail, from=1-1, to=1-2]
	\arrow["e"', tail, from=1-1, to=2-1]
	\arrow["{e'}", tail, from=1-2, to=2-2]
	\arrow["f"', from=2-1, to=2-2]
\end{tikzcd}\]
    commutes. In order to prove that there exists a diagonal filler $d: \+Q \to \T$, we proceed by induction in the parameter $\#|\+Q| - \#|\+P|$, i.e.\ the height difference between $\+Q$ and $\+P$. The base case $\#|\+Q| - \#|\+P| = 0$ is trivial since an embedding between $\sigma$-structures of the same size is an isomorphism, hence we can take $d = e \circ m^{-1}$.

    For the inductive step, let $\tilde{q}$ be the unique point of height $\#|\+Q| - 2$, i.e.\ the unique predecessor of the unique leaf of $\+Q$. Then since $\#|\+Q| - \#|\+P| > 0$ it is immediate that $m$ factors as the composite $m''\circ m'$ of the two path embeddings $m' = m|^{\+Q_{\tilde{q}}} : \+P \emb \+Q_{\tilde{q}}$ and $m'': \+Q_{\tilde{q}} \emb \+Q$.
    Now consider the following diagram.
\[\begin{tikzcd}
	{\+P} & {\+Q_{\tilde{q}}} & {\+Q} \\
	\T & {\T'} & {\T'}
	\arrow["{m'}", tail, from=1-1, to=1-2]
	\arrow["e"', tail, from=1-1, to=2-1]
	\arrow["{m ''}", tail, from=1-2, to=1-3]
	\arrow["{d'}"', dashed, from=1-2, to=2-1]
	\arrow["{e'\circ m''}"{pos=0.8}, tail, from=1-2, to=2-2]
	\arrow["d"{description, pos=0.7}, dashed, from=1-3, to=2-1]
	\arrow["{e'}", tail, from=1-3, to=2-3]
	\arrow["f"', from=2-1, to=2-2]
	\arrow[Rightarrow, no head, from=2-2, to=2-3]
\end{tikzcd}\]
    By inductive hypothesis there exists a diagonal filler $d': \+Q_{\tilde{q}} \to \T$ making the two triangles inside the square on the left commute. Since the diagonal filler is an embedding, we apply openness of $f$ again to obtain a second diagonal filler $d: \+Q \to \T$. The commutativity of the triangle below $d$ is immediate while the commutativity of the triangle above $d$ follows from the commutativity of the two triangles sharing $d'$ as one of their sides.
\end{proof}

Finally, putting everything together we obtain our desired characterisation of $k$-bisimilarity for \logic.

\begin{theorem}\label{thm:full_logical_equivalence}
    Two structures $(\A, a), (\B, b) \in \Structpointed$ are $k$-bisimilar iff there exists a span of bounded morphisms $\T \to \Ck(\A, a)$ and $\T \to \Ck(\B, b)$ with some pp-tree $\T$ of height at most $k$ as common domain.
\end{theorem}
\begin{proof}
    By \cite[Proposition 46]{abramsky2021arboreal}, since $\Structpointed$ has binary products
    and considering the resource-indexed arboreal cover given by the comonadic adjunctions of $\{\Ck\}_k$, there exists a span of open pathwise embeddings in $\EM(\Ck)$ with codomains $\Ck \Ap$ and $\Ck \Ap$ if and only if Duplicator has a winning strategy in the game $\mathscr{G}(\Ck \Ap, \Ck \Bp)$.
    
    By Prop.~\ref{prop:games_coincide} this game coincides with $\mathscr{G}(\Ck \Ap, \Ck \Bp)$, and thus by Theorem~\ref{thm:bisimilarity_game}, there exists such a span of open pathwise embeddings iff $\Ap \bisim_k \Bp$. Finally, open pathwise embeddings in $\EM(\Ck)$ coincide with bounded morphisms by Prop.~\ref{prop:open_PE}.
\end{proof}

\begin{figure}[h]
\centering
\includegraphics[scale=0.3]{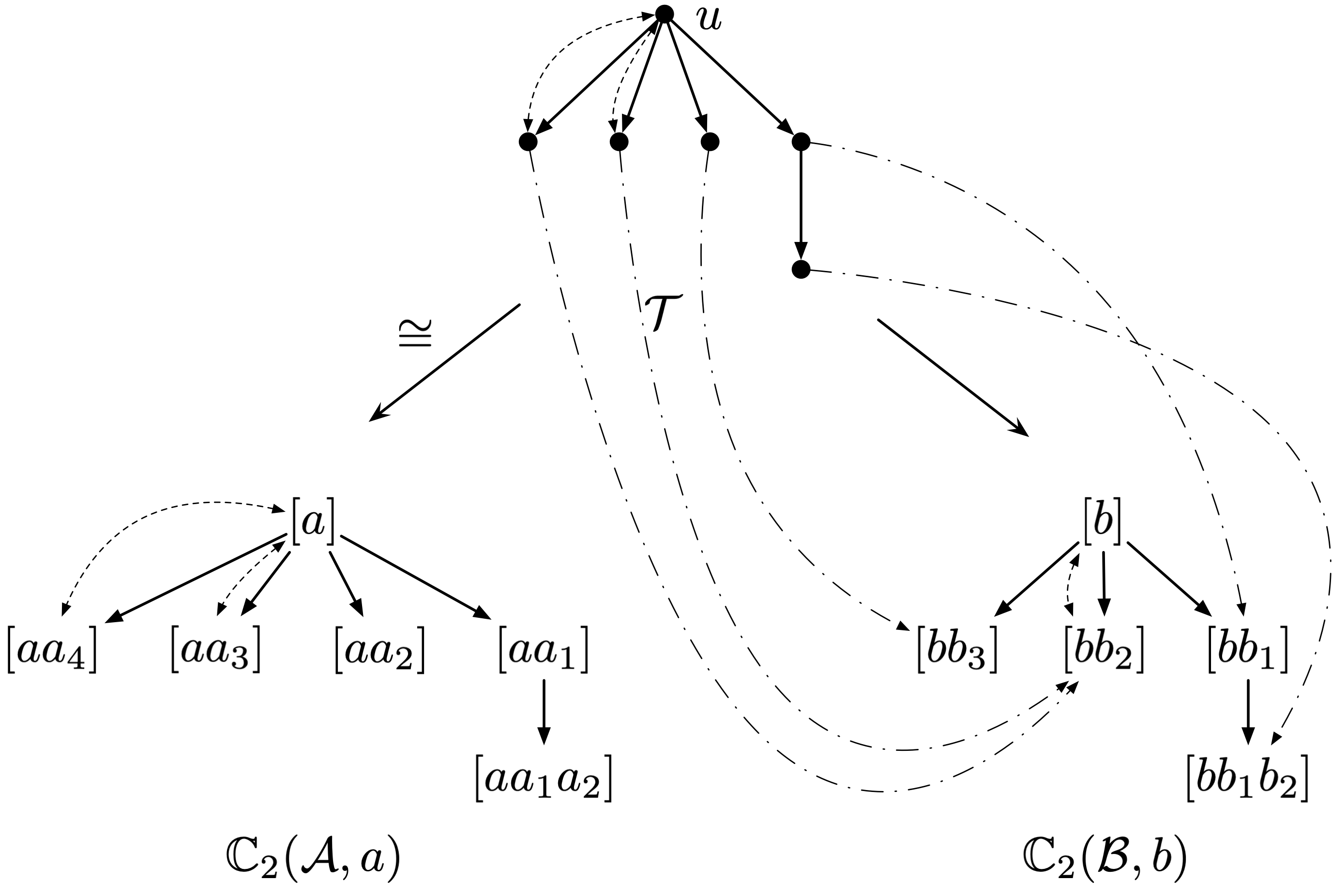}
\caption{
A span of bounded morphisms between the unravellings $\bb C_2(\A,a)$ and $\bb C_2(\B,b)$ of the structures $\Ap$, $\Bp$ as in Figures~\ref{fig:example_bisimilar_structures} and~\ref{fig:pptrees}. The left leg of the span is an isomorphism (there are two possible choices), while the dashed arrows indicate how the right leg acts on each point. By Corollary~\ref{coro:full_logical_equivalence}, the existence of this span of bounded morphisms shows that $\Ap \equiv_2 \Bp$.
} 
\label{fig:span}
\end{figure}

\begin{corollary}\label{coro:full_logical_equivalence}
    Let $\Ap, \Bp \in \Structpointed$ and suppose that either $\sigma$ is finite or both $\A$ and $\B$ are finitely branching structures. Then $\Ap \equiv_k \Bp$ iff there exists a span of pointed bounded morphisms $\Ck(\A, a) \leftarrow \T \to \Ck(\B, b)$ with some pp-tree $\T$ of height at most $k$ as common domain.
\end{corollary}

\begin{example}
Consider structures $\A$ and $\B$ as in Example~\ref{ex:bisimilar_structures}. We have already seen that $\Ap \bisim_k \Bp$ for all $k$ using explicit bisimulations and bisimulation games. In Figure~\ref{fig:span} we show a proof of the same fact based on bounded morphisms. Indeed, we show a span of bounded morphisms between $\Ck\Ap$ and $\Ck\Bp$
whose apex is a pp-tree of height $2$, hence by Theorem~\ref{thm:full_logical_equivalence} $\Ap \bisim_2 \Bp$ (and thus $\Ap \equiv_2 \Bp$). The fact that this also holds for $k > 2$ follows immediately from the observation that $\Ck\Ap \cong \bb C_2\Ap$ and $\Ck\Bp \cong \bb C_2\Bp$ for all $k > 2$.

Notice that in this example the left leg of the span is an isomorphism (indeed, isomorphisms are trivially bounded morphisms). This means that the span can be rewritten as simply a bounded morphism $\bb C_2 \Ap \to \bb C_2\Bp$. Since this bounded morphism can be seen to arise as $\bb C_2 f$ from a morphism $f: \Ap \to \Bp$, we may express this fact as saying that $f$ is a \define{functional $k$-bisimulation} (with $k=2$) for \logic.
\end{example}

\section{Model Theory of \boldmath{\logic}}\label{sec:model_th_of_ppml}

In this section, we explore some additional consequences of the comonadic formalism for the model theory of \logic. We first describe an extension of \logic with graded modalities, which we call $\logic^\#$, which of course generalises the graded modalities of \bml (see e.g.~\cite{de2000note}). In the context of logics described through game comonads, this kind of extension is captured by a simpler condition than the existence of spans of open pathwise embeddings—in our case, it is captured by isomorphism of \logic-unravellings. Considering $\logic^\#$ allows us to obtain immediately a homomorphism counting theorem for \logic by an application of a general comonadic result, namely \cite[Corollary 14]{dawar2021lovasz}. It also allows us to state and prove in full strength the tree-model property that \logic enjoys; this is the second topic covered in this section. Finally, we define a Chandra-Merlin-like correspondence between \logicpos formulas and finite pp-trees, which leads to an alternative proof of the Hennessy-Milner property for $k$-similarity between finite structures.
\subsection{Graded Modalities and Homomorphism Counting}\label{sec:graded_modalities}

Game comonads function as a mathematical framework in two different ways: it can provide general theorems which can be directly instantiated in new game comonads, or, when there is currently no such theorem, it can provide a guideline for producing new results by simple adaptations of previous arguments.
In this section, we combine both methodologies: first, inspired by similar results for many other game comonads (and in particular for the Modal comonad), we prove that isomorphism between the \logic-unravellings of two pointed structures captures logical indistinguishability for a suitably defined graded extension of $\logic$. Then we apply a general theorem of~\cite{dawar2021lovasz} to obtain a homomorphism-counting or Lovász-type theorem for our logic.

Note that having isomorphic unravellings is a stronger condition than the one featuring in Theorem~\ref{thm:full_logical_equivalence}, since given an isomorphism $\phi: \Ck(\A, a) \xrightarrow[]{\sim} \Ck(\B, b)$ we have a span of bounded morphisms $\Ck(\A, a) = \Ck(\A, a) \xrightarrow[]{\phi} \Ck(\B, b)$.

Analogously to the Modal Logic case, we extend \logic with graded modalities of the form $\Diamond_n$ with intended meaning `there exist at least $n$ successors such that...'. 
Formally, the syntax of $\logic^\#$ is as in Definition~\ref{def:syntax_ppml} but replacing $\Diamond\varphi$ by $\Diamond_n\varphi$ where $n$ ranges over the positive integers. The semantics are as in Definition~\ref{def:semantics_ppml} but this time with
\begin{align*}
    \A,s \models \Diamond_n\varphi  &\text{\quad iff \quad} |\{ a\in |\A|\colon (\epsilon(s),a)\in \acc^{\A}, \A, s.a \models \varphi\}|\geq n.
\end{align*}
A suitable notion of bisimulation for graded \bml was introduced in~\cite{de2000note}. Here we take the view of~\cite{abramsky2021relating} and adapt it to $\logic^\#$: we modify the game of Definition ~\ref{def:ppml_games} to define a new $k$-round game.\footnote{We change our style of presentation with respect to previously introduced games into a recursive definition of $\G_k^{\#}$ in terms of $\G_{k-1}$. In doing so, we define a game where the initial position can be given by arbitrary $E$-chains $s, t$ in $\A$ and $\B$, respectively, allowing us to define graded $k$-bisimilarity between $E$-chains directly in terms of the game.}

\begin{definition}
    Given structures $\A, \B \in \St$ and $E$-chains $s$ in $\A$ and $t$ in $\B$, the game $\G_k^{\#}((\A, s), (\B, t))$ is played between the two players Spoiler and Duplicator as follows. If $s$ and $t$ do not satisfy exactly the same relations, or if there is no bijection $\acc^{\A}(\epsilon(s)) \xrightarrow{\cong} \acc^{\B}(\epsilon(t))$, then Duplicator loses the game. Otherwise, if $k=0$, Duplicator wins the game, while if $k>0$, Duplicator chooses one such bijection $\theta$ and Spoiler chooses a pair $(a', b')$ in the graph of $\theta$, which we denote by $\Gamma(\theta)$. The players then continue playing the game $\G_{k-1}^{\#}((\A, s.a'), (\B, t.b'))$. If Spoiler cannot choose such a pair (because $\acc^{\A}(\epsilon(s)) = \acc^{\B}(\epsilon(t)) = \varnothing$), Duplicator wins the game.
    We use the notation $(\A, s) \bisim_k^{\#} (\B, t)$ (or simply $s \bisim_k^{\#} t$) to mean that there exists a winning strategy for Duplicator in the game $\G_k^{\#}((\A, s), (\B, t))$. In what follows we may specialise to the case where $|s| = |t| = 1$, in which case we do not distinguish between points $a, b$ and the corresponding sequences $[a], [b]$.
\end{definition}

To relate graded bisimilarity to the comonad $\Ck$, we use a construction of a $k$-unravelling of a structure $\A$ at a chain $s$ in $\A$, which we denote by $\Ck(\A ,s)$. The key property of $\Ck(\A, s)$ is that it can be decomposed as a (non-disjoint) union of embedded substructures of the form $\bb C_{k-1}(\A, s.a')$ where $\epsilon(s) \prec a'$.

\begin{definition}\label{def:unravelling_at_a_chain}
    Given $\A \in \Struct$ and an $E$-chain $s$ in $\A$, let $\Ck(\A, s)$ denote the embedded substructure of $\bb C_{k+|s|-1}(\A, s(1))$ whose universe is the set of $E$-chains that are comparable with $s$ in the prefix order.
\end{definition}

We can follow similar a similar approach to the proof for the analogous statement for Modal Logic \cite[Proposition 5.5]{abramsky2021relating} in order to obtain the following the result. However, a non-trivial adaptation is needed in order for the induction to go through, for which the construction in Definition~\ref{def:unravelling_at_a_chain} plays a key role.

\begin{lemmarep}\label{lem:graded_bisimilarity_isomorphism_strings}
    Let $\A, \B \in \Struct$ and $k \geq 0$. Then for all pairs $(s, t)$ where $s$ is a chain in $\A$, $t$ is a chain in $\B$ and $|s| = |t|$ the following are equivalent:
    \begin{enumerate}
        \item $s \bisim_k^{\#} t$ and $\bb C_0(\A, s) \cong \bb C_0(\B, t)$;
        \item $\bb C_{k}(\A,s) \cong \bb C_{k}(\B, t)$.
    \end{enumerate} 
\end{lemmarep}
\begin{proof}
    We proceed by induction on $k$. The base case $k = 0$ is evident. Now assume the statement holds for $k$ and let us prove it for $k + 1$. For the $\text{(1)} \implies \text{(2)}$ implication, suppose that $s \bisim_{k+1}^{\#} t$ and $\bb C_0(\A, s) \cong \bb C_0(\B, t)$. Then there exists a bijection $\theta: E^\A(\epsilon(s)) \xrightarrow{\cong} E^\B(\epsilon(t))$ such that $s.a \bisim_k^{\#} t.b$ for all $(a, b) \in \Gamma(\theta)$.
    In order to use the inductive hypothesis, let us prove that $\bb C_0(\A, s.a) \cong \bb C_0(\B, t.b)$ holds for all $(a, b) \in \Gamma(\theta)$. Let $\phi$ be the (unique) isomorphism $\phi: \bb C_0(\A, s) \cong \bb C_0(\B, t)$ and given such a pair $(a, b) \in \Gamma(\theta)$ let $\phi_a: \bb C_0(\A, s.a) \cong \bb C_0(\B, t.b)$ be the bijection which extends $\phi$ by $\phi_a(s.a) \coloneqq t.b$. Now given a relation symbol $R$ of arity $r \leq |s| + 1$ and $(s_1,\dots,s_r) \in |\bb C_0(\A, s.a)|^r$, if $s_1,\dots,s_r$ are prefixes of $s$ then
    $$(s_1,\dots,s_r) \in R^{\bb C_0(\A, s.a)} \iff (\phi_a(s_1),\dots, \phi_a(s_r)) = (\phi(s_1),\dots, \phi(s_r)) \in R^{\bb C_0(\B, t)} \subseteq R^{\bb C_0(\B, t.b)}$$
    where we have used that $\phi$ is an isomorphism. If on the other hand $s.a$ is among $s_1,\dots,s_r$, then it must be that $s_1 \prec \dots \prec s_r = s.a$. The assumption that $s.a \bisim_k^{\#} t.b$ implies that $s.a \bisim_0^{\#} t.b$, hence 
    \begin{align*}
    (s_1,\dots,s_r) \in R^{\bb C_0(\A, s.a)} & \iff \A, s.a \models R \\
    &\iff \B, t.b \models R \\
    &\iff (\phi(s_1),\dots, \phi(s_r)) \in R^{\bb C_0(\B, t.b)}.
    \end{align*}
    This proves that $\bb C_0(\A, s.a) \cong \bb C_0(\B, t.b)$ for all $(a, b) \in \Gamma(\theta)$, thus the inductive hypothesis gives us a pair of inverse homomorphisms $$f_a: \bb C_k(\A, s.a) \rightleftarrows \bb C_k(\B, t.b) : g_b$$ for each $(a, b) \in \Gamma(\theta)$.
    
    Notice that each $f_a$ is defined on a subset of $|\bb C_{k+1}(\A, s)|$. Thus we have a collection of homomorphisms $(f_a)_{a \in E^\A(\epsilon(s))}$ such that they agree on the pairwise intersection of their domains and such that the union of their domains is equal to $|\bb C_{k+1}(\A, s)|$, thereby defining a unique function $f: |\bb C_{k+1}(\A, s)| \to |\bb C_{k+1}(\B, t)|$ by $f(s') \coloneqq f_a(s')$ for any $f_a$ such that its domain contains $s'$. This function is actually a homomorphism: whenever $(s_1,\dots,s_r) \in R^{\bb C_{k+1}(\A, s)}$, it must be the case that $\{s_1,\dots,s_r\}$ is in the domain of at least one $f_a$ and then, since $f_a$ is a homomorphism, $(f(s_1),\dots,f(s_r)) = (f_a(s_1),\dots,f_a(s_r)) \in R^{\bb C_{k}(\B, t.b)} \subseteq R^{\bb C_{k+1}(\B, t)}$.

    The same argument defines a homomorphism $g: \bb C_{k+1}(\B, t) \to \bb C_{k+1}(\A, s)$ which by construction is the inverse of $f$, thus showing that $\bb C_{k+1}(\B, t) \cong \bb C_{k+1}(\A, s)$.

    We now prove the leftwards implication starting from the assumption that there exists an isomorphism $f: \bb C_{k+1}(\A,s) \xrightarrow{\cong} \bb C_{k+1}(\B, t)$. Trivially, $\bb C_0(\A, s) \cong \bb C_0(\B, t)$. Now let $\theta$ be the composite bijection
    $$ E^\A(\epsilon(s)) \cong E^{\bb C_{k+1}(\A, s)}(s) \xrightarrow{f} E^{\bb C_{k+1}(\B, t)}(t) \cong E^\B(\epsilon(t))$$
    where the intermediate arrow represents the appropriate (co)restriction of $f$. We must show that for each pair $(a, b) \in \Gamma(\theta)$ it is the case that $s.a \bisim_k^{\#} t.b$. By the inductive hypothesis, it is enough to show that there exist isomorphisms $f_a : \bb C_k(\A, s.a) \xrightarrow{\cong} \bb C_k(\B, t.\theta(a))$ for all $a \in E^\A(\epsilon(s))$. These isomorphisms are supplied by appropriately restricting and co-restricting $f$, since by our choice of $\theta$, if $s.a$ is a prefix of $s' \in \bb C_{k+1}(\A, s)$ then $f(s.a) = t.\theta(a)$ is a prefix of $f(s')$.
\end{proof}

\begin{theorem}\label{thm:graded_bisimilarity_by_iso_classes}
    For all $(\A, a), (\B, b) \in \Structpointed$, $(\A, a) \bisim_k^{\#} (\B, b)$ iff $\Ck(\A, a) \cong \Ck(\B, b)$.
\end{theorem}
\begin{proof}
    Apply Lemma~\ref{lem:graded_bisimilarity_isomorphism_strings} with $s = [a], t = [b]$. In this case, condition (1) reduces to $a \bisim_k^{\#} b$ since this already implies that $\bb C_0(\A, a) \cong \bb C_0(\B, b)$.
\end{proof}

From the side of logic, a routine adaptation of \cite[Proposition 4.11]{aceto2010resource} proves a Hennessy-Milner-type result analogous to Theorem~\ref{thm:hennessy-milner}. Let $\logick^\#$ denote the fragment of $\logic^\#$ of modal depth at most $k$. Given $\Ap, \Bp \in \Structpointed$, we write $\Ap \equiv_k^{\#} \Bp$ to mean that $\Ap\models\varphi$ iff $\Bp\models\varphi$ for all $\varphi$ in $\logick^{\#}$.

\begin{proposition}\label{prop:hm_for_graded_modalities}
Let $\sigma$ be a relational signature with $E \in \sigma$ and let $(\A,a)$, $(\B,b) \in \Structpointed$. Assume that $\sigma$ is finite or $\A$ and $\B$ are finitely branching. Then 
$(\A,a) \bisim_k^{\#} (\B,b)$ if and only if $(\A,a) \equiv_k^{\#} (\B,b)$.
\end{proposition}

\begin{corollary}\label{coro:graded_types_are_iso_classes}
    Let $(\A,a), (\B,b) \in \Structpointed$. If $\sigma$ is finite or $\A$ and $\B$ are finitely branching, then $(\A,a) \equiv_k^{\#} (\B,b)$ if and only if $\Ck(\A,a) \cong \Ck(\B,b)$.
\end{corollary}

\begin{remark}\label{rem:isomorphisms_where}
    In general, a game comonad $G$ captures the indistinguishability of two objects $X$ and $Y$ for its corresponding graded or counting logic through their isomorphism in the category $\Kl(G)$. Here, instead, we have shown that $\Ap \bisim_k^\# \Bp$ if and only if $\Ck \Ap \cong \Ck \Bp$ as pointed $\sigma$-structures. This simplification is possible thanks to the idempotence of $\Ck$. Simply put, for an idempotent comonad $G$ on a category $\cat{E}$, the Kleisli and EM categories are equivalent, and as we have already mentioned, the EM category is a full subcategory of $\cat{E}$. Hence $\Ck \Ap \cong \Ck \Bp$ in $\Stp$ if and only if $\Ap \cong \Bp$ in $\Kl(\Ck)$, if and only if $\Ck \Ap \cong \Ck \Bp$ in $\EM(\Ck)$.
\end{remark}

As an application of this equivalence, the fact that \logic admits a description through an idempotent comonad readily implies a homomorphism-counting theorem akin to the classic theorem of Lovász~\cite{lovasz1967operations} and the more recent theorems of Grohe~\cite{grohe2020counting} and Dvořak~\cite{dvovrak2010recognizing}. This is what we now prove, essentially as a Corollary of \cite[Corollary 14]{dawar2021lovasz}. For this proof, we will make use again of the language of adjunctions.

\begin{remark}\label{rem:ppml_comonad_restricts_to_finite}
    Let $\Stpf$ denote the full subcategory of $\Stp$ on the finite structures, and let $\EMf(\Ck)$ denote the full subcategory of $\EM(\Ck)$ on the finite pp-trees of height at most $k$ (thus excluding the trees with an infinite number of branches of bounded length). To ease notation, in what follows we will write $\Spf$ for $\Stpf$ and $\EMf$ for $\EMf(\Ck)$.

    Notice that, since the $k$-unravelling of a finite structure is again finite, the comonadic adjunction $U \dashv F$ of $\Ck$ restricts to an adjunction
\[\begin{tikzcd}[ampersand replacement=\&]
	{\Spf \coloneqq \Stpf} \&\& {\EMf \coloneqq \EMf(\Ck)}
	\arrow[""{name=0, anchor=center, inner sep=0}, "{F_\text{f}}", curve={height=-18pt}, from=1-1, to=1-3]
	\arrow[""{name=1, anchor=center, inner sep=0}, "{U_\text{f}}", curve={height=-18pt}, hook', from=1-3, to=1-1]
	\arrow["\dashv"{anchor=center, rotate=90}, draw=none, from=1, to=0]
\end{tikzcd}\]
    between the corresponding full subcategories of finite objects. The resulting comonad $\Ck^\text{f} \coloneqq U_\text{f}F_\text{f}$ is the restriction of $\Ck$ to $\Spf$.
\end{remark}

Recall that given a locally small category $\cat{C}$ and $c, c' \in \cat{C}$, we write $\cat{C}(c, c')$ for the set of morphisms from $c$ to $c'$ and hence $\#\cat{C}(c, c')$ denotes the cardinality of that set.

\begin{theorem}\label{thm:lovasz_for_ppml}
    Let $\Ap, \Bp \in \Structpointed$ be finite $\sigma$-structures. Then $\Ap \equiv_k^\# \Bp$ if and only if
    $$\#\Spf((\+T, u), \Ap) = \# \Spf((\+T, u), \Bp)$$
    for all finite pp-trees $(\+T, u)$ of height at most $k$.
\end{theorem}
\begin{proof}
    By \cite[Corollary 14]{dawar2021lovasz} and \cite[Remark 15]{dawar2021lovasz}, $\EMf$ is \define{combinatorial}, meaning that for all $X, Y \in \EMf$, $X \cong Y$ if and only if $\# \EMf(Z, X) = \# \EMf(Z, Y)$ for all $Z \in \EM^\text{f}$.\footnote{It is arguably the most fundamental fact of Category Theory that in any category $\cat{C}$, $X \cong Y$ if and only if $\cat{C}(-,X) \cong \cat{C}(-,Y)$ in the sense of a natural isomorphism between functors. In this sense, a combinatorial category is one in which it is enough to count the number of elements in each set of the form $\cat{C}(Z, X)$ in order to determine the functor $\cat{C}(-,X)$ up to isomorphism.}
    Consider now the following chain of equivalences:
\begin{align*}
    \Ap \equiv_k^\# \Bp &\iff \Ap \bisim_k^\# \Bp \tag{Prop.~\ref{prop:hm_for_graded_modalities}, $\A, \B$ finite} \\
    &\iff F_\text{f}\Ap \cong F_\text{f}\Bp \tag{Thm.~\ref{thm:graded_bisimilarity_by_iso_classes}, Rmk.~\ref{rem:isomorphisms_where}} \\
    &\iff \forall X \in \EMf : \# \EMf(X, F_\text{f}\Ap) = \# \EMf(X, F_\text{f}\Bp) \tag{$\EMf$ combinatorial}\\
    &\iff \forall X \in \EMf : \# \Spf(U_\text{f}X, \Ap) = \# \Spf(U_\text{f}X, \Bp) \tag{$U_\text{f} \dashv F_\text{f}$}\\
\end{align*}
where in the last line we have used the restricted adjunction of Remark~\ref{rem:ppml_comonad_restricts_to_finite}.
This is exactly what we wished to prove, since finite pp-trees of height at most $k$ are precisely the objects of $\EMf$ and the left adjoint $U_\text{f}$ is simply the inclusion of $\EMf$ in $\Spf$.
\end{proof}

We note in passing that the proof given above seems simpler than the one given for the analogous result for \bml \cite[Theorem 31]{dawar2021lovasz}, which we recover as a particular case for unimodal choices of $\sigma$.

\subsection{The pp-tree-model Property}\label{subsec:expressivity}

In the same way one can reason with games or with bisimulations to conclude that a certain property is not expressible within a logic such as \logic, one can also do so through Corollary~\ref{coro:full_logical_equivalence}.

\begin{definition}
    Following the literature~\cite{abramsky2022hybrid,abramsky2022emerging, abramsky2022preservation}, if $G$ is the comonad arising from an arboreal cover of a category $\cat{E}$, we say that $X, Y \in \cat{E}$ are \define{bisimilar} iff there exists a span of open pathwise embeddings $GX \leftarrow Z \to GY$ in $\EM(G)$. We say that $G$ has the \define{bisimilar companion property} if $GX$ is bisimilar to $X$ for all $X \in \cat{E}$.
\end{definition}

In~\cite{abramsky2022preservation} comonads (and arboreal covers) that satisfy either idempotence or the bisimilar companion property are referred to as \emphat{tame}, and the good properties that these comonads enjoy seem to reflect the fact that they correspond to less expressive logics. First notice that if $G$ is idempotent then it satisfies the bisimilar companion property: given $X \in \cat{E}$, just take the span of open pathwise embeddings in $\EM(G)$ to be $G G X \xleftarrow{\delta_X} G X = G X$.\footnote{Here we are again using the fact that $\EM(G)$ can be identified with a full subcategory of $\cat{E}$ thanks to idempotence. Contrast this with the language used in e.g.\ \cite[Prop. 5.4]{abramsky2022preservation}.}

However, the converse is not true, as is exemplified by the comonads for guarded fragments~\cite{abramsky2021comonadic}.
To understand why idempotence is even stronger than the bisimilar companion property, we can make use of the graded logic $\logic^\#$ introduced in Section~\ref{sec:graded_modalities}.

\begin{corollary}\label{coro:bisimilar-companion-ppml}
    $\Ap \bisim_k^\# \Ck \Ap$ for all $\Ap \in \Structpointed$. Thus,
    if $\Ap$ is finitely branching or $\sigma$ is finite, $\Ap \equiv_k^\# \Ck \Ap$, and hence also $\Ap \equiv_k \Ck \Ap$.
\end{corollary}
\begin{proof}
    Since $\Ck$ is idempotent, $\Ck \Ap \cong \Ck \Ck \Ap$ for all $\Ap \in \Structpointed$, hence the result follows by Corollary~\ref{coro:graded_types_are_iso_classes}.
\end{proof}

Note that the argument does not depend on the concrete description of open pathwise embeddings as bounded morphisms. Instead, the result follows immediately from the fact that $\Ck$ is idempotent, using only abstract notions.

Corollary~\ref{coro:bisimilar-companion-ppml} allows us to derive many expressivity results about \logic. One such result says that \logic enjoys a \define{pp-tree-model property} which generalises \bml's tree-model property:
\begin{corollary}\label{coro:coalgebra-model_property}
    A $\logic^\#$ formula $\phi$ is satisfiable if and only if it is satisfied by a pp-tree of finite height. In particular this also holds for \logic formulas.
\end{corollary}
\begin{proof}
    Let $k$ be the modal depth of $\phi$. Without loss of generality, we can assume that $\sigma$ contains only the relation symbols in $\phi$ together with $\acc$, hence it is finite. Corollary~\ref{coro:bisimilar-companion-ppml} then implies that $\phi$ is satisfiable if and only if it is satisfiable in some $\Ck$-coalgebra. Thus any $\phi$ is satisfiable if and only if it is satisfied by a pp-tree of finite height.
\end{proof}
More generally, any idempotent game comonad will imply some ``coalgebra-model property'' for its corresponding logic.

This allows us to prove that many properties are not \logic-expressible, e.g.\ the property ``in the interpretation of $R$, there is a tuple that is not an $\acc$-chain'' for some $R \in \overline{\sigma}$. This is obviously the case for many structures and yet it cannot be true of any pp-tree; thus it is not expressible in \logic.

\subsection{Canonical Models and the Hennessy-Milner Property}\label{sec:chandra_merlin}

Following the general relationship between coalgebras and conjunctive queries presented in~\cite{abramsky2021relating}, we can think of finite pp-trees as reifications of \logicpos formulas

via a Chandra-Merlin-like correspondence~\cite{chandra1977optimal}.
In this way, the comonadic formalism leads to an alternative proof of Theorem~\ref{thm:hennessy-milner} (1) for finite structures.

In adapting the idea from First Order Logic, we must be careful with the fact that in \logic not all positive formulas are satisfiable, and hence not all of them have a canonical model: we must restrict to the well nested formulas.

Given a finite pp-tree $\T$ of height $k$, we wish to construct a \logicpos formula $\nu(\T)$ such that for all $\Ap \in \Stp$, $\Ap \models \nu(\T)$ if and only if there exists a morphism $\T \to \Ap$. The construction deviates from that of canonical modal conjunctive queries for \bml \cite[Section 8.3]{abramsky2021relating} since the presence of relations over paths implies that we cannot write $\nu(\T)$ recursively in terms of $\nu(\T')$ for each of the subtrees of $\T$. We solve this by a technique analogous to Lemma~\ref{lem:graded_bisimilarity_isomorphism_strings}, which involves generalising to formulas $\nu(\T,s)$ where $s$ is a \emphat{stem} of $\T$, in the sense of the definition below.

\begin{definition}
    Given a pp-tree $\T$, we say that an $E$-chain $s$ in $\T$ is a \define{stem} of $\T$ if it is a prefix of all branches of $\T$. Moreover, given $w \in |\T|$ we define $\T^w$ as the embedded sub-pp-tree of $\T$ containing all points of $\T$ that are comparable with $w$ in the partial order $(\acc^\T)^*$.
\end{definition}
Notice that given a pp-tree $\T$ and $v \in |\T|$, $\T_v$ is the maximal stem of $\T^v$. These notions are related to Definition~\ref{def:unravelling_at_a_chain}: given $\A \in \St$ and a chain $s$ in $\A$, $\Ck(\A,s)$ can be equivalently defined as $(\C_{k+|s|-1}(\A, s(1)))^s$.
    
\begin{definition}
    Let $\T$ be a pp-tree and $s$ be a stem of $\T$. Suppose that there are only finitely many non-empty interpretations in $\T$, i.e.\ $\{R \in \bsigma \mid \T^R \neq \varnothing\}$ is finite.%
    \footnote{This condition is necessary since we are working with finitary conjunctions only. In particular, this is trivial when $\sigma$ is finite.}
    Then we define inductively the formula\footnote{We define $\nu(\T,s)$ only up to a choice of ordering on $\bsigma$ and on the successor sets, but this is immaterial for the current discussion. Also note that the empty conjunction is taken to be syntactically equal to $\top$.}
    $$\nu(\T, s) \coloneqq \bigwedge\{R\in\bsigma\mid \T,s\models R\} \land \bigwedge_{w\in\acc^\T(\epsilon(s))} \Diamond \nu(\T^w, s.w).$$
    Since $\T$ is a finite tree, the definition is well-founded, and moreover the resulting formula is clearly in \logicposk where $k$ is the height of $\T$. When $s=[\bpt]$, we denote $\nu(\T, [\bpt])$ by $\nu(\T)$.
\end{definition}

\begin{lemmarep}
    Let $\T$ be a finite pp-tree with finitely many non-empty interpretations. Let $s$ be a stem of $\T$ and let $\Ap \in \Stp$. If there exists a morphism $f: \T \to \Ap$, then $\A, f(s) \models \nu(\T,s)$. Conversely, if there exists an $E$-chain $t$ in $\A$ starting at $\bpa$ with $|t| = |s|$ such that $\A, t \models \nu(\T, s)$, then there exists a morphism $f: \T \to \Ap$ such that $f(s) = t$.
\end{lemmarep}
\begin{proof}
    Let $k$ be the height of $\T$. We proceed by induction on the parameter $q \coloneqq k - |s| + 1$. If $q=0$, then $s$ is a branch of $\T$, hence since it is also a stem of $\T$, it must be that $\T$ is a pp-path. On the other hand $\nu(\T, s)$ reduces to $\bigwedge\{R \in \bsigma \mid \T, s \models R\}$, hence the result is immediate.

    For the inductive step, let $q + 1 = k - |s| + 1$ and suppose the result holds for all pairs $(\T', s')$ where $\T'$ is a pp-tree of height $k'$ and $s'$ is a stem of $\T'$ such that $k' - |s'| + 1 \leq q$. First, suppose there exists a morphism $f: \T \to \Ap$. We wish to show that $\A, f(s) \models \phi_1 \land \phi_2$ where $\phi_1 \coloneqq \bigwedge\{R \in \bsigma \mid \T, s \models R\}$ and $\phi_2 \coloneqq \bigwedge_{w \in E^\T(\epsilon(s))} \nu(\T^w, s.w)$. Clearly $\A, f(s) \models \phi_1$ since $f$ preserves relations. As for $\phi_2$, for each $w \in E^\T(\epsilon(s))$ consider the restriction $f|_{\T^w}: \T^w \to \Ap$. Setting $\T' = \T^w$ and $s' = s.w$, by the inductive hypothesis we obtain that $\A, f(s.w) \models \nu(\T^w, s.w)$. Since $f(s.w) = f(s).f(w)$, $f(w)$ is a witness showing that $\A, f(s) \models \Diamond \nu(\T^w, s.w)$.

    In the other direction, suppose now that $\A, t \models \nu(\T, s)$ where $t$ is an $E$-chain starting at $\bpa$ with $|t| = |s|$. We wish to construct a morphism $f: \T \to \Ap$ such that $f(s) =t$. For each $w \in E^\T(\epsilon(s))$, let $a^w$ be some successor of $\epsilon(t)$ such that $\A, t.a^w\models \nu(\T^w, s.w)$. Then by the inductive hypothesis there exists a morphism $f^w: \T^w \to \Ap$ such that $f(s.w) = t.a^w$, in particular $f(s) = t$. Since the family of partially defined morphisms $(f^w)_w$ matches on their intersection and moreover $\bigcup_{w} \T^w = \T$, this defines a unique morphism $f: \T \to \Ap$ such that $f|_{\T^w} = f^w$ for all $w$.
\end{proof}

We write $\Ap \to \Bp$ to mean that there exists a homomorphism from $\Ap \to \Bp$.

\begin{corollary}\label{coro:chandra_merlin_for_ppml_1}
    For every finite pp-tree $\T$ with finitely many non-empty interpretations and for all $\Ap \in \Stp$, $\T \to \Ap$ if and only if $\Ap \models \nu(\T)$.
\end{corollary}

This establishes one of the two directions of the correspondence between \logicpos formulas and finite pp-trees. In the other direction,for each well-nested \logicpos formula $\phi$ we wish to define a pp-tree $\+{M}(\phi)$ such that $\+M(\phi) \to \Ap$ if and only if $\Ap\models\phi$. Again, we must define more generally a pp-tree $\+M(\phi, \T, v)$ where $\phi$ is any, not necessarily well-nested formula in \logicpos, $\T$ is a pp-tree and $v$ is a leaf of $\T$ whose height is large enough with respect to the modal debt of $\phi$ (see Def.~\ref{def:modal_debt}). To this end we will make use of the following operations on pp-trees.
\begin{itemize}
    \item \emph{Adding a tuple to an interpretation.} Given a pp-tree $\T$, a symbol $R \in \bsigma$ of arity $r$ and a point $v \in |\T|$ of height at least $r-1$, we define the pp-tree $\T(R,v)$ by $|\T(R,v)| \coloneqq |\T|$, $R^{\T(R,v)} \coloneqq R^\T \cup \{\last_r(\T_v)\}$ and $R'^{\T(R,v)} := R'^\T$ for all $R' \in \bsigma \setminus \{R\}$.
    \item \emph{Pushout along a common sub-pp-tree.} Given two pp-trees $\T, \T'$ and pointed injective homomorphisms $f: \+U \to \T$ and $g: \+U \to \T'$ whose domain $\+U$ is a pp-tree, we denote by $\T +_{\+U} \T'$ the corresponding pushout in $\Stp$, i.e.\ the colimit of the span $\T \leftarrow \+U \to \T'$. For concreteness, we may assume (by renaming points if necessary) that $|\+U| = |\T| \cap |\T'|$ and set $|\T +_{\+U} \T'| \coloneqq |\T| \cup |\T'|$ and $R^{\T +_{\+U} \T'} \coloneqq R^\T \cup R^{\T'}$ for all $R \in \sigma$.
    \item \emph{Edge creation at a leaf.} Given a pp-tree $\T$ and a leaf $v \in |\T|$, we denote by $\T \cup_v w$ the pp-tree obtained by adding to $\T$ a new point $w$ as a successor of $v$ (we guarantee that $w \not\in |\T|$ by a renaming of points if necessary). Formally, $|\T \cup_v w| = |\T| \cup \{w\}$, $\acc^{\T\cup_v w} = \acc^\T \cup \{(v,w)\}$ and $R^{\T \cup_v w} = R^\T$ for all $R \in \bsigma$.
\end{itemize}

\begin{definition}
    Given a \logicpos formula $\phi$, a pp-tree $\T$ and a leaf $v$ such that $\debt(\phi)$ is at most equal to the height of $v$, we define the pp-tree $\+M(\phi, \T, v)$ inductively as follows:
    \begin{align*}
        \+M(\top, \T, v) &\coloneqq \T \\
        \+M(R, \T, v) &\coloneqq \T(R,v) \tag{$R \in \bsigma$} \\
        \+M(\phi_1 \land \phi_2, \T, v)  &\coloneqq \+M(\phi_1, \T, v) +_\T \+M(\phi_2, \T, v) \\
        \+M(\Diamond \psi, \T, v) &\coloneqq  \+M(\psi, \T \cup_{v} w, w).
    \end{align*}
    Notice that the second clause is well defined because $h \geq \debt(\phi)$, the third clause is well defined because $\T$ is always a sub-pp-tree of $\+M(\phi, \T, s)$ (although it will not in general be an \emphat{embedded} sub-pp-tree), and the fourth clause is well defined because although the formula $\psi$ on the right-hand side has higher debt, the newly added leaf $w$ has higher height as well.

    If $\T = \{*\}$ is the singleton universe with empty interpretations and $\phi$ is a well-nested \logicpos formula, we write $\+M(\phi) \coloneqq \+M(\phi, \{*\}, *)$. Notice that $\+M(\phi)$ is finite and its height coincides with the modal depth of $\phi$.
\end{definition}
    
    Given $\phi, \T$ and $v$ as in the definition above and a morphism $f: \+M(\phi, \T, v) \to \Ap$, since $\T$ is always a (not necessarily embedded) sub-pp-tree of $\+M(\phi, \T, v)$, we may always restrict $f$ to a well-defined morphism $f|_\T: \T \to \Ap$ (in the case $\phi = R \in \bsigma$, this just means that a homomorphism is still a homomorphism if a tuple is erased from the interpretation of $R$ in the domain). The following lemma answers the question of when a morphism $f: \T \to \Ap$ can be \emph{extended} from $\T$ to $\+M(\phi, \T, v)$, i.e.\ whether there exists some $F: \+M(\phi, \T, v) \to \Ap$ such that $F|_\T = f$.

\begin{lemmarep}
    Let $\phi$ be a \logicpos formula, let $\T$ be a finite pp-tree and let $v$ be a leaf of $\T$ such that $\debt(\phi)$ is at most equal to the height of $v$. Then given a morphism $f: \T \to \Ap$, $f$ extends to a morphism $F: \+M(\phi, \T, v) \to \Ap$ if and only if $\A, f(\T_v) \models \phi$.
\end{lemmarep}
\begin{proof}
    We proceed by structural induction on $\phi$. The case $\phi = \top$ is trivial. As for the case $\phi = R \in \bsigma$, the claim is that $f$ extends from $\T$ to $\T(R, v)$ if and only if $\A, f(\T_v) \models R$, which is clearly the case.

    Consider now $\phi = \phi_1 \land \phi_2$. Suppose that $\A, f(\T_v) \models \phi$. Then $\A, f(\T_v) \models \phi_1$ and $\A, f(\T_v) \models \phi_2$. Notice that the height of $v$ is greater than or equal to $\max(\debt(\phi_1), \debt(\phi_2))$, hence by the inductive hypothesis, $f$ extends to morphisms $F_1: \+M(\phi_1,\T,v) \to \Ap$ and $F_2: \+M(\phi_2,\T,v)$. Since $F_1$ and $F_2$ coincide on $\T$, they determine a unique morphism $F: \+M(\phi_1,\T,v) +_\T \+M(\phi_2, \T,v) \to \Ap$.
    
    In the other direction, suppose that $f$ extends to a morphism $F$ as above. Then in particular it extends to morphisms $F_1 \coloneqq F|_{\+M(\phi_1,\T,v)}: \+M(\phi_1,\T,v) \to \Ap$ and $F_2 \coloneqq F|_{\+M(\phi_2,\T,v)}: \+M(\phi_2,\T,v) \to \Ap$, which by the inductive hypothesis implies $\A, f(\T_v) \models \phi_1$ and $\A, f(\T_v) \models \phi_2$.

    Finally, consider the case $\phi = \Diamond \psi$. Suppose that $\A, f(\T_v) \models \phi$, i.e.\ there exists some successor $a'\succ f(v)$ such that $\A, f(\T_v).a' \models \psi$. We extend $f$ in two steps. First, let $\tilde{F}: \T \cup_v w \to \Ap$ be defined by extending $f$ with $\tilde{F}(w) \coloneqq a'$. This is a well defined morphism by definition of $\T \cup_v w$. Notice that $f(\T_v).a' = \tilde{F}((\T\cup_v w)_w)$. Therefore, since $\debt(\psi) = \debt(\phi) + 1$ is at most the height of $w$ in $\T\cup_v w$, by the inductive hypothesis $\tilde{F}$ extends to a morphism $F: \+M(\psi, \T\cup_v w, w) \to \Ap$. This is the morphism we wanted since $F|_\T = f$.

    In the other direction, let $F: \+M(\psi, \T\cup_v w, w) \to \Ap$ be a morphism such that $F|_\T = f$. Then since $\debt(\psi)$ is at most the height of $w$, by the inductive hypothesis we have that $\A, F((\T\cup_v w)_w) \models \psi$. Since $F((\T\cup_v w)_w) = f(\T_v).F(w)$, $F(w)$ witnesses that $\A, f(\T_v) \models \Diamond \psi$.
\end{proof}

\begin{corollary}\label{coro:chandra_merlin_for_ppml_2}
    Let $\phi$ be a well nested \logicpos formula. Then for all $\Ap \in \Stp$, $\Ap \models \phi$ if and only if $\+M(\phi) \to \Ap$.
\end{corollary}

Corollaries~\ref{coro:chandra_merlin_for_ppml_1} and~\ref{coro:chandra_merlin_for_ppml_2} establish the desired Chandra-Merlin-like correspondence. In particular we obtain the following consequence.

\begin{corollary}\label{coro:de_chandra_merlin}
    Given finite structures $\Ap, \Bp \in \Stp$ over an arbitrary signature $\sigma$ with $E \in \sigma$, the following are equivalent:
    \begin{enumerate}
        \item $\forall \T \in \EM(\Ck^\text{f}). \T \to \Ap \implies \T \to \Bp$
        \item $\Ap \Rrightarrow_k^+ \Bp$.
    \end{enumerate}
\end{corollary}
\begin{proof}
    The implication $(1) \implies (2)$ is immediate by Corollary~\ref{coro:chandra_merlin_for_ppml_2}. For the converse implication, suppose that $(2)$ holds and that there exists a morphism $\T \to \Ap$. For any finite subset $\tau \subseteq \bsigma$, let $\T_\tau$ denote the sub-pp-tree of $\T$ defined by $|\T_\tau| \coloneqq |\T|$, $E^{\T_\tau} \coloneqq E^\T$, $R^{\T_\tau} \coloneqq R^\T$ if $R \in \tau$ and $R^{\T_\tau} \coloneqq \varnothing$ otherwise. Then, by Corollary~\ref{coro:chandra_merlin_for_ppml_1}, we know that $\Ap \models \nu(\T_\tau)$ for all choices of $\tau$. By $(2)$, therefore, $\Bp \models \nu(\T_\tau)$ for all $\tau$, which is to say that $\T_\tau \to \Bp$ for all $\tau$.

    Arguing by contradiction, suppose that $\T \not\to \Bp$. In other words, for all functions $f: |\T| \to |\B|$ that map the root of $\T$ to $\bpb$ there exists some $R \in \sigma$ and some $s \in |\T|^+$ such that $\T,s\models R$ but $\B, f(s)\not\models R$.
    Then pick a morphism $f_0: \T_{\{E\}} \to \Bp$. By the preceding observation, there exists some $R_1 \in \bsigma$ and some $s$ such that $\T, s \models R_1$ but $\B, f_0(s) \not\models R_1$. Now let $f_1$ be some morphism $f_1:\T_{\{E,R_1\}} \to \Bp$. Clearly $f_1 \neq f_0$ since $f_1$ preserves the relation $R_1$. Inductively, given $i \in \N$ and having chosen symbols $R_1,\dots,R_i$ and morphisms $f_{1},\dots,f_{i}$, let $f_{i+1}$ be a morphism $f_{i+1}: \T_{\{E,R_1,\dots,R_{i+1}\}}$ where $R_{i+1}$ is some relation which is not preserved by $f_i$. By construction, $f_{i+1} \not\in \{f_0,\dots,f_i\}$ for all $i \in \N$, hence we have obtained an infinite family of pairwise distinct functions $|\T| \to |\B|$, which is absurd since $|\T|$ and $|\B|$ are finite sets.
\end{proof}

This correspondence between positive formulas and pp-trees allows us to give a simple, alternative proof of the one-way Hennessy-Milner-type property for \logic restricted to the case of finite structures. Once one has internalised the correspondence, the Hennessy-Milner-type property becomes an immediate consequence of the following elementary observation about coalgebras of an arbitrary comonad.

\begin{proposition}\label{prop:coalgebras_and_mappings}
    Let $G$ be a comonad on $\cat{E}$ and let $X, Y \in \cat{E}$. Then there exists a morphism $GX \to Y$ if and only if for all $G$-coalgebras $(Z, \gamma: Z \to GZ)$, if $Z$ maps into $X$ then it also maps into $Y$.
\end{proposition}
\begin{proof}
    Given a morphism $f: GX \to Y$, a coalgebra $(Z, \gamma: Z \to GZ)$ and a map $g: Z \to X$, we obtain a morphism $f \circ Gg \circ \gamma: Z \to Y$. In the other direction, take the cofree coalgebra on $X$, $(GX, \delta_X: GX \to GGX)$. Then from the morphism $\epsilon_X: GX \to X$ we obtain a morphism $GX \to Y$.
\end{proof}

\begin{theorem}[Theorem~\ref{thm:hennessy-milner} (1) for finite structures]\label{thm:hm_with_chandra-merlin}
    Let $\Ap$ and $\Bp$ be finite, pointed $\sigma$-structures. Then $\Ap \simu_k \Bp$ if and only if $\Ap \Rrightarrow_k^+ \Bp$.
\end{theorem}
\begin{proof}
Given $\Ap, \Bp$ as above,
\begin{align*}
    \Ap \simu_k \Bp &\text{ iff }
    \Ck \Ap \to \Bp \tag{Coro.~\ref{coro:logical_equivalence_kleisli}}\\
    &\text{ iff }
    \Ck^\text{f} \Ap \to \Bp \tag{Rmk.~\ref{rem:ppml_comonad_restricts_to_finite}}\\
    &\text{ iff }
    \forall \+T \in \EM(\Ck^\text{f}). \+T \to \Ap \implies \+T \to \Bp \tag{Prop.~\ref{prop:coalgebras_and_mappings}}\\
    &\text{ iff }
    \forall \phi \in \logicposk. \Ap \models \phi \implies \Bp \models \phi \tag{Cor.~\ref{coro:de_chandra_merlin}}\\
    &\text{ iff } \Ap \Rrightarrow_k^+ \Bp
\end{align*}
and this concludes the proof.\end{proof}

\section{Relating \boldmath{\logic} to Other Logics}\label{sec:ppml_and_other_logics}

\subsection{\boldmath{\logic} and First Order Logic}\label{sec:fragments_and_subcomonads}

In order to relate \logic and the \logic comonad to other well-known logics and their corresponding comonads, we begin by giving a standard translation, akin to that of \bml, from \logic to First Order Logic. As anticipated in Remark~\ref{rem:bounded_arity_signatures}, when $\sigma$ has bounded arity (and in particular when $\sigma$ is finite) this translation lands in a fragment of First Order Logic with bounded variable number.

Let $N$ be the maximum arity of relations in $\sigma$ if such a number exists, or $N = \infty$ otherwise. We fix an indexed set of first order variables, $\{x_i\}_{0\leq i\leq N}$ if $N$ is finite or $\{x_i\}_{0\leq i}$ otherwise.
We write
\[
\bar x = [x_{(j\!\!\! \mod N)}\ , x_{(j+1\!\!\! \mod N)}\ , \dots, \ x_{(j+\ell-1\!\!\! \mod N)}]
\]
for a cyclic sequence of variables of length $1 \leq \ell\leq N$, where $0 \leq j \leq N - 1$. If $N = \infty$, then $j \mod N$ is defined as $j$.

Given any such cyclic sequence $\bar{x}$, we define a mapping $\st_{\bar x}$ computable in polynomial time from $\sigma$-\logic formulas to First Order Logic formulas over the signature $\sigma$ in variable context $\bar{x}$ as follows:
\begin{align*}
\st_{\bar x}(\top) &\coloneqq \top \\
\st_{\bar x}(R) &\coloneqq \begin{cases}R(\last_{\arity(R)}(\bar x)) &\mbox{if $\arity(R)\leq|\bar x|$}\\
\bot&\mbox{if $\arity(R)>|\bar x|$}\end{cases}\tag{$R\in\bsigma$}\\
\st_{\bar x}(\lnot\varphi) &\coloneqq \lnot \st_{\bar x}(\varphi)\\
\st_{\bar x}(\varphi\land\psi) &\coloneqq \st_{\bar x}(\varphi) \land \st_{\bar x}(\psi)\\
\st_{\bar x}(\Diamond\varphi) &\coloneqq \exists y\ (\acc(x,y)\ \land \st_{\bar z}(\varphi))
\end{align*}
where $x = \epsilon(\bar x)$, $y$ is the next variable after $x$ in the cyclic order,
and $\bar z = (\last_{N-1}(\bar x).y)$. If $N = \infty$, we define $\last_{N-1}(s) = \last_N(s) = s$.

The following proposition follows immediately by structural induction.

\begin{proposition}\label{prop:PPML_FOL_tranlation}
For any $\sigma$-\logic formula $\varphi$ we have
$\A,a\models\varphi$ iff $\A\models\st_{[x_0]}(\varphi)[x_0\mapsto a]$. Furthermore $\st_{[x_0]}(\varphi)$ has at most $N$ variables and the depth of $\varphi$ is equal to the quantifier rank of $\st_{[x_0]}(\varphi)$.
\end{proposition}

Notice that when $\sigma$ is unimodal, we recover the standard translation for \bml whose image is the two-variable fragment of First Order Logic.

Proposition~\ref{prop:PPML_FOL_tranlation} allows us to treat $\sigma$-\logic as a fragment of First Order Logic over $\sigma$ with $N$ variables. On the other hand, it is immediate that a formula of modal depth $\ell$ is translated into a formula of quantifier rank $\ell$; hence the translation identifies \logick with a fragment of First Order Logic which is contained in First Order Logic with $N$ variables and quantifier rank at most $k$.

The comonadic formalism reflects this fact. For $\sigma$ of maximum arity $N < \infty$, $\Ck$ turns out to be a subcomonad of the comonads corresponding to these fragments. To be more precise, since these are comonads over $\Struct$ we must consider their liftings to the category $\Structpointed$.
We first recall the definition of subcomonad and the definitions of the EF and Pebbling comonads following~\cite{abramsky2021relating} as well as the combined comonad $\bb P_{k,n}$ introduced in~\cite{paine2020pebbling}.\footnote{For compatibility with our notation, we invert the names of the parameters from~\cite{paine2020pebbling}.}

\begin{definition}
    Given two comonads $(F, \epsilon^F, \delta^F)$ and $(G, \epsilon^G, \delta^G)$ over a common category, we say that $F$ is a \define{subcomonad} of $G$ whenever there exists a comonad morphism $F \Rightarrow G$ whose components are monomorphisms.\footnote{In the case of $\Structpointed$, these are the injective homomorphisms.}
\end{definition}

\begin{definition}\label{def:EF_and_pebbling_comonads}
Let $\sigma$ be any relational signature. The \define{Ehrenfeucht-Fraïssé comonad}~\cite{abramsky2021relating} with parameter $k \geq 0$, $\Ek: \Struct \to \Struct$ is defined as follows. For each $\sigma$-structure $\A$, define a new structure $\Ek \A$, with universe $\Ek \A \, \coloneqq \, |\A|^{\leq k}$. For each $\A$ we define a counit morphism $\epsilon^{\Ek}_\A$ and a comultiplication morphism $\delta^{\Ek}_\A$ by the same formulas as those of $\Ck$, given in Definition~\ref{def:ppml_comonad}. For each relation symbol $R$ of arity $r$, we define $R^{\Ek\A}$ to be the set of $r$-tuples $(s_1, \ldots , s_r)$ of sequences which (1) are pairwise comparable in the prefix ordering, and such that (2) $R^\A(\epsilon^{\Ek}_\A(s_1), \ldots , \epsilon^{\Ek}_\A(s_r))$. For each homomorphism $f: \A \to \B$, $\Ek f$ is defined by elementwise application of $f$ (also analogously to Def.~\ref{def:ppml_comonad}). This defines a comonad $\Ek$, which in turn lifts to a comonad $\Ek^*: \Structpointed \to \Structpointed$ by letting $\Ek^*\Ap \coloneqq (\Ek\A, [\bpa])$.

The \define{Pebbling comonad}~\cite{abramsky2021relating} with parameter $n \geq 1$, $\bb P_n: \Struct \to \Struct$ is defined as follows. Given a structure $\A$, define a new structure $\bb P_n \A$ with universe $(\{1,\dots,n\} \times |\A|)^{+}$, intuitively interpreted as the set of finite non-empty sequences of moves $(p, a)$ in an $n$-pebble game, where $p$ is a pebble index and $a \in |\A|$. The counit $\epsilon^{\bb P_n}_\A$ and the comultiplication $\delta^{\bb P_n}_\A$ have analogous definitions to those of $\Ck$ and $\Ek$, discarding and duplicating the information about pebble indexes respectively (in particular, $\epsilon^{\bb P_n}(s)$ is the position $a \in |\A|$ of the last move in $s$).
For each relation symbol $R$ of arity $r$, we define $R^{\bb P_n \A}$ to be the set of $r$-tuples $(s_1, \ldots , s_r)$ of sequences such that satisfy (1) and (2) as in the previous paragraph and for which moreover (3) the pebble index of the last move in each $s_i$ does not appear in the suffix of $s_i$ in $s_j$ for any $s_j$ extending $s_i$. For any homomorphism $f: \A \to \B$, $\bb P_n f$ is also defined by elementwise application of $f$. This defines a comonad $\bb P_n$ on $\Struct$, which in turn lifts to a comonad $\bb P_n^*: \Structpointed \to \Structpointed$ by letting $\bb P_n^*\Ap \coloneqq (\bb P_n\A, [(1, \bpa)])$.

Given $\A$ and some $k > 0$, let $\Pkn \A$ denote the embedded substructure of $\bb P_n \A$ with universe 
$$\{[(p_1, a_1), \dots, (p_\ell, a_\ell)] \in (\{1,\dots,n\} \times |\A|)^{\leq k} \mid p_1,\dots,p_{\min(\ell, n)} \text{ pairwise distinct}\}.$$
This restriction of universes defines a subcomonad of $\bb P_n$~\cite{paine2020pebbling} which we denote by $\Pkn$, and a subcomonad of $\bb P_n^*$ which we denote by $\Pkn^*$.\footnote{The restriction to sequences in which the first $n$ pebbles must be pairwise distinct is there to make $\bb P_{k,k}$ isomorphic to $\Ek$ as comonads. In this sense, both $\Ek$ and $\Pn$ are somehow expressible by $\Pkn$, although $\Ek$ is emphatically not a subcomonad of $\Pkn$.}

\end{definition}

\begin{propositionrep}\label{prop:subcomonads}
    $\Ck$ is a subcomonad of $\bb E_{k+1}^*$. Moreover, if $\sigma$ is a signature of bounded arity with maximum arity $N$, $\Ck$ is also a subcomonad of $\bb P_{k+1, N}^*$ and $\bb P_N^*$.
\end{propositionrep}
\begin{proof}
    For the EF comonad, it is straightforward that the universe inclusions $|\Ck \Ap| \to |\Ek \Ap|$ for every $\Ap$ induce injective homomorphisms and thus define a natural transformation with monic components.\footnote{Note that these universe inclusions do not define embeddings since they are not strong, as alluded to in Remark~\ref{rem:Ck_vs_Ek_vs_Hybrid}.} Compatibility with the counit and comultiplication is trivial since these are defined by the same formulas for both comonads.

    Now assume that there is some finite number $N$ which is the maximum arity over symbols in $\sigma$. Notice that for each $k \geq 0$, $\Pkns$ is a subcomonad of $\bb P_N^*$. Hence it is enough to establish the existence of a monic comonad morphism $\Ck \nat \Pkns$.

    For each $\Ap \in \Structpointed$ we define a map $\alpha_{\Ap}: \Ck \Ap \to \Pkns\Ap$ by $\alpha_{\Ap} ([a_0,\dots,a_\ell]) \coloneqq [(p_0, a_0), \dots, (p_\ell, a_\ell)]$ where $p_j \coloneqq (j+1) \textup{ mod $N$}$, e.g.\ for $\ell = 5$ and $N = 3$, $(p_0, p_1, p_2, p_3, p_4) = (1, 2, 3, 1, 2)$.

    This mapping is well defined since the length of the image sequence is $\ell + 1 < k + 1$ and the first $N$ pebbles are pairwise distinct. It also clearly preserves the basepoint; let us see that it preserves interpretations.

    Suppose $(s_1, \dots, s_{r}) \in R^{\Ck \A}$ for some $R \in \sigma$ of arity $r$. By definition of $R^{\Ck \A}$, we have that $(\epsilon(s_1),\dots,\epsilon(s_r)) \in R^{\A}$ and $s_{j+1}$ is the immediate successor of $s_j$ in the prefix ordering for all $j$. This means that $(\epsilon(\alpha_{\Ap}(s_1)),\dots,\epsilon(\alpha_{\Ap}(s_{r}))) \in R^{\A}$ and that the elements $\alpha_{\Ap}(s_j) \in \Pkns \Ap$ are pairwise comparable in the prefix ordering.
    The only remaining condition to be checked is that the pebble index of the last move in each $\alpha_{\Ap}(s_j)$ does not appear in the suffix of $\alpha_{\Ap}(s_j)$ in any sequence $\alpha_{\Ap}(s_j')$ extending $\alpha_{\Ap}(s_j)$. Since each $\alpha_{\Ap}(s_{j+1})$ is an immediate successor of $\alpha_{\Ap}(s_j)$, it is enough to check this for $j' = r$.
    If $|s_1| = \ell+1$, we may write
    \begin{align*}
        s_{r} &= [a_0,\dots,a_\ell,\dots,a_{\ell+r-1}] \\
        \alpha_{\Ap}(s_{r}) &= [(p_0, a_0),\dots,(p_\ell, a_\ell),\dots,(p_{\ell+r-1}, a_{\ell+r-1})].
    \end{align*}
    We must show that the last position where $p_j$ appears in $\alpha_{\Ap}(s_{r})$ is the $j$-th position, for all $\ell \leq j \leq \ell+r-1$. By definition of $\alpha_{\Ap}$, the next position where $p_j$ appears after position $j$ would be position $j + N$, but this will always be out of range because $r \leq N$.

    This proves that $\alpha_{\Ap}$ is a well defined morphism in $\Structpointed$. Naturality of $\alpha$ is precisely the fact that for any $f: \Ap \to \Bp$ and for any chain $[a_0,\dots,a_\ell] \in \Ck$, $\alpha([f(a_0),\dots,f(a_\ell)]) = \Pkns f (\alpha([a_0,\dots,a_\ell])$, and the preservation of counit and comultiplication amounts to a similar verification. Hence $\alpha$ is a comonad morphism. Finally, it is evident that the components of $\alpha$ are injective, hence monomorphisms in $\Stp$.

\end{proof}

The subcomonad inclusion $\Ck \hookrightarrow \bb P_{k+1, N}^*$ hints towards representing the \logic comparison games as restricted pebble games. Indeed, we see that elements of $\Ck\Ap$, when interpreted as sequences of Spoiler's moves in the \logic $k$-simulation game,
get translated to certain sequences of Spoiler's moves in the $k$-round $N$-pebble game. These sequences are constrained by the fact that Spoiler (and therefore, Duplicator as well) must move the $N$ pebbles in a cyclic pattern. Thus, the resulting embedded substructure of $\A$ which is picked out by the positions of the $N$ pebbles always consists in the last $N$ visited positions, with a particular ordering.

\subsection{\boldmath{\logic} and Data-Aware Logics}

\subsubsection{\logic and \datagl}\label{subsec:datagl}

We now return to one of our main motivations for this work: the study of data-aware logics. We will focus on explaining how \logic `contains' \datagl, and how this lets us conclude model-theoretical properties of \datagl using the \logic comonad. Throughout this discussion we will take the data-aware logic \cdxp~\cite{CoreDataXPath} as a point of reference, since it is expressive enough to contain as fragments other data-aware logics of interest.

Indeed, \datagl captures a fragment of $\cdxp(\downarrow^+)$, i.e.\ \cdxp with the `descendant' accessibility relation. In~\cite{dataGL}, \datagl is presented as a modal logic with two different modal operators, $\Diamond_=$ and $\Diamond_{\neq}$, which are called \emphat{data-aware modalities} since their associated accessibility relations contain and indeed encapsulate all of the information about data values that can be accessed by this language, namely checking for equality.

\begin{definition}\label{def:datagl_syntax}
    The syntax of \datagl is that of a modal logic with two modalities, $\Diamond_=$ and $\Diamond_{\neq}$, namely
    \begin{align*}
    \varphi & \ \Coloneqq  
    p						\ \mid \ 
    \lnot\varphi            \ \mid \ 
    \varphi\land\varphi     \ \mid \ 
    \Diamond_=\varphi       \ \mid \ 
    \Diamond_{\neq}\varphi 
    \tag{$p \in \text{PROP}$}
    \end{align*}
    where $\text{PROP}$ is a finite set of unary symbols.
\end{definition}

Following the discussion in \cite[Section 2.2.2]{dataGL}, 
even though \datagl is originally taken to predicate over finite data trees, thanks to a tree-model property \cite[Prop. 4]{dataGL} we can instead choose to work with equivalent semantics based on data Kripke structures, which we do since it is more general. This does not introduce any significant differences in the results to be presented below. To this end, fix a countably infinite set of data values $\bb D$.

\begin{definition}\label{def:datagl_model_and_semantics}
    A \define{data Kripke structure} is a tuple $\MM = \tup{W, R, d, \ell}$ where $\tup{W, R}$ is a directed graph specified by a set $W$ and a binary relation $R$, $d: W \to \bb D$ is a function labelling each node with a data value from a countably infinite set, and $\ell: W \to 2^\text{PROP}$ labels each node $w$ with a subset $\ell(p) \subseteq \text{PROP}$ of atomic propositions, which are said to hold at $w$.

    A \define{\datagl model} is a data Kripke structure $\MM = \tup{W, R, d, \ell}$ where $W$ is finite and $R$ is transitive irreflexive.

    Given a \datagl model $\MM = \tup{W, R, d, \ell}$ and $w \in W$, the semantics of \datagl are defined by
    \begin{align*}
        \MM,w &\models p &\text{\quad iff\quad } &p \in \ell(w) \tag{$p \in \text{PROP}$}\\
        \MM,w &\models \lnot\varphi&\text{\quad iff\quad } &\MM,w\not\models \varphi\\
        \MM,w &\models \varphi\land\psi&\text{\quad iff\quad } &\MM,w\models \varphi \text{ and } \MM,w \models \psi\\
        \MM,w &\models \Diamond_=\varphi   &\text{\quad iff\quad } &\exists w'\in W. (w,w')\in R^{(=)}  \text{ and }\MM,w'\models \varphi\\
        \MM,w &\models \Diamond_{\neq}\varphi   &\text{\quad iff\quad } &\exists w'\in |\A|.(w,w')\in R^{(\neq)} \text{ and }\MM,w'\models \varphi
    \end{align*}
    where $R^{(=)} \coloneqq \{(w, w') \in R \mid d(w) = d(w')\}$ and $R^{(\neq)} \coloneqq \{(w, w') \in R \mid d(w) \neq d(w')\}$.
\end{definition}

\begin{remark}
    The requirement that the relation $R$ in a \datagl model as above is transitive irreflexive generalises the definition of data trees, in which $R$ is the transitive closure of the successor relation of a tree.
\end{remark}

As we see in Def.~\ref{def:datagl_model_and_semantics}, the accessibility relations for $\Diamond_=$ and $\Diamond_{\neq}$ are defined by intersecting two relations: an accessibility relation representing the underlying graph structure, and a data-derived relation which encapsulates either data equality or non-equality. In contrast, by thinking about \datagl from the point of view of its bisimulation game, it becomes natural to represent this language in a different way: instead of encapsulating information about equality and non-equality of data values through two different modalities, which in turn depend on data-aware accessibility relations, allow the data-aware relations as literals of the syntax.
Let $\sigmadgl \coloneqq \{\acc, R_=\} \cup \text{PROP}$ where $\acc$ and $R_=$ are binary. The resulting \logic syntax
\begin{align*}
\varphi & \ \Coloneqq \top              \ \mid \ 
                R_=                 \ \mid \ 
                p                   \ \mid \
                \lnot\varphi        \ \mid \ 
                \varphi\land\varphi \ \mid \ 
                \Diamond\varphi                 \tag{$p \in \text{PROP}$}.
\end{align*}
consists of a language with a single modal operator and a separate binary relation symbol $R_=$ expressing data equality.

We adopt the perspective that data-aware logics predicate not over data graphs but over relational structures obtained by forgetting the actual data values and retaining only the information about how these values relate to each other according to the comparison operations of our language. This motivates the following definition.

\begin{definition}\label{def:datagl_model_to_sigma_structure}
    Given a \datagl model $\MM = \tup{W, R, d, \ell}$, we define a $\sigmadgl$-structure $\tMM$ as follows:
    \begin{align*}
        |\tMM| &\coloneqq W, &p^\tMM &\coloneqq \{w \in W \mid p \in \ell(w)\}, \tag{$p \in$ PROP} \\
        E^\tMM &\coloneqq R, &R_=^\tMM &\coloneqq \{(w,w') \in W^2 \mid d(w) = d(w')\}. 
    \end{align*}
\end{definition}

This defines a mapping $t$ from the collection of \datagl models to $\sigmadgl$-structures (see Figure \ref{fig:datakripke_to_rel_structure}). This mapping is obviously not injective since many data assignments give rise to the same equal data relation, and yet it does not lose any information relevant to \datagl, as is made precise by the following result. We will return to this issue of encapsulating data values in Remark~\ref{rem:replacing_data_values_with_equal_data}.

\begin{figure}[h]
\centering
\includegraphics[scale=0.3]{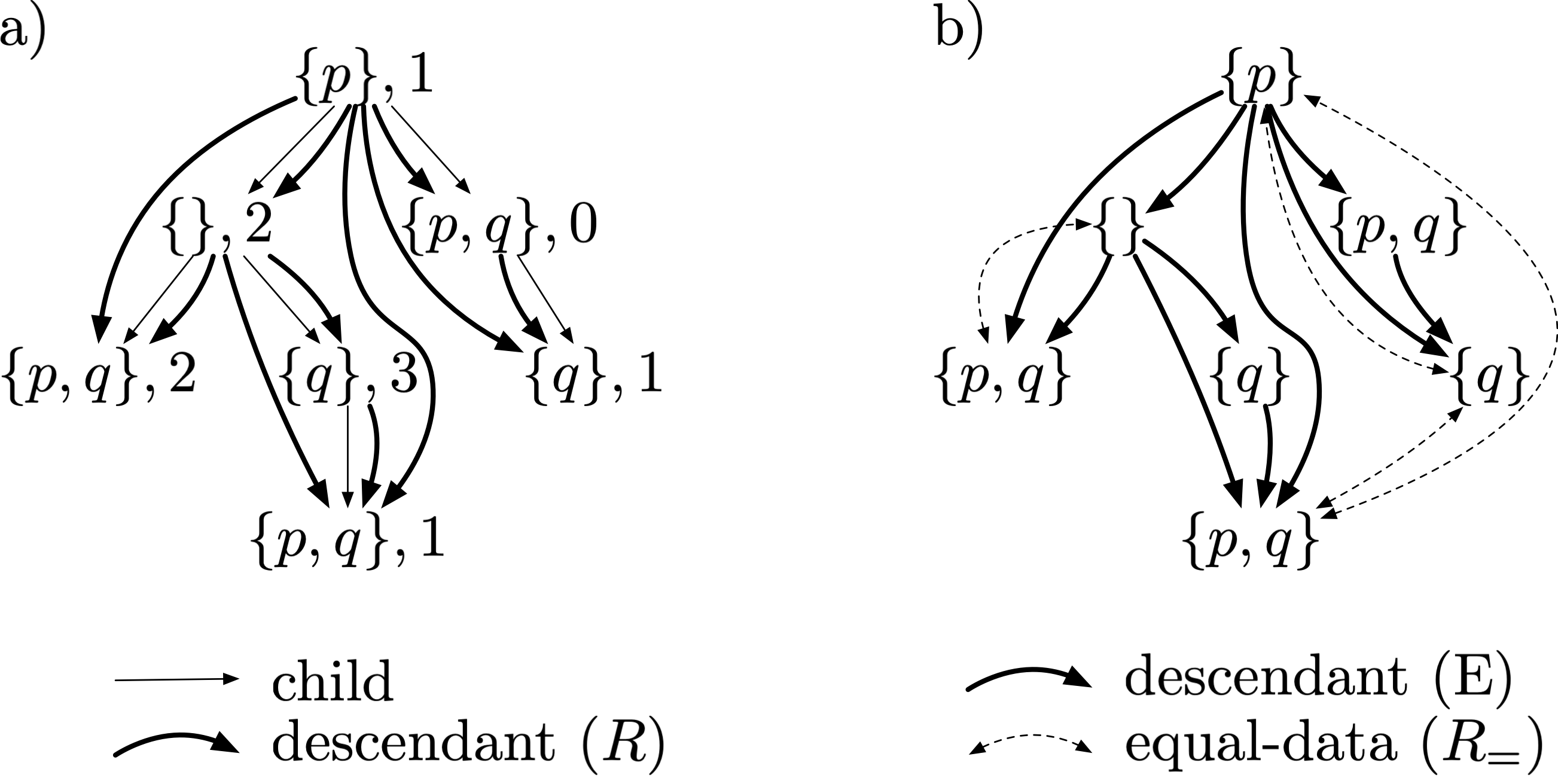}
\caption{a) An example of a \datagl model $\MM$ (in particular a data tree) with $\text{PROP} = \{p, q\}$ and $\bb D$ equal to the set of integers. Notice how $\MM$ is a valid \datagl model (Def.~\ref{def:datagl_model_and_semantics}) since $R$ is the transitive closure of the accessibility relation of a tree, and hence it is transitive irreflexive. b) The corresponding $\sigmadgl$-structure $\tMM$ (Def.~\ref{def:datagl_model_to_sigma_structure}), where each point is labelled with the set of unary relation symbols that are true at that point.
}
\label{fig:datakripke_to_rel_structure}
\end{figure}

\begin{toappendix}
For simplicity, we add the falsum symbol $\bot$ to the language of $\logic$, which of course does not increase the expressive power. We also write $\phi \lor \psi$ as syntactic sugar for $\lnot(\lnot \phi \land \lnot \psi)$.
Given a \logic-formula $\varphi$ and given $*\in\{\top,\bot\}$, let $\varphi\subs *$ be the result of replacing every occurrence of $R_=$ which is not in the scope of a $\Diamond$ with $*$. Formally, 
\begin{align*}
\top\subs *&\coloneqq \top\\
R_=\subs *&\coloneqq *\\
p\subs *&\coloneqq p \tag{$p \in \text{PROP}$}\\
(\lnot\varphi)\subs *&\coloneqq \lnot(\varphi\subs *)\\
(\varphi\land\psi)\subs *&\coloneqq (\varphi\subs *)\land(\psi\subs *)\\
(\Diamond\varphi)\subs *&\coloneqq \Diamond\varphi.
\end{align*}

\begin{lemma}\label{lem:traduccion}
For any \datagl model $\A$, for any $(a, a') \in E^\A$ and any $\sigmadgl$-\logic-formula $\varphi$ we have that
$\A,aa'\models\varphi$ iff the following two conditions hold:
\begin{itemize}
\item $(a,a')\in R_=^{\A}$ implies $\A,a'\models\varphi\subs\top$, and 
\item $(a,a')\notin R_=^{\A}$ implies $\A,a'\models\varphi\subs\bot$.
\end{itemize}
\end{lemma}
\begin{proof}
We proceed by induction in $||\varphi||$, the complexity $\varphi$, defined by
\begin{align*}
||R_i||=||\top||=||\bot||&=0\\
||\lnot\varphi||=||\Diamond\varphi||&=1+||\varphi||\\
||\varphi\land\psi||&=1+||\varphi||+||\psi||
\end{align*}
If $\varphi=p \in P$ then the result follows from the fact that $\A,aa'\models\varphi$ iff $\A,a'\models\varphi$ and that $\varphi\subs\top=\varphi\subs\bot=\varphi$. The same argument applies if $\varphi$ is of the form $\Diamond\psi$, $\top$ or $\bot$. If on the other hand $\varphi=R_=$, then $\varphi\subs\top=\top$ and $\varphi\subs\bot=\bot$; 
$\A,aa'\models\varphi$ iff $(a,a')\in R_=^{\A}$ and from this the result follows straightforwardly. The case $\phi = \psi_1 \land \psi_2$ is immediate by the inductive hypothesis. Finally, if $\varphi=\lnot\psi$ then 
\begin{align*}
\A,aa'\models\varphi&\text{ iff } \A,aa'\not\models\psi
\\
&\text{ iff } (a,a')\in R_=^{\A} \text{ and }\A,a'\not\models\psi\subs\top, \text{ or } \\
&\qquad(a,a')\notin R_=^{\A} \text{ and }\A,a'\not\models\psi\subs\bot\tag{ind.\ hyp.}
\\
&\text{ iff } (a,a')\in R_=^{\A}\text{ and }\A,a'\models\varphi\subs\top, \text{ or } \\
&\qquad(a,a')\notin R_=^{\A}\text{ and }\A,a'\models\varphi\subs\bot\tag{$\dagger$}
\\
&\text{ iff } (a,a')\in R_=^{\A}\text{ implies }\A,a'\models\varphi\subs\top, \text{ and } \\
&\qquad(a,a')\notin R_=^{\A}\text{ implies }\A,a'\models\varphi\subs\bot.
\end{align*}
Observe that $||\psi\subs\top||=||\psi\subs\bot|| < ||\varphi||$ and hence\footnote{This is the technical reason why we insist in adding $\bot$ as another primitive of the language instead of defining it as syntactic sugar for $\lnot \top$.} we can apply the inductive hypothesis; and $(\dagger)$ holds by the fact that $\varphi=\lnot\psi$ and the definition of $(\lnot\varphi)\subs *$.
\end{proof}
\end{toappendix}

\begin{theoremrep}\label{thm:datagl-equiexpressive}
\datagl and $\sigmadgl$-\logic are equi-expressive over the class of \datagl models. More precisely, there is a translation $\tr_1$ mapping \datagl-formulas to $\sigmadgl$-\logic-formulas and a translation $\tr_2$ in the reverse direction such that for any \datagl model $\MM = \tup{W, R, d, \ell}$ and $w \in W$,
\begin{align*}
    \MM, w \models\varphi &\text{\qquad iff\qquad } \tMM, w \models\tr_1(\varphi), \\
    \MM, w \models \tr_2(\psi) &\text{\qquad iff\qquad } \tMM, w \models \psi.
\end{align*}
Moreover, both translations preserve modal depth.
\end{theoremrep}

\begin{proof}
Define the translation $\tr_1$ of \datagl into $\sigmadgl{\text -}\logic$ as follows:
\begin{align*}
\tr_1(p)&=p\tag{$p \in \text{PROP}$}\\
\tr_1(\lnot\varphi)&=\lnot\tr_1(\varphi)\\
\tr_1(\varphi\land\psi)&=\tr_1(\varphi) \land \tr_1(\psi)\\
\tr_1(\Diamond_=\varphi)&=\Diamond(R_=\land\tr_1(\varphi))\\ 
\tr_1(\Diamond_{\neq}\varphi)&=\Diamond(\lnot R_=\land\tr_1(\varphi)) 
\end{align*}

A straightforward induction in the structure of \datagl formulas shows that for any \datagl model $\MM = \tup{W,R,d,\ell}$, any $w\in W$ and any \datagl-formula $\phi$, we have $\MM,w\models\varphi$ iff $\tMM, w\models\tr_1(\phi)$.

Next, define a translation $\tr_2$ of $\sigmadgl$-\logic into \datagl as follows:
\begin{align*}
\tr_2(\top) &= \top \\
\tr_2(\bot) &= \bot \\
\tr_2(R_=)&=\bot\\
\tr_2(p)&=p\tag{$p \in \text{PROP}$}\\
\tr_2(\lnot\varphi)&=\lnot\tr_2(\varphi)\\
\tr_2(\varphi\land\psi)&=\tr_2(\varphi)\land\tr_2(\psi)\\
\tr_2(\Diamond\varphi)&=\left(\Diamond_=\tr_2(\varphi\subs\top)\right) \lor \left(\Diamond_{\neq}\tr_2(\varphi\subs\bot)\right) 
\end{align*}
where the expressions $\top$ and $\bot$ on the right-hand side are to be regarded as syntactic sugar for some \datagl-tautology and its negation such as $\lnot(q \land \lnot q)$ and $q \land \lnot q$, respectively, for some fixed choice of $q \in \text{PROP}$.

We must show that for any \datagl model $\MM = \tup{W,R,d,\ell}$, any $w\in W$ and any $\sigmadgl$-\logic formula $\phi$, we have $\tMM, w \models\varphi$ iff $\MM, w \models\tr_2(\varphi)$.
We proceed by induction in $||\varphi||$, defined in the proof of Lemma~\ref{lem:traduccion}.
The only interesting case is when $\varphi=\Diamond\psi$, where we reason as follows:

\begin{align*}
\tMM,w\models\Diamond\psi&\text{ iff } \exists w'.w \prec w'\text{ such that }\tMM,aa'\models\psi
\\ 
&\text{ iff }\exists w'.w \prec w' \text{ such that } \\
&\quad\qquad (w,w')\in R_=^{\tMM} \text{ and }\tMM, w'\models\psi\subs\top, \text{ or } \\
&\quad\qquad(w,w')\notin R_=^{\tMM} \text{ and }\tMM,w'\models\psi\subs\bot\tag{Lemma~\ref{lem:traduccion}}
\\
&\text{ iff }\exists w'.w \prec w' \text{ such that } \\
&\quad\qquad (w,w')\in R_=^{\tMM} \text{ and }\MM,w'\models\tr_2(\psi\subs\top), \text{ or } \\
&\quad\qquad(w,w')\notin R_=^{\tMM} \text{ and }\MM,w'\models\tr_2(\psi\subs\bot)\tag{ind.\ hyp.}
\\
&\text{ iff }\exists w'. (w,w') \in R^{(=)} \text{ such that } \MM,w'\models\tr_2(\psi\subs\top), \text{ or } \\
&\qquad\exists w'. (w,w') \in R^{(\neq)} \text{ such that } \MM,w'\models\tr_2(\psi\subs\bot)
\\
&\text{ iff }\MM,w\models \left(\Diamond_{=}\tr_2(\psi\subs\top)\right) \lor \left(\Diamond_{\neq}\tr_2(\psi\subs\bot)\right)\tag{def.\ semantics}
\\
&\text{ iff }\MM,w\models\tr_2(\Diamond\psi).
\end{align*}

Finally, it is immediate from the definition of both translations that they preserve modal depth.
\end{proof}

We now extend $t$ to a functor from pointed \datagl models into $\Stpgl$. To this end, we must define the notion of morphism of \datagl models, for which we have some freedom.%
\footnote{Technically, we do not need to regard $t$ as a functor for our main result in this section (Theorem~\ref{thm:full_logical_equivalence_datagl}). Equivalently, we could consider $t$ as a functor out of a discrete category, meaning that the only morphisms between \datagl models would be the identities. Our definition of morphisms of \datagl models is motivated by trying to explain as best as possible how \logic generalises \datagl.}
We define morphisms of \datagl morphisms as structure-preserving functions obtained from the requirement that they preserve the truth of \datagl formulas without the symbols $\{\lnot, \Diamond_{\neq}\}$. Although, of course, other choices are possible, this choice of `positive fragment' is compatible with our translation into \logic, as our next theorem shows.

\begin{definition}\label{def:functor_from_datagl_models}
    We denote by $\ModDGL_*$ the category whose objects are \datagl models together with a choice of basepoint and whose morphisms are given as follows. Given pointed \datagl models $\MMp = \tup{W, R, d, \ell, w}$ and $\MMpp = \tup{W', R', d', \ell', w'}$, a morphism $f: \MMp \to \MMpp$ is a function $f: W \to W'$ such that $f(w) = w'$ and for all $w_1, w_2 \in W$, $(w_1, w_2) \in R \implies (f(w_1), f(w_2)) \in R'$, $d(w_1) = d(w_2) \implies d'(f(w_1)) = d'(f(w_2))$ and $\ell(w_1) \subseteq \ell'(f(w_1))$.

    We define a functor $t: \ModDGL_* \to \Stpgl$ as follows. It is defined on objects by the construction of Definition~\ref{def:datagl_model_to_sigma_structure}, extended to pointed models by declaring the basepoint of $\tMMp$ to be $w$. On morphisms, it takes $f : \MMp \to \MMpp$ to the pointed homomorphism of $\sigmadgl$-structures $\tMMp \to \tMMpp$ whose underlying function $W \to W'$ is the underlying function of $f$.\footnote{There is an implicit verification to be made that given $f$, $tf$ is a well defined homomorphism. Notice how, once well definition is established, functoriality of such a mapping is immediate.}
\end{definition}

\begin{proposition}\label{prop:t_is_ff}
    $t: \ModDGL_* \to \Stpgl$ is fully faithful, and its image consists of all the finite $\sigmadgl$-structures for which the interpretation of $E$ is transitive irreflexive and the interpretation of $R_=$ is an equivalence relation.
\end{proposition}
\begin{proof}
    A straightforward verification shows that $t$ is fully faithful. It is also immediate that any $\sigmadgl$-structure in the image of $t$ satisfies the conditions stated above. Now consider any $\sigmadgl$-structure $\Ap$ satisfying those conditions and choose an ordering $a_1,\dots,a_n$ of $|\A|$. Without loss of generality we consider $\bb D = \bb N$ to be the non-negative integers, and define a data assignment $d: |\A| \to \bb N$ inductively as
    \begin{align*}
        d(a_1) &\coloneqq 0 \\
        d(a_{j+1}) &\coloneqq \begin{cases}
            d(a_{j'}) & \text{if $(a_{j'}, a_j) \in R_=^\A$ for some $j'\leq j$} \\
            d(a_j) + 1 & \text{otherwise.}
        \end{cases}
    \end{align*}
    This is well defined since $R_=$ is an equivalence relation. Let $\ell(a) \coloneqq \{p \in \text{PROP} \mid a \in p^\A\}$. It is then straightforward that $\MM \coloneqq \tup{|\A|, E^\A, d, \ell}$ is a \datagl model and that $t (\MM, a) = \Ap$.
\end{proof}

Since $t$ is a fully faithful functor, $\ModDGL_*$ can be identified with the image of $t$, which we denote by $t(\ModDGL_*)$ and which is a full subcategory of $\Stpgl$.\footnote{Since any $\sigmadgl$-structure isomorphic to one in $t(\ModDGL_*)$ is also in $t(\ModDGL_*)$, the image of $t$ also coincides with what is known as its \emph{essential} image.}

\begin{remark}\label{rem:replacing_data_values_with_equal_data}    
    Proposition~\ref{prop:t_is_ff} completes the argument that we can safely replace the data assignment functions with a relation $R_=$ encapsulating the relevant or operationally accessible information. Indeed, combined with Theorem~\ref{thm:datagl-equiexpressive}, it shows that we may safely identify \datagl models with $\sigmadgl$-structures in which $E$ is transitive irreflexive and $R_=$ is an equivalence relation. From this perspective, Theorem~\ref{thm:datagl-equiexpressive} says that the only essential difference between \datagl and $\sigmadgl$-\logic is that the latter admits more general models.

    On the other hand, although $t$ is not injective on objects, as we have already noted, since the corestriction $t|^{t(\ModDGL_*)}: \ModDGL_* \to t(\ModDGL_*)$ is fully faithful and surjective on objects, is an equivalence of categories \cite[Def. 1.5.4]{riehl2017category}. This means that from the categorical point of view\footnote{More precisely, working up to equivalence of categories corresponds to working up to isomorphism of objects.} there is no loss of information when moving from $\ModDGL_*$ to $t(\ModDGL_*)$ (or viceversa).
\end{remark}

Note that if $\MMp \in \ModDGL_*$, then $\Ck \tMMp$ is not in $t(\ModDGL_*)$ except for trivial models: $\acc^{\Ck \tMM}$ will not in general be transitive irreflexive and $R_=^{\Ck \tMM}$ will not be an equivalence relation. This implies that $\Ck$ cannot be restricted to a comonad on $t(\ModDGL_*)$. However, this does not stop us from tapping into the comonadic formalism.

Given $\MMp, \MMpp \in \ModDGL_*$, we write 
$\MM,w\equiv_k\MM',w'$ if $\MM,w\models\varphi \iff \MM',w'\models\varphi$ for all $\varphi$ in $\datagl_k$
where $\datagl_k$ is the fragment of \datagl of modal depth at most~$k$.

\begin{theorem}\label{thm:full_logical_equivalence_datagl}
    Let $\MM = \tup{W, R, d, \ell}$ and $\MM' = \tup{W', R', d', \ell'}$ be \datagl models and let $w \in W$ and $w' \in W'$. Then
    \begin{enumerate}
        \item $\MM, w \equiv_k \MM', w'$ if and only if there exists a span of bounded morphisms $\Ck (\tMM, w) \leftarrow \+T \to \Ck (\tMM', w')$ with some pp-tree $\+T$ of height at most $k$ as common domain; and
        \item there exists a homomorphism $\Ck \tMMp \to \tMMpp$ if and only if $\MM, w \models \phi$ implies $\MM', w' \models \phi$ for all formulas $\phi$ in the translation of \logicposk under $\tr_2$ (defined in the proof of Theorem~\ref{thm:datagl-equiexpressive}).
    \end{enumerate}
\end{theorem}
\begin{proof}
    Immediate from Theorem~\ref{thm:datagl-equiexpressive} and Corollary~\ref{coro:full_logical_equivalence}.
\end{proof}

Notice that the translation of $\logicpos$ under $\tr_2$ does not coincide with \datagl without $\{\lnot, \Diamond_{\neq}\}$, nor with its negation-free fragment, as can be seen for example in the translation $\tr_2(\Diamond p) = \lnot(\lnot\Diamond_= p \land \lnot\Diamond_{\neq} p)$. Thus we have obtained a correspondence between certain Kleisli morphisms and a fragment of \datagl which is not obviously interpretable as a positive fragment, at least not from the point of view of \datagl syntax alone.

The procedure we have put in practice to derive results about \datagl from a translation into \logic is representative of a general technique, which, for instance, is used to add equality to First Order Logic in the context of the EF and Pebbling comonads. In~\cite{abramsky2021relating}, this is originally presented in terms of \emphat{relative comonads}~\cite{altenkirch2010monads}, a theoretical device which equips a functor between two different categories with comonad-like structure. For instance, in our case we can define \datagl relative comonad $\bb D_k: \ModDGL_* \to \Stpgl$ whose underlying functor is $\Ck \circ t$ and which automatically inherits a relative comonad structure from the comonad structure of $\Ck$ \cite[Prop. 2.3]{altenkirch2010monads}. This encapsulates the comonad and the translation functor $t$ in a single mathematical object. However, the concept of relative comonad is not required at a technical level for the core of the procedure, just as presented in this section and other works (e.g.~\cite{dawar2021lovasz},~\cite{jakl2023categorical}).

\subsubsection{\logic as Framework for Other Data-Aware Logics}\label{sec:ppml_as_framework}

The idea underlying the description of \datagl as a particular case of \logic (rather than as a Modal Logic with two modalities) is to split each data-aware modality into two different syntactic building blocks. This idea can be applied more generally and in this sense we propose that \logic gives a flexible approach to data-aware logics. Practically speaking, this approach allows us both to express previously existing data-aware logics such as \datagl or other fragments of \cdxp and to build new data-aware logics starting from these more fundamental language components. We illustrate this with an example.

\begin{example}\label{ex:more_expressive}\rm
Let $\sigma = \sigmadgl \union \{T\}$ where $T$ is ternary. We interpret $\sigma$-\logic over structures $\A$
in which the interpretation of $\acc$ is transitive irreflexive, that of $R_=$ is an equivalence relation, and $T^\A \coloneqq \{(x,y,z):(x,z)\in R_=^\A\}$.
Following Remark~\ref{rem:replacing_data_values_with_equal_data}, we think of these structures as $T$-expanded \datagl models. Given the requirement on the interpretation of $E$, we do not distinguish between successors and strict descendants.

Consider the $\sigma$-\logic formula $\psi_1 \coloneqq \Diamond (\lnot R_= \land \Diamond (\lnot R_= \land \lnot T))$, which is the formula of Example~\ref{ex:ppml_first_examples} where we have renamed $S$ to $R_=$. $\psi_1$ evaluated at a point $a$ of a structure $\A$ as above expresses the existence of a descendant $a'$ and a descendant $a''$ of $a'$ such that $a$, $a'$ and $a''$ have pairwise distinct data values.

Let us see that $\psi_1$ is not expressible in \datagl. To this end, recall the structures $\A$ and $\B$ from Example~\ref{ex:bisimilar_structures}(3), which we now interpret as $\sigmadgl$-structures (by renaming $S$ to $R_=$ and giving empty extensions to propositional variables; see Figure~\ref{fig:psi1_datagl}-i)). Since $\Ap \equiv_k \Bp$ for all $k$, by Theorem~\ref{thm:datagl-equiexpressive} $\Ap$ and $\Bp$ are also indistinguishable in \datagl. If we now expand these models with relations $T^\A, T^\B$ as above, then $\Ap$ satisfies $\psi_1$ but $\Bp$ does not. We conclude that the property expressed by $\psi_1$ is not expressible in \datagl, and hence that $\sigma$-\logic is a strictly more expressive extension of \datagl.

\begin{figure}[h]
\centering
\includegraphics[scale=0.275]{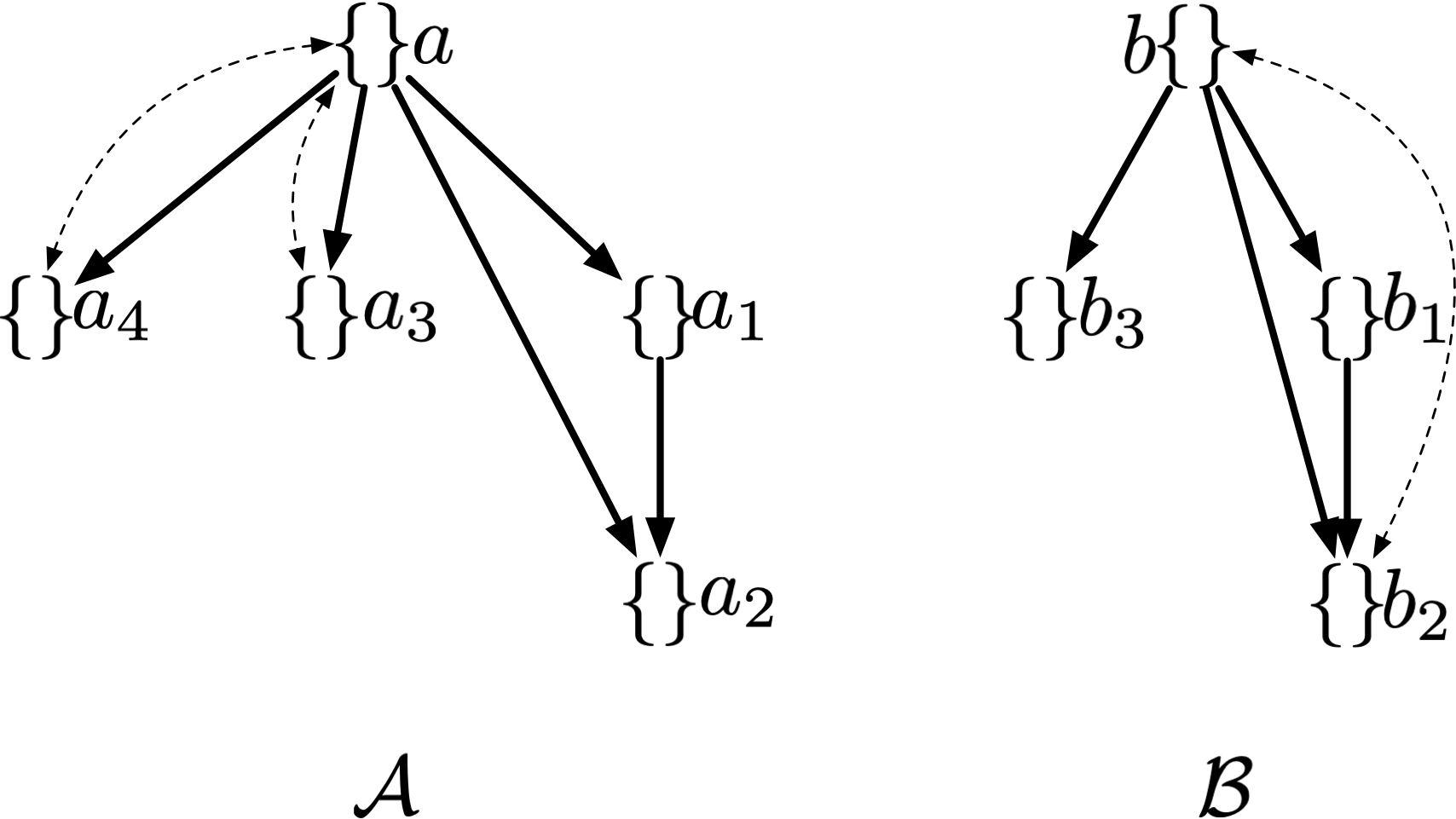}
\caption{$\sigmadgl$-structures which are $k$-bisimilar for all $k$. Bold arrows represent $\acc$ and dotted arrows represent $R_=$. 
$\psi_1 \coloneqq \Diamond (\lnot R_= \land \Diamond (\lnot R_= \land \lnot T))$ is true in $\Ap$ but false in $\Bp$, hence $\psi_1$ is not expressible in $\sigmadgl$-\logic. 
} 
\label{fig:psi1_datagl}
\end{figure}

Since our motivation for studying \datagl came from regarding it as a simple fragment of $\cdxp(\downarrow^+)$, we may wonder whether $\sigma$-\logic is also a fragment of $\cdxp(\downarrow^+)$. This is not the case, as $\psi_1$ is not expressible in $\cdxp(\downarrow^+)$ \cite[Proposition 39.2]{FFA15}. 
\end{example}

Another relevant data-aware logic of interest for \logic is the fragment of $\cdxp(\dow^+)$ restricted to data comparisons of the form $\tup{\epsilon=\dow^+[\phi_1]\dow^+[\phi_2]\dow^+\cdots\dow^+[\phi_n]}$ or of the form $\tup{\epsilon\neq\dow^+[\phi_1]\dow^+[\phi_2]\dow^+\cdots\dow^+[\phi_n]}$, which we call {\boldmath $\cdxp_\epsilon(\dow^+)$}.
As with \datagl, this logic predicates over pointed \datagl models $\MM = \tup{W, R, d, \ell}$ and extends \datagl to paths with intermediate tests. In a nutshell, for for $w \in W$ the semantics of $\tup{\epsilon=\dow^+[\phi_1]\dow^+[\phi_2]\dow^+\cdots\dow^+[\phi_n]}$ [resp.\ $\tup{\epsilon\neq\dow^+[\phi_1]\dow^+[\phi_2]\dow^+\cdots\dow^+[\phi_n]}$] in $\MM,w$ is ``there is an $E$-chain $[w_0,\dots,w_n]$ in $\MM$ such that 1) $w_0=w$, 2) $\MM,w_i\models\phi_i$ for $i=1,\dots,n$, and 3) $d(w)=d(w_n)$ [resp.\ $d(w)\neq d(w_n)$]''.
The idea of the ternary relation of Example~\ref{ex:more_expressive} can be generalized to $n$-ary relations $R_n$ for $n\geq 2$. One can extend the construction of $\tMM$ given in Definition~\ref{def:datagl_model_to_sigma_structure} to the signature $\sigmacdxp=\sigmadgl \union \{R_n \mid n\geq 2\}$ as follows:
\begin{align*}
        R_n^\tMM &\coloneqq \{(w_1,\dots,w_n) \in W^n \mid d(w_1) = d(w_n)\}.
\end{align*}
One can also modify the translation of $\tr_1$ given in the proof of Theorem~\ref{thm:datagl-equiexpressive} in order to map $\cdxp_\epsilon(\dow^+)$-formulas to $\sigmacdxp$-\logic formulas as follows:
\begin{align*}
\tr_1(\tup{\epsilon=\dow^+[\phi_1]\dow^+[\phi_2]\dow^+\cdots\dow^+[\phi_n]})&\coloneqq\Diamond(\tr_1(\phi_1)\land\Diamond(\tr_1(\phi_2)\land\Diamond(\dots(\tr_1(\phi_n)\land R_n)\dots ) )),
\\
\tr_1(\tup{\epsilon\neq\dow^+[\phi_1]\dow^+[\phi_2]\dow^+\cdots\dow^+[\phi_n]})&\coloneqq\Diamond(\tr_1(\phi_1)\land\Diamond(\tr_1(\phi_2)\land\Diamond(\dots(\tr_1(\phi_n)\land \lnot R_n)\dots ) )). 
\end{align*}
As with \datagl,  one can show that $\MM, w \models\varphi$ iff $\tMM, w \models\tr_1(\varphi)$ for any formula $\varphi\in\cdxp_\epsilon(\dow^+)$. 
However, in this case one can also show that there is no translation $\tr_2$ from $\cdxp_\epsilon(\dow^+)$ to $\sigmacdxp$-\logic such that $\MM, w \models\tr_2(\varphi)$ iff $\tMM, w \models \varphi$. Hence, over the appropriate classes of structures, $\sigmacdxp$-logic is strictly more expressive than $\cdxp_\epsilon(\dow^+)$.

As a final remark, observe that there is nothing special with the fact that the navigation axis of $\cdxp_\epsilon$ is the `descendant' relation $\downarrow^+$; similar results may be obtained with the `child' relation $\downarrow$ instead.

\subsection{\boldmath{\logic} and Basic Modal Logic}\label{sec:relationship_modal}\label{sec:bml}

\logic shares many properties with \bml. It contains \bml for particular choices of $\sigma$, it has essentially the same syntax as \bml (aside for the information about arities of symbols) and has a tree-model property, owing to its game comonad being idempotent, just like the one for \bml. Although the similarities are evident, a direct comparison is hindered by the fact that these two logics predicate, in the general case, over different classes of models.

In what follows, we will give a way of transforming pp-trees into Kripke trees. This transformation will preserve and reflect open pathwise embeddings, as well as the truth value of formulas in a suitable sense. The first property will allow us to reduce checking $k$-bisimilarity in \logic to checking $k$-bisimilarity between Kripke trees of height at most $k$. Meanwhile, the second property of this transformation will give polynomial reductions from the model checking and satisfiability problems for $\logic$ to those for \bml.

\begin{definition}
    Given a signature $\sigma$ with $E \in \sigma$, we define a new signature $\tilde{\sigma}$ by replacing its arity function $\arity_\sigma: \sigma \to \N$ with
    $$\arity_{\tilde{\sigma}}(R) \coloneqq \begin{cases}
        2 & \text{if } R = E \\
        1 & \text{otherwise.}
    \end{cases}$$
    In other words, $\tilde{\sigma}$ has the same relation symbols as $\sigma$ but all symbols except $E$ are now considered to be unary. Notice that $\tilde{\sigma}$ is a unimodal signature.
    
    Moreover, given a pp-tree $\T$ over $\sigma$, let $K\T$ be the $\tilde{\sigma}$-structure with universe $|K\T| \coloneqq |\T|$, basepoint $u$, and the following relations: $\acc^{K\T} \coloneqq \acc^{\+ T}$, and for $R \in \bsigma$, $R^{K\T} \coloneqq \{\epsilon(s) \mid s \in R^\T\}$. Notice that, since $\T$ is a pp-tree, this is equivalent to saying that for all $v \in |\+ T|$, $v \in R^{K\T}$ if and only if $\T, \T_v \models R$.
    In particular, if $v \in R^{K\T}$ then the height of $v$ must be at least $\arity(R)-1$.
    Since $K\T$ is a pp-tree over $\tilde{\sigma}$, we omit its basepoint in notation.
\end{definition}
Intuitively, we lose no information when moving from $\T$ to $K\T$ since, for any pp-tree $\T$, knowing the arity of a relation $R$ and the last element of tuples in $R^\T$ determines those tuples uniquely.

\begin{remark}
    Given $\sigma$ and $\tilde{\sigma}$ as above, since the \logic comonad is defined uniformly over all signatures containing $E$, we also have a $k$-indexed family of \logic comonads on $\Structpointed[\tilde{\sigma}]$. Since the \logic comonad on a unimodal signature coincides with the Modal comonad with a unique modality, we use the notation $\Mk$ for this latter comonad.

    Note that $\Mk$-coalgebras are the pp-trees of height at most $k$ over $\tilde{\sigma}$, or, equivalently, the rooted Kripke trees of height at most $k$. We retain the notation $\ncat{ppTree}$ for the category of pp-trees over the original signature $\sigma$, and denote the category of pp-trees over $\tilde{\sigma}$, i.e.\ of rooted Kripke trees, by $\ncat{KripkeTree}$.
\end{remark}

\begin{proposition}\label{prop:translation_functor}
    $K$ as given above defines the action on objects of a functor $K: \ncat{ppTree} \to \ncat{KripkeTree}$, which acts as the identity on morphisms.
    
    Moreover this functor is fully faithful, and its image%
    \footnote{As was the case for the translation functor $t: \ModDGL_* \to \Stpgl$, the image of the functor $K$ turns out to be closed under isomorphisms and hence coincides with the notion of essential image of the functor.}
    is the full subcategory of $\ncat{KripkeTree}$ spanned by the rooted Kripke trees $\T'$ that satisfy the following condition:
    \begin{itemize}
        \item[{\em ($\star$)}]
     for all $R \in \bsigma$ and for all $v' \in |\T'|$, if $v' \in R^{\T'}$ then the height of $v'$ is at least $\arity(R)-1$.
 \end{itemize}
\end{proposition}
\begin{proof}
    The claim that $K$ extends to a functor acting as the identity on morphisms reduces to the claim that given a function $f: |\T| \to |\T'|$, if $f$ constitutes a morphism $\T \to \T'$ of pp-trees over $\sigma$, then it also constitutes a morphism $K\T \to K\T'$ of pp-trees over $\tilde{\sigma}$. On the other hand, checking that $K$ is full reduces to checking the converse implication. This is immediate from the definition of $K$. Meanwhile, a functor acting as the identity on morphisms is automatically faithful.

    Since $K$ is fully faithful, its image is a full subcategory of $\ncat{KripkeTree}$, and by definition of the interpretation $R^{K\T}$ of symbols $R \in \bsigma$ on structures of the form $K\T$ for some $\+T$, it is clear that all objects in the image of $K$ satisfy condition $(\star)$. 
    Conversely, any Kripke tree $\+T'$ satisfying condition $(\star)$ is the image of the pp-tree $\+T$ defined by $|\+T| \coloneqq |\+T'|$ and for each $R \in \bsigma$ of arity $r$, $R^{\+T} \coloneqq \{s \in |\+T'|^r \midd \text{$s$ is an $E$-chain and } \epsilon(s) \in R^{\+T'}\}$.
\end{proof}

Intuitively, this means that for a fixed $\sigma$ we can identify pp-trees with the Kripke trees where the truth value of propositional variables cannot be true too close to the root, and where how close is too close is controlled by the arities of propositional variables when seen as symbols in $\sigma$.

\begin{proposition}\label{prop:K_preserves_and_reflects_opes}
    $K$ preserves and reflects open pathwise embeddings.
\end{proposition}
\begin{proof}
    Since open pathwise embeddings in both the domain and codomain categories are bounded morphisms, all we must show is that a function $f : |\T| \to |\T'|$ constitutes a bounded morphism $\T \to \T'$ iff it constitutes a bounded morphism $K\T \to K\T'$. The back condition in both cases is exactly the same, while the mutual implication between the harmony conditions amounts to a straightforward verification.
\end{proof}

In the following theorem, we use the symbol $\bisim_k$ to refer to both $k$-bisimilarity between $\sigma$-structures and $k$-bisimilarity between $\tilde{\sigma}$-structures. Recall that since $\tilde{\sigma}$ is unimodal, $\tilde{\sigma}$-structures are Kripke models and the relation $\bisim_k$ between them is precisely $k$-bisimilarity in \bml \cite[Section 10.3]{abramsky2021relating}.

\begin{theorem}\label{thm:K_preserves_bisimilarity}
    Given two pointed $\sigma$-structures $(\A, a)$ and $(\B,b)$, $(\A,a) \bisim_k (\B,b)$ if and only if $K\Ck(\A,a) \bisim_k K\Ck(\B,b)$.
\end{theorem}
\begin{proof}
Given two pp-trees, we write a decorated arrow $\xrightarrow{\text{ope}}$ to indicate the existence of an open pathwise embedding between them. Let $(\A, a)$ and $(\B,b)$ be as above. Then
\begin{align*}
(\A, a) \bisim_k (\B, b) &\text{ iff } \exists \T \in \EM(\Ck). \, \Ck(\A,a) \xleftarrow{\text{ope}} \T \xrightarrow{\text{ope}} \Ck(\B,b) &\text{(Thm.~\ref{thm:full_logical_equivalence})}&\\
    &\text{ iff } \exists \+T \in \EM(\Ck). \,  K\Ck(\A,a) \xleftarrow{\text{ope}} K\T \xrightarrow{\text{ope}} K\Ck(\B,b) &\text{($1$) (Prop.~\ref{prop:K_preserves_and_reflects_opes})}&.
\end{align*}
On the other hand, we know that $K\Ck(\A, a) \bisim_k K\Ck(\B, b)$ if and only if
\begin{align*}
    &\exists \+T' \in \EM(\Mk). \, \Mk K\Ck(\A,a) \xleftarrow{\text{ope}} \T' \xrightarrow{\text{ope}} \Mk K\Ck(\B,b) \\
    \text{iff } &\exists \+T' \in \EM(\Mk). \, K\Ck(\A,a) \xleftarrow{\text{ope}} \T' \xrightarrow{\text{ope}} K\Ck(\B,b) \tag{2} 
\end{align*}
where we have used that $\Mk$ is idempotent, which implies that the coalgebra maps $K\Ck(\A,a) \to \Mk K\Ck(\A,a)$ and $K\Ck(\B,b) \to \Mk K\Ck(\B,b)$ are isomorphisms and hence, in particular, open pathwise embeddings. Hence we must show the bi-implication $(1) \iff (2)$. The rightward implication is immediate, while for the leftward implication it is enough to show that given $\T'$ as in $(2)$, $\T'$ is in the image of $K$. Indeed, let $f: \T' \to K\Ck\Ap$ be any morphism (not necessarily bounded), let $v \in |\T'|$ and $R \in \bsigma$, and suppose that the height of $v$ is $h < \arity(R) - 1$. Then by Prop.~\ref{prop:morphisms_of_pptrees_preserve_height}, the height of $f(v)$ is also $h$, hence $f(v) \not\in R^{K\Ck\Ap}$, hence $v \not\in R^{\T'}$. Thus $\T'$ satisfies condition $(\star)$ and hence by Prop.~\ref{prop:translation_functor} we conclude that $\T'$ is in the image of $K$.
\end{proof}

We now discuss the relationship between the functor $K$ and existing notions of transformations between comonads and their EM categories.

Given $k\geq 0$, since $K$ preserves the height of pp-trees, it restricts and corestricts to an operation $\EM(\Ck) \to \EM(\Mk)$. More generally, one could wonder whether this operation extends to all $\sigma$-structures, turning them into related Kripke structures. This line of reasoning reverses the one presented in~\cite{jakl2023categorical}, where the authors develop a general and systematic approach to studying operations on structures and whether these operations lift to functors between EM categories which preserve open pathwise embeddings. Operations admitting such liftings are shown to enjoy Feferman-Vaught-Mostowski-style or `FVM' compositionality theorems for the logic(s) corresponding to the comonads involved. An interesting application of their formalism is the fact that all logics admitting a comonadic characterisation enjoy an FVM theorem for categorical products, which in the case of \logic readily implies the following result.
\begin{proposition}\label{prop:fvm}
    Let $\approx_k$ be any of $\Rrightarrow^+_k, \equiv^+_k, \equiv_k$ or $\equiv_k^\#$. Then given finitely-branching $\sigma$-structures $(\A_1, a_1)$, $(\A_2, a_2)$, $(\B_1, b_1)$, and $(\B_2, b_2)$, we have that
\begin{align*}
    &(\A_1, a_1) \approx_k (\A_2, a_2) \text{ and } (\B_1, b_1) \approx_k (\B_2, b_2) \\
    &\text{ implies }(A_1, a_1) \times (B_1, b_1) \approx_k (A_2, a_2) \times (B_2, b_2).
\end{align*}
\end{proposition}
\begin{proof}
    Apply Prop. VI.1, Prop VI.3, and Thm. VI.4 in~\cite{jakl2023categorical}.
\end{proof}

The question at hand, then, is whether our functor $K: \EM(\Ck) \to \EM(\Mk)$ arises from a more general unary operation $\underline{K}: \Structpointed \to \Structpointed[\tilde{\sigma}]$ which plays a role analogous to $\times$ in the Corollary above.

Given comonads $\bb C$ and $\bb C$ on categories $\cat{E}$ and $\cat{E}'$, respectively, there is a standard notion of comonad morphism $\bb C \Rightarrow \bb D$ which generalises the one given in Definition \ref{def:comonad_morphism}, consisting of a functor $F: \cat{E} \to \cat{E}'$ together with a natural transformation $\kappa: \bb DF \Rightarrow D \bb C$ such that $F\epsilon \circ \kappa = \epsilon F$ and $F\delta \circ \kappa = \kappa \bb C \circ \bb D \kappa \circ \delta F$ (see~\cite{street1972formal} for the dual notion for monads, where it is referred to as a \emph{monad functor}). In~\cite{jakl2023categorical}, such natural transformations $\kappa$ are referred to \define{Kleisli laws} for the functor $F$. Kleisli laws for a functor $F$ and comonads $\bb C$ and $\bb D$ as above are in one-to-one correspondence with liftings of $F$ to the corresponding Kleisli categories, i.e. functors $\overline{F}: \Kl(\bb C) \to \Kl(\bb D)$ such that $\overline{F} \circ F^{\bb C} \cong F^{\bb D} \circ F$, where $F^{\bb C}$ and $F^{\bb D}$ are the right adjoints of the Kleisli adjunctions of $\bb C$ and $\bb D$, respectively~\cite{jakl2023categorical}. Since $\Ck$ and $\Mk$ are idempotent, their Kleisli categories are equivalent to their EM categories (see Remark \ref{rem:isomorphisms_where}), hence Kleisli laws also classify liftings to the EM categories.

In our case, it is not hard to see that $K$ arises from a comonad morphism $(\underline{K}, \kappa): \Structpointed \to \Structpointed[\tilde{\sigma}]$ where $\underline{K}(\Ap)$ is defined to have universe $|\A|$, basepoint $\bpa$, and interpretations
$$R^{\underline{K}\Ap} \coloneqq \{a' \in |\A| \midd \text{there is an $E$-chain $s$ from $\bpa$ to $a'$ such that $\A, s \models R$}\}.$$
The Kleisli law $\kappa: \Mk \underline{K} \Rightarrow \underline{K} \Ck$ in this case is the identity, i.e.\ we have an equality of functors $\Mk \underline{K} = \underline{K} \Ck$.

\subsubsection{Polynomial Reductions Using the Translation Functor}

We now use the functor $K$ to give computational reductions from \logic problems to their \bml analogues. Although the complexity results thus obtained may also be established directly, the reductions exhibit the close relationship between \logic and \bml.

\paragraph*{Deciding $k$-bisimilarity.} Given a finite signature $\sigma$ with $E \in \sigma$, the problem \kbisimproblemsig has as inputs two finite, pointed $\sigma$-structures $(\A,a)$ and $(\B,b)$, and asks whether $(\A,a) \bisim_k (\B,b)$.
Note that when $\sigma$ is unimodal, \kbisimproblemsig consists in checking whether two Kripke models are $k$-bisimilar in the usual sense of \bml.

\begin{corollary}\label{coro:reduction_bisimilarity}
    There is a polynomial-time reduction from \kbisimproblemsig to \kbisimproblemsig[\tilde{\sigma}]. Thus \kbisimproblemsig is in \ptime.
\end{corollary}
\begin{proof}
    Given $(\A,a)$ and $(\B,b)$, the reduction simply computes $K\Ck(\A,a)$ and $K\Ck(\B,b)$, since by Theorem~\ref{thm:K_preserves_bisimilarity}, $(\B,b)$, $(\A,a) \bisim_k (\B,b)$ iff $K\Ck(\A,a) \bisim_k K\Ck(\B,b)$. Computing the action of $\Ck$ on finite structures is polynomial in the size of the structure
    as can be seen from inspection of Definition~\ref{def:ppml_comonad}. On the other hand, $K$ can be computed in linear time as its action can be calculated with just one pass over the input data (for each $R \in \bsigma$ and for each tuple $s \in R^{\Ck\A}$, write $\epsilon(s) \in R^{K\Ck\A}$).
\end{proof}

We now turn to the issue of truth preservation, related to giving reductions to \bml for the problems of model checking and satisfiability.

\begin{remark}
    The syntax of \logic for a given signature is independent of the arity of the relation symbols. Hence a $\sigma$-\logic-formula can always be regarded as a $\tilde{\sigma}$-\bml-formula, and viceversa.
\end{remark}

\paragraph*{Model checking.} Given a signature $\sigma$ with $E \in \sigma$, the problem \modelcheckproblemsig has as inputs a finite structure $\Ap \in \Structpointed$ and a \logic-formula $\phi$, and asks whether $\Ap \models \phi$.

\begin{theorem}\label{thm:translation_preserves_truth}
    Given $\T \in \ncat{ppTree}$, a \logic formula $\phi$ and $v \in |\T|$, $\T, \T_v \models \phi$ if and only if $K\T, v \models \phi$. In particular, when considering single-point semantics at the root of $\T$, $\T \models \phi$ if and only if $K\T \models \phi$.
\end{theorem}
\begin{proof}
    We proceed by structural induction on $\phi$. The non-trivial cases are those of $\phi = R$ for some $R \in \bsigma$, which holds precisely by definition of $R^{K\T}$, and of $\phi = \Diamond \psi$, for which we reason as follows:
    \begin{align*}
    \T, \T_v \models \phi &\text{ iff } \exists w. v \prec w \text{ and } \T, \T_v.w \models \psi
    \\
    &\text{ iff } \exists w. v \prec w \text{ and } \T, \T_w \models \psi \\
    &\text{ iff } \exists w. v \prec w \text{ and } K\T, w \models \psi \tag{inductive hypothesis}\\
    &\text{ iff } K\T, v \models {\phi}.
    \end{align*}
This concludes the proof.
\end{proof}

\begin{corollary}\label{coro:model_checking}
    There is a polynomial-time reduction from \modelcheckproblemsig to the problem \modelcheckproblemsig[\tilde{\sigma}].
\end{corollary}
\begin{proof}
    Just as for \kbisimproblemsig, the reduction amounts to computing $K\Ck\Ap$ and checking whether $K\Ck\Ap \models \phi$.
\end{proof}

\paragraph*{Satisfiability.} The problem \satproblem has as input a \logic formula $\phi$, and asks whether $\phi$ is satisfiable over $\sigma$-structures where $\sigma$ is the finite signature consisting of the symbols appearing in $\phi$ together with their specified arities%
\footnote{We consider the information of $\arity(R)$ for each relation symbol $R$ appearing in $\phi$ to be codified in unary. This is an intuitive requirement if we recall that in First Order Logic, the arity of relation symbols appearing in a given formula is explicitly codified in unary through the variables appearing in the relational atom. In contrast, the arity of symbols is not reflected at all in the syntax of a \logic formula.}
and the binary symbol $\acc$. The satisfiability problem for \bml, \satproblembml, is defined in the same way except that all relation symbols are presumed to be unary.

\begin{theorem}\label{thm:satisfiability_reduction}
    There exists a polynomial-time reduction from \satproblem to
    \satproblembml.
    Moreover, since the former includes the latter, we deduce that \satproblem is \pspace-complete.
\end{theorem}
\begin{proof}
    Let $\phi$ be a \logic formula, and let $\bsigma$ consist of the symbols in $\phi$ with their prespecified arities, so that $\sigma = \{\acc\} \union \bsigma$ is the \logic signature obtained from the input of \satproblem.

    By Theorem~\ref{thm:translation_preserves_truth}, if $\phi$ is \logic-satisfiable then it is \bml-satisfiable. However the converse does not hold since some \bml-satisfiable formulas such as $\psi = R$, for $R \in \bsigma$ of arity $2$, are badly nested as \logic-formulas and hence \logic-unsatisfiable. This can be seen as a consequence of the fact that $K$ is not essentially surjective on objects, and as a consequence of the fact that the syntactic redundancy of badly-nested formulas is not mirrored in \bml.

    To solve this problem, notice that, although the class of Kripke trees is not definable internal to \bml, the image of $K$ is \bml-definable with respect to class of Kripke trees: a Kripke tree is in the image of $K$ if and only if it satisfies the formula
    $$\phi_K \coloneqq \bigwedge_{\substack{R \in \bsigma \\ \arity_\sigma(R) > 1}} \bigwedge_{j=1}^{\arity_\sigma(R)-2} \lnot (\underbrace{\Diamond \dots \Diamond}_{\text{$j$ times}} R).$$
    Moreover, $\phi_K$ is a tautology of \logic.
    Hence, given $\phi$, set $\phi' \coloneqq \phi \land \phi_K$. We aim to show that the mapping $\phi \mapsto \phi'$ is our desired reduction.

    Since $\phi_K$ is a \logic-tautology, $\phi'$ is \logic-equivalent to $\phi$. Hence $\phi$ is \logic-satisfiable iff $\phi'$ is \logic-satisfiable. By Corollary~\ref{coro:coalgebra-model_property}, we conclude that $\phi$ is \logic-satisfiable iff $\phi'$ is \logic-satisfiable on the class of pp-trees of finite height.
    
    On the other hand, by the tree-model property of \bml, $\phi'$ is \bml-satisfiable iff it is satisfiable in the class of finite Kripke trees, and since we know it cannot hold in any Kripke tree outside the image of $K$, we conclude that $\phi'$ is \bml-satisfiable iff it is \bml-satisfiable in the image of $K$.
    
    Finally, using Theorem~\ref{thm:translation_preserves_truth} we connect both chains of equivalences and conclude that $\phi$ is \logic-satisfiable if and only if $\phi'$ is \bml-satisfiable.

    Note the length of $\phi'$ is polynomial in the length of $\phi$ and in the prespecified arities of the relation symbols.

    Since $\phi'$ can be computed from $\phi$ in polynomial time, this gives a polynomial reduction from \satproblem to \satproblembml, which means that \satproblem is in \pspace. Finally, since \satproblem includes \satproblembml for certain choices of input, and since \satproblembml is \pspace-complete~\cite{blackburn2001modal}, \satproblem is \pspace-complete as well.
\end{proof}

From the proof above we may also conclude that \logic inherits the finite-model property from \bml.
\begin{corollary}
    \logic has the finite-model property: a \logic-formula is satisfiable if and only if it is satisfied by a finite structure.
\end{corollary}
\begin{proof}
    Let $\phi$ by a \logic formula and suppose that it is \logic-satisfiable. Then, from the proof of Theorem~\ref{thm:satisfiability_reduction} we know that $\phi' \coloneqq \phi \land \phi_K$, as defined in the proof, is \bml-satisfiable. But a \bml formula is satisfiable if and only if it is satisfied by a finite Kripke tree. Let $\T'$ be a finite Kripke tree such that $\T' \models \phi'$. Then since $\T' \models \phi_K$, $\T'$ is in the image of $K$, hence there exists a pp-tree $\T$ with $|\T| = |\T'|$ such that $K\T = \T'$. In particular, $\T$ is finite, and by Theorem~\ref{thm:translation_preserves_truth}, $\T \models \phi'$, hence $\T \models \phi$.
\end{proof}

We close with a note on the expressivity of \logic. Since \logic shares many of the complexity properties of \bml (and, we may add, since the \logic comonad shares with the \bml comonad the `tameness' property of idempotence, which in particular implies a tree-model property), it makes sense to expect that these two logics might somehow be also equivalent in expressive power. Although we cannot make a direct comparison, since in the general case \logic and \bml predicate over different classes of models, the translation functor $K$ allows us to establish a connection between the logical types for these two languages.

\begin{definition}
    Given $k \geq 0$ and a pointed $\sigma$-structure $\Ap$, let $[\A, \bpa]_k^\sigma$ denote the type of $\Ap$ with respect to \logick over the signature $\sigma$, i.e.\ its equivalence class with respect to $\equiv_k$. Let $\ncat{Type}_k(\sigma)$ be the set of types of $\sigma$-structures with respect to \logick and $\ncat{Type}_k(\tilde{\sigma})$ the set of types of Kripke structures (i.e.\ $\tilde{\sigma}$-structures) with respect to $\bml_k$ (i.e.\ \logick over $\tilde{\sigma}$).
\end{definition}

\begin{proposition}
    Assuming finite $\sigma$, the assignment
    $$[\A, \bpa]_k^\sigma \mapsto [K\Ck\Ap]_k^{\tilde{\sigma}}$$
    determines a well-defined function $\overline{K}: \ncat{Type}_k(\sigma) \to \ncat{Type}_k(\tilde{\sigma})$ which is injective and whose image consists of all classes of the form $[\B,\bpb]_k^{\tilde{\sigma}}$ where $\Bp$ is a Kripke structure satisfying the condition
    \begin{itemize}
        \item[{\em ($\star'$)}] for all $R \in \bsigma$ and for all $b' \in |\B|$, if $b' \in R^\B$ then there exists an $E$-chain $s$ from $b$ to $b'$ such that $|s| \geq \arity(R)$.
    \end{itemize}
\end{proposition}
\begin{proof}
    Since $\sigma$ is finite, given $\Ap, \Bp \in \Stp$ we have that $[\A, \bpa]_k^\sigma = [\B, \bpb]_k^\sigma$ if and only if $\Ap \bisim_k \Bp$, and analogously for $\tilde{\sigma}$. Thus $\overline{K}$ is well defined and injective by Theorem~\ref{thm:K_preserves_bisimilarity}.

    As for the image of $\overline{K}$, clearly every class in the image of $\overline{K}$ contains a Kripke structure of the form $K\Ck\Ap$ for some $\Ap \in \Stp$ which in particular satisfies $(\star')$. Conversely, suppose $\Bp \in \Structpointed[\tilde{\sigma}]$ satisfies condition $(\star')$. Notice that $[\B,\bpb]_k^{\tilde{\sigma}} = [\Mk\Bp]_k^{\tilde{\sigma}}$ and, since $\Bp$ satisfies $(\star')$, $\Mk\Bp$ is a Kripke tree of height at most $k$ satisfying condition $(\star)$ in Prop.~\ref{prop:translation_functor}, hence there exists some pp-tree $\T$ of height at most $k$ over $\sigma$ such that $K\T = \Mk\Bp$. Finally, since $\T \cong \Ck\T$, we may assume without loss of generality that $\T = \Ck\Ap$ for some $\Ap \in \Stp$. Hence $[\B,\bpb]_k^{\tilde{\sigma}} = [\Mk\Bp]_k^{\tilde{\sigma}} = [K\Ck\Ap]_k^{\tilde{\sigma}}$ is in the image of $\overline{K}$.
\end{proof}

Since $\overline{K}$ embeds \logick types as a subset of $\bml_k$ types, we may say that what \logick sees of $\sigma$-structures is `the same' as what $\bml_k$ sees of $\tilde{\sigma}$-structures, except for the fact that some $\bml_k$ types do not correspond to any \logick type. For instance, if $S \in \sigma$ is binary, for the singleton Kripke structure $\{*\}$ with $S^{\{*\}} \coloneqq \{*\}$ there does not seem to be any natural choice of element in $\ncat{Type}_k(\sigma)$ to which it may be assigned.

We conclude that there are two distinct yet complementary perspectives on the relationship between \logic and \bml in terms of expressivity. On one hand, \logic predicates over a larger class of models, by allowing non-unimodal choices of signature $\sigma$. On the other hand, once the choice of $\sigma$ is fixed, the chosen arities restrict the properties that \logic may express about $\tilde{\sigma}$-structures.

\section{Conclusions and Future Work}\label{sec:conclu}

\logicfullname, or \logic, is a generalisation of Basic Modal Logic which arose from the insight that the `equal data' relationship in \datagl can be split from the definition of the modal operators and added as an atom of the language. This allows us to express \datagl not as a bimodal logic but as a unimodal logic of a new kind, and induces a natural representation of the \datagl bisimulation game as a particular case of the \logic bisimulation game.

What do we gain? We give two main motivations for \logic. The first one is that we can think of \logic as a way of interpreting Basic Modal Logic over general relational structures.
Given a first order signature $\sigma$ that contains at least one binary relation, we can select it to function as an accessibility relation and then use \logic to reason modally about the structure, replacing $\exists$ and $\forall$ with $\Diamond$ and $\Box$. This represents an important relaxation on what kinds of first order signatures admit a modal interpretation. The standard translation into First Order Logic shows how this amounts to a restriction on which bound variables can appear inside first order atoms, and in what order they appear.

Our second motivation for \logic is to present a framework for capturing and designing simple data-aware logics which is different from (multi-)modal logic, while at the same time retaining a modal-like syntax and semantics. This is seen in the case of \datagl splitting data-aware modalities into two separate syntactic constructs. More generally, this work constitutes an exploration of comonadic semantics as a framework for studying data-aware logics.

We emphasise that the comonadic formalism allowed a systematic study of multiple properties of interest of a new logic. This perspective on \logic takes as a fundamental starting point the notion of pp-tree and the $k$-step unravelling construction. From there, multiple lines of thinking open up naturally, such as the characterisation of the expressivity of graded modalities through isomorphism of unravellings, a homomorphism-counting property with respect to the class of finite pp-trees, the pp-tree model property, or the Chandra-Merlin-like correspondence between finite pp-trees and positive \logic formulas. We even mention in passing a FVM-type theorem for products of structures obtained `by free' from the comonadic formalism (Prop.~\ref{prop:fvm}). These are either applications of general results for game comonads or follow analogous results previously established in the literature, although sometimes requiring non-trivial adaptations. Moreover, using a translation technique we obtain a characterisation of the expressivity of \datagl, and still taking as fundamental the notion of unravelling and the ensuing comonads, both for \logic and for \bml, we obtain polynomial-time computational reductions from \logic to \bml. The fundamental algorithmic observation in this context is that unravellings of finite structures are themselves computable in polynomial time for fixed values of $k$.

We close with a discussion of some lines for future work. Of course, we can continue to apply comonadic techniques to the study of \logic, such as looking into homomorphism preservation properties~\cite{abramsky2022preservation} or FVM-type properties as explored in~\cite{jakl2023categorical}. Here we sketch two lines of research inspired by our two main motivations: developing modal languages for general relational structures, and giving comonadic semantics to data-aware logics.

\paragraph*{Multimodal, Polyadic \boldmath{\logic}.}

We have presented the theory of \logic as corresponding to Basic Modal Logic, but modal languages may be constructed more generally by choosing a \emphat{modal similarity type} consisting of a finite number of modal operators which moreover may be polyadic, i.e.\ correspond to accessibility relations of arbitrary finite arity~\cite{blackburn2001modal}. Allowing any subset of $\sigma$ to be interpreted as accessibility relations for modal exploration is particularly interesting from the motivation of extending the modal lens to arbitrary relational structures. Moreover, in this more general case we may ask what happens when any relation symbol can be used both as accessibility relation for a modality and as an atom at the same time. In \logic allowing $E$ as an atom of the language would not add expressive power at all, since e.g.\ $\Diamond E$ would be equivalent to $\Diamond \top$. However, this is no longer true if we introduce polyadic modalities.

\paragraph*{Comonadic semantics for \boldmath{\cdxp}.}
As we begin to explore more complex comparison games, such as bisimulation games for \cdxp~\cite{FFA15}, we expect that these games will be captured by a comonad together with a translation technique similar to our treatment of \datagl. In the case of \datagl, we translated \datagl models into relational structures and then applied the comonad $\Ck$ corresponding to \logic. For other fragments of \cdxp the comonad in itself might not correspond to an easily recognizable logic.

As a next step in this direction, recall from Section~\ref{sec:ppml_as_framework} the fragment $\cdxp_\epsilon(\downarrow^+)$ of $\cdxp(\downarrow^+)$. As we noted, although it is possible to translate this fragment into $\sigmacdxp$-\logic, this latter logic is strictly more expressive even when restricted to an appropriate class of models. This rules out the possibility of using the \logic comonad $\Ck$ to capture indistinguishability for $\cdxp_\epsilon(\downarrow^+)$. Instead, we may obtain a comonadic characterisation of $\cdxp_\epsilon(\downarrow^+)$ by recurring to a new comonad.

When we interpret $\cdxp_\epsilon(\downarrow^+)$ over tree-shaped models, the $k$-(bi)simulation game for this logic can be stated in terms of a certain two-pebble game, where the two pebbles must be moved in alternation. Here the parameter $k$, which on the side of games corresponds to the maximum number of rounds, represents the number of nested occurrences of $\dow$ in a given formula, including those inside tests.
To obtain a comonad encoding this game, start with the signature $\sigmadgl$ and the full subcategory $\ModDGL_*^\text{tree}$ of $\ModDGL_*$ spanned by data trees. Then for each $k \geq 0$ we can construct an extension of $\sigmadgl$, $\sigma_k \coloneqq \sigmadgl \union \{E_0,\dots,E_k\}$, where all the new symbols are binary, and a functor $t_k: \ModDGL_*^\text{tree} \to \Structpointed[\sigma_k]$ extending the functor from Def.~\ref{def:functor_from_datagl_models} with $E_j^{t\MMp} \coloneqq \bigcup_{i=0}^j (E^{t\MMp})^i$. In this context, for each $k\geq 0$ there exists a comonad with underlying functor $\bb X_k : \Structpointed[\sigma_k] \to \Structpointed[\sigma_k]$ such that winning strategies for Duplicator in the one-way simulation game between the data trees $\MMp, \MMpp \in \ModDGL_*^\text{tree}$ correspond to Kleisli morphisms $\bb X_k t_k \MMp \to t_k \MMpp$.

It is not clear whether it is possible to give a language for $\bb X_k$ in such a way that $\bb X_k$ becomes `its' game comonad. Further study is needed in order to assess such possibility, for this and other fragments of \cdxp. In this way it might be possible to obtain new logics closely related to \cdxp in the same way in which \logic relates to \datagl, by extending already known languages into unknown territory.

On the other hand, since we expect that the translation technique will continue to be necessary for data-aware logics, this limits some of the benefits arising from the existence of a related comonad, e.g.\ we cannot directly interpret coalgebras of these comonads as reifications of positive formulas in our language of interest. This may motivate the development of a proper theory of relative game comonads, by e.g.\ characterising the logical meaning of the relative coalgebras of a relative comonad \cite[Def. 2.11]{altenkirch2010monads}.

\bigskip

\paragraph{Acknowledgements:}
This work was partially funded by UBACyT 20020190100021BA and PICT-2021-I-A-00838. We thank Tomáš Jakl for helpful conversations on the topic of relative comonads.

\bibliographystyle{plain}
\bibliography{PPML_JLC_revised.bib}

\end{document}